\documentclass[aos,preprint]{imsart}

\RequirePackage{amsthm,amsmath,amsfonts,amssymb,mathtools,bm,bbm,mathabx,tikz}
\RequirePackage[authoryear]{natbib}
\RequirePackage[colorlinks,citecolor=blue,urlcolor=blue]{hyperref}
\usepackage[nameinlink]{cleveref}
\usepackage[boxed, ruled, lined, linesnumbered]{algorithm2e}
\usepackage{listings}
\let\tablewidth\relax
\usepackage{tabularray}
\UseTblrLibrary{diagbox}
\usepackage{pgfplots}
\pgfplotsset{compat=1.18}
\usepackage{fontenc}
\usepackage{anyfontsize}
\RequirePackage{graphicx}

\startlocaldefs
\numberwithin{equation}{section}
\theoremstyle{plain}
\newtheorem{thm}{Theorem}
\numberwithin{thm}{section}

\newtheorem{lem}{Lemma}
\numberwithin{lem}{section}

\newtheorem{prop}{Proposition}
\numberwithin{prop}{section}

\newtheorem{cor}{Corollary}
\numberwithin{cor}{section}

\newtheorem{claim}{Claim}
\numberwithin{claim}{section}

\theoremstyle{definition}
\newtheorem{ass}{Assumption}

\newtheorem{defn}{Definition}
\numberwithin{defn}{section}

\newtheorem{rem}{Remark}
\numberwithin{rem}{section}

\newtheorem{example}{Example}

\newenvironment{ass*}
 {\expandafter\def\expandafter\theass\expandafter{\theass*}\ass}
 {\endass}
\newenvironment{prop*}
 {\expandafter\def\expandafter\theprop\expandafter{\theprop*}\prop}
 {\endprop}

\Crefname{ass}{Assumption}{Assumptions}
\Crefname{prop}{Proposition}{Propositions}
\Crefname{section}{Section}{Sections}
\Crefname{appendix}{Appendix}{Appendices}
\Crefname{cor}{Corollary}{Corollaries}
\Crefname{lem}{Lemma}{Lemmas}
\Crefname{rem}{Remark}{Remarks}
\Crefname{example}{Example}{Examples}
\Crefname{thm}{Theorem}{Theorems}

\newcommand{\R}{\mathbb{R}}

\DeclareMathOperator*{\argmin}{arg min}

\DeclareMathOperator{\tr}{tr}

\DeclareMathOperator{\blkdiag}{blkdiag}

\DeclareMathOperator{\sign}{sign}
\DeclareMathOperator{\supp}{supp}
\DeclareMathOperator{\range}{range}

\DeclareMathOperator{\diag}{diag}

\newcommand{\rank}{\operatorname{rk}} 
\endlocaldefs

\begin{document}

\begin{frontmatter}
\title{Joint Learning of Panel VAR models with Low Rank and Sparse Structure}
\runtitle{Low-Rank and Sparse Panel VAR Models}

\begin{aug}
\author[A]{\fnms{Yuchen}~\snm{Xu}\ead[label=e1]{yuchenxu95@g.ucla.edu}}
\and
\author[A]{\fnms{George}~\snm{Michailidis}\ead[label=e2]{gmichail@g.ucla.edu}}
\address[A]{University of California, Los Angeles\printead[presep={,\ }]{e1,e2}}
\end{aug}

\begin{abstract}
Panel vector auto-regressive (VAR) models are widely used to capture the dynamics of multivariate time series across different subpopulations, where each subpopulation shares a common set of variables. In this work, we propose a panel VAR model with a shared low-rank structure, modulated by subpopulation-specific weights, and complemented by idiosyncratic sparse components. To ensure parameter identifiability, we impose structural constraints that lead to a nonsmooth, nonconvex optimization problem. We develop a multi-block Alternating Direction Method of Multipliers (ADMM) algorithm for parameter estimation and establish its convergence under mild regularity conditions. Furthermore, we derive consistency guarantees for the proposed estimators under high-dimensional scaling. The effectiveness of the proposed modeling framework and estimators is demonstrated through experiments on both synthetic data and a real-world neuroscience data set.
\end{abstract}


\begin{keyword}
\kwd{ADMM}
\kwd{joint model}
\kwd{low rank matrix}
\kwd{panel time series}
\kwd{sparse matrix}
\kwd{vector auto-regression}
\end{keyword}

\end{frontmatter}
\tableofcontents


\section{Introduction}
\label{sec:intro}

Vector autoregressions represent a popular modeling framework for capturing dynamic relationships between multivariate time series and have been extensively used in macroeconomics \citep{kilian2017structural}, functional genomics \citep{michailidis2013autoregressive} and neuroscience/neuroimaging \citep{seth2015granger,aslan2024granger} applications. Their use in applications wherein the number of time series under consideration is large relative to the number of time points available, led to the introduction of regularized estimators for the VAR model parameters --e.g., assuming sparsity \citep{basu2015,melnyk2016estimating,kock2015oracle,kastner2020,medeiros2016ℓ1}, group sparsity \citep{billio2019bayesian,ghosh2019high}, or low rankness \citep{basu2019,wang2024}- and the investigation of their consistency properties under high dimensional scaling, including statistical inference \citep{krampe2023structural}. 

In many of these application domains, it has become increasingly common to consider multivariate time series data from a \textit{collection of entities}. For example, in macroeconomic applications, the entities may represent different countries, while in neuroimaging studies, they may correspond to individual subjects. The resulting data consist of $p$ time series (variables), observed over $T$ time periods, across $M$ entities. A key challenge is to develop modeling strategies that capture both \textit{common effects} shared across entities, as well as \textit{entity-specific} heterogeneities, while simultaneously controlling the total number of parameters to ensure computational and statistical efficiency. 

To effectively capture both shared structures and entity-specific variations in panel VAR models, several modeling strategies have been proposed in the literature. One approach involves imposing constraints on the entity-specific VAR parameters through structural assumptions \citep{canova2013}, penalization techniques \citep{skripnikov2019}, or prior distributions \citep{korobilis2016}; further details on this line of work are provided in \autoref{sec:dis} in the Supplement. An alternative strategy is to represent the data as a time series of matrix-variate variables \citep{chen2021} or as a three-dimensional tensor \citep{chen2022a}, enabling the use of factor models or tensor decomposition techniques. A brief overview of these approaches is also provided in \autoref{sec:dis} in the Supplement.

In this paper, we propose a \text{panel VAR}-based modeling strategy that captures both shared structures across entity-specific VAR models and idiosyncratic heterogeneities. Specifically, we consider multivariate time series data comprising $p$ variables observed over $T$ time periods for $M$ entities, denoted as $\{X_t^m \in\mathbb{R}^p; t=1,\cdots,T, \ m=1,\cdots,M\}$.  For ease of presentation, we assume a common number of time periods $T$ across entities, although the proposed model can be readily extended to unbalanced panels where each entity $m$ has its own number of observations $T_m$. We posit the following model—referred to hereafter as \textbf{LSPVAR} (Low-rank \& Sparse Panel VAR)—for a single lag for ease of exposition:
\begin{equation}\label{eq:pVAR-model}
    X_t^m = A_m X_{t-1}^m +\epsilon_t^m, ~ \epsilon_t^m \sim N(0, \Sigma_m); \qquad A_m=W_m\Phi+S_m,
\end{equation}
where $A_m\in\mathbb{R}^{p\times p}$ is the transition matrix of the VAR model for the $m$-th entity, and $\Sigma_m \in \mathbb{R}^{p \times p}$ is the corresponding covariance structure.
Models with multiple lags can be handled similarly by imposing a comparable structure on the coefficient matrices at higher lags; see the discussion in \autoref{sec:multi_lag} in the Supplement.
The transition matrices $A_m$ are assumed to exhibit a \textit{low rank plus sparse} structure. Specifically, we model each $A_m$ as a rescaled variant of a \textit{common low rank basis} $\Phi$, modulated by a diagonal matrix of entity-specific weights $W_m \in \mathbb{R}^{p \times p}$, along with an \textit{entity-specific, sparse} component $S_m$. The shared low-rank component $\Phi$ efficiently captures global dependencies across variables, while keeping model complexity low, facilitating scalable estimation  in high-dimensional settings. A visual depiction of the transition matrices $A_m$ under the posited model is provided in \autoref{fig:model_setup}.

\begin{figure}[htbp]
    \centering
    \includegraphics[width=\linewidth]{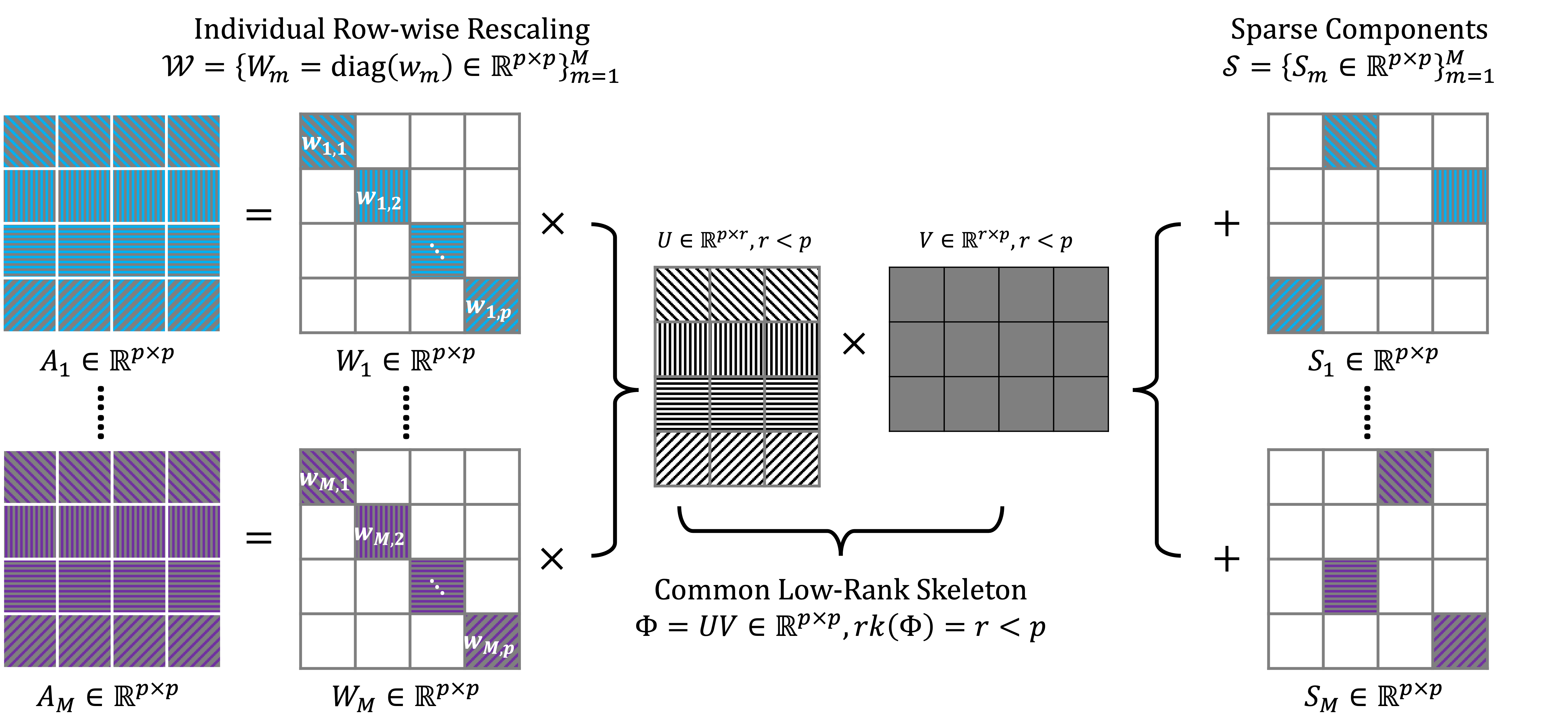}
    \caption{{\footnotesize Illustration of the model setup as specified in \eqref{eq:pVAR-model}. The  lead-lag relationships among variables are encoded in the transition matrices $A_m$. Shared global structures are captured by the common low-rank basis $\Phi$, while entity-specific variations are reflected in the sparse components $S_m$ and the rescaling weights $W_m$. The different fill patterns across the rows of the matrices visually indicate heterogeneity in the dependency structures across entities.}}
    \label{fig:model_setup}
\end{figure}

Next, we demonstrate the versatility of the proposed model through a synthetic data example. We consider time series data generated from a heterogeneous collection of VAR models as briefly described next.

\begin{example}\label{ex:mixture}
    Consider $M=20$ entities, $p=20$ variables and $T=1000$ time points. The entities are organized into distinct clusters as follows: (i) two clusters ($m=1,\cdots,5$ and $m=6,\cdots,10$, respectively) where the corresponding weight matrices $W_m$ share identical diagonal entries; (ii) one cluster ($m=11,\cdots,14$) with $W_m=0$, resulting in a sparse transition matrices $A_m$; (iii) one cluster ($m=15,\cdots,18$) where the weight matrices $W_m$ have identical entries and in addition $S_m=0$ (no entity specific component); and (iv) two additional VAR models ($m=19$ and 20) with transition matrices $A_m$ that differ from all previous groups. Additional details on the data generation process and model parameter specifications are provided in \autoref{subsec:cluster}.
\end{example}

The parameters of the panel VAR model in \autoref{ex:mixture} were estimated using the iterative algorithm described in \autoref{algo:DGP}. To summarize and visualize the results, we applied principal components analysis (PCA) to the collection of the diagonal entries form the estimated $W_m$ matrices -i.e., $\widehat{\mathbf{W}}=\left(\diag(\widehat{W}_1), \cdots, \diag(\widehat{W}_M)\right)$. The resulting PCA plot, shown in \autoref{fig:PC}, clearly illustrates that the proposed model effectively captures the underlying heterogeneity in the data, correctly reflecting both similarities within clusters and entity-specific differences. These results highlight the ability of the model in \eqref{eq:pVAR-model} to flexibly capture diverse patterns—such as latent cluster structures and purely sparse models—making it highly versatile compared to existing panel VAR approaches; see further discussion in \autoref{sec:dis} in the Supplement.
\begin{figure}[htbp]
    \centering
    \includegraphics{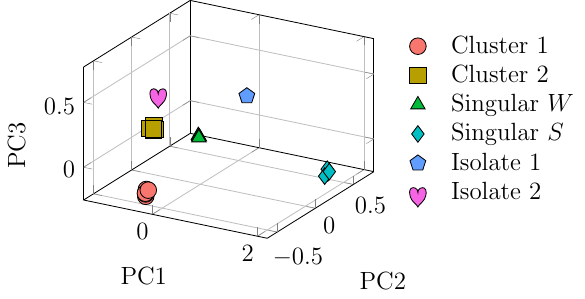}
    \caption{{\footnotesize The top three principal components of the estimated $\mathcal{W}$ of a model consisting of mixed sub-models. Here Singular $W$ (Singular $S$) represents the special case when the $m$-th subject has a purely sparse (low-rank) setting with $W_m = 0$ ($S_m = 0$).}}
    \label{fig:PC}
\end{figure}

However, the structure of the proposed model presents several technical challenges. The first stems from the presence of both a low-rank component ($W_m\Phi$) and a sparse component ($S_m$), and the need to reliably distinguish between them—a problem previously studied in multivariate regression \citep{neghaban2011estimation} and single VAR models \citep{basu2019}. The second challenge concerns the lack of identifiability in the low-rank component $W_m\Phi$ due to the presence of the weight matrices $W_m$. For instance, rescaling $W_m$ and $\Phi$ by a common scalar leaves the product $W_m\Phi$ unchanged, making the parameters non-identifiable without further constraints. For a single VAR model, the first challenge has been effectively addressed in the existing literature by using a nuclear norm and an $\ell_1$ regularizer to estimate the low-rank and sparse components, respectively, while imposing an incoherence condition on the entries of the low-rank component to distinguish it from the sparse component.  However, these strategies alone are insufficient for resolving the identifiability issues in our model. To overcome this, we impose additional normalization constraints on $\Phi$, which increase both the complexity of the estimation problem and the difficulty of designing a provably convergent optimization algorithm.  We resolve these issues satisfactorily, by designing a \textit{multi-block} Alternating Direction Method of Multipliers (ADMM) and establish its convergence properties under mild assumptions, despite the non-convex, non-smooth nature of the underlying objective function. Note that unlike the well-studied convergence analysis of the 2-block ADMM (see \citealp{boyd2011distributed} and references therein), the multi-block version is significantly more challenging, due to requiring careful handling of the interdependencies between the multiple block updates and sometimes additional regularization to guarantee reliable performance. Further, consistency results for the proposed estimators of the various model parameters are provided.

The remainder of the paper is organized as follows. \Cref{sec:est} addresses parameter identifiability and develops an ADMM algorithm for estimation, along with convergence guarantees. \autoref{sec:cons} establishes consistency of the parameter estimates obtained via this algorithm. \Cref{sec:simu} covers tuning parameter selection and evaluates the algorithm’s performance on synthetic data. \autoref{sec:app} demonstrates the model and algorithm on a neuroscience data set. \autoref{sec:conclusion} provides concluding remarks.
For appendices, \autoref{sec:VAR} provides background on VAR models, while \autoref{sec:dis} expands the literature review on panel VAR models and offers a detailed comparison of our approach with the MAR model. Estimation details are presented in \autoref{sec:solution}. Properties of the proposed fixed nuclear-norm constraint space are studied in \autoref{sec:constraint}, and \autoref{sec:inco} discusses the incoherence between the low-rank and sparse components. \autoref{sec:multi_lag} extends our panel VAR setup to the multi-lag case and examines the applicability of our method. The proofs for algorithm convergence and estimator consistency are given in \autoref{sec:proof_convergence} and \autoref{sec:proof_consistent}, respectively. Finally, \autoref{sec:simu_sup} provides additional simulation details, including data generation, initialization, and supplementary results.

\textbf{Notation.} 
Denote by $\mathbf{1}_n$ an $n$-dimensional vector comprising of ones, and by $e_i$ the $i$-th standard basis vector with an appropriate dimension, $I_n$ as an identity matrix of size $n \times n$. The operator $\diag(x)$ returns either a diagonal matrix with the diagonal elements from $x$ if $x$ is a vector, or a vector consisting of the diagonal elements of $x$ if $x$ is a square matrix.
For matrix $X$, $X'$ denotes the transpose of $X$,
$X[i,j]$ (or $X_{i,j}$ if unambiguous within context) denotes the $(i,j)$-th element of the matrix $X$, $\rank(X)$ returns the rank of $X$, $\supp(X)$ is the support of matrix $X$ with nonzero elements, $\range(X)$ spans the matrix space with column-space contained in that of $X$, $\det(X)$ and $\tr(X)$ are the determinant and trace function of $X$, and the sign function $\sign(X)$ maps positive (zero, or negative) matrix entries to 1 (0, or -1) as the signs element-wise. $O(\cdot)$ and $o(\cdot)$ are the usual big and small O notations respectively, and we use $x \succsim y$ ($x \precsim y$) if there exists an absolute positive constant $c$ such that $x \ge c y$ ($x \le c y$).
For norms, $\|X\| = \|X\|_2$ by default represents the spectral norm, $\|X\|_1 = \sum_{i,j} |X_{i,j}|$ is the $\ell_1$ norm that sums up all the absolute values of matrix entries, and $\|X\|_\infty = \max_{i,j} |X_{i,j}|$ is the $\infty$-norm that returns the maximum absolute value. The norm $\|X\|_{q \to r}$ is defined analogous to matrices' operator norm such as $\|X\|_{q \to r} = \max\{\|Xv\|_r: \|v\|_q \le 1\}$. The nuclear norm of matrix $X$ is denoted as $\|X\|_*$. The $\ell_0$-norm $\|X\|_0$ calculates the support size of matrix (or vector) $X$, i.e., the number of non-zero elements in $X$. 
For a square matrix $X$ of dimension $n$, we write $\lambda_i(X)$ as the $i$-th eigenvalue of $X$ sorted non-increasingly by magnitudes.
Define the indicator function $\mathbbm{1}\{A\}$ that outputs 1 if the event $A$ is satisfied, and 0 otherwise.

\section{Estimation Procedure for the LSPVAR Model Parameters}\label{sec:est}

Recall that the LSPVAR model under consideration is defined as
$$
Y_m = A_m X_m + \epsilon_m; \qquad A_m = W_m \Phi + S_m, \quad m=1,\cdots,M
$$
where $Y_m = (X^m_1, \cdots, X^m_T), X_m = (X^m_0, \cdots, X^m_{T-1}) \in \mathbb{R}^{p \times T}$ and $\epsilon_m = (\epsilon^m_1, \cdots, \epsilon^m_T)\in \mathbb{R}^{p \times T}$. The next assumption ensures that the posited model is stable.

\begin{ass} LSPVAR model stability:\label{ass:stat}
    The coefficient matrices $A_m, m=1,\cdots,M$ satisfy the stability condition $|\lambda_1(A_m)| < 1$; i.e.,  the eigenvalue of largest magnitude is less than one —an analogous condition to that used in a single VAR model \citep{basu2015}.
\end{ass}

The following assumption ensures that the decomposition of $A_m$ in \eqref{eq:pVAR-model} is nonexplosive and the low-rank component $W_m \Phi$ is non-degenerate.
\begin{ass}\label{ass:exp_deg}
    The two summands of $A_m$, i.e., both $W_m \Phi$ and $S_m$, are upper bounded in terms of their spectral norms. Further, there exists $w > 0$,
    \begin{equation}\label{eq:wphi_nonsing}
        \min_j \frac{1}{M} \sum_{m=1}^M \| e_j' W_m \Phi \|^2 \ge w^2.
    \end{equation}
\end{ass}
Specifically, \eqref{eq:wphi_nonsing} implies that for each variable $j=1,\cdots,p$, the set of entities exhibiting singular dynamic patterns (i.e., $W_m=0$) in the low-rank component must constitute a minority within the panel. Note that the setup in \autoref{ex:mixture}, where certain $W_m=0$, does not violate \autoref{ass:exp_deg}, given the assumed structure of the other clusters.

The model parameters are estimated by minimizing the least squares objective function:
\begin{equation}\label{eq:obj_fun}
    f(\mathcal{W}, \mathcal{S}, \Phi; \mathcal{X}) = \sum_{m=1}^M \frac{1}{2T} \| Y_m - (W_m \Phi + S_m ) X_m \|_F^2,
\end{equation}
where $\mathcal{W} = \{W_1, \dots, W_M\}$, $\mathcal{S} = \{S_1, \dots, S_M\}$, and by $\mathcal{X} = \{ (X_1, Y_1), \dots, (X_M, Y_M) \}$.
In the unbalanced panel setting, where the number of observations $T_m$ varies across entities, the normalization factor $\frac{1}{2T}$ in \eqref{eq:obj_fun} is replaced by $\frac{1}{2 T_m}$ for each entity $m$. The subsequent analysis follows analogously by replacing $T$ by $T_m$ for entity-specific quantities, while setting $T = \min_m T_m$ when considering the entire panel.

Since the idiosyncratic components $S_m$ are assumed to be sparse, we employ the popular $\ell_1$-regularizer, given by 
\begin{equation}\label{eq:pen_S}
    \mathcal{P}_S(S_m; \eta) = \eta \|S_m\|_1, \ \ m=1,\cdots,M.
\end{equation}
To address the identifiability issue regarding the low rank component $W_m\Phi$-- i.e., distinguishing the weight  matrices $W_m$ from the common low-rank basis $\Phi$-- we impose the following constraints: (i) all rows of $\Phi$ have the same $\ell_2$ norm and (ii) its nuclear norm and an upper bound on its rank are fixed; namely,
\begin{equation}\label{eq:con_phi}
\begin{cases}
    & \Phi \in \mathbb{B}_p \coloneqq \{\Phi \in \mathbb{R}^{p\times p}: \|e_1' \Phi\| = \dots = \|e_p' \Phi\| \ne 0\},\\
    & \Phi \in \mathbb{L}_p(\hat r, \ell) \coloneqq \{\Phi \in \mathbb{R}^{p \times p}: \|\Phi\|_* = \ell, ~ \rank(\Phi) \le \hat r\},
\end{cases}
\end{equation}
where $\ell > 0$ is a constant and $\hat r$ a positive integer satisfying $r \le \hat r \le p$. 
\begin{rem} 
Note that each row of the low rank constraint has $p+1$ unknown parameters and hence a constraint of the type in \eqref{eq:con_phi} is required to address the identifiability issue. The space $\mathbb{B}_p$ ensures that the rows of $\Phi$ are strictly bounded away from singular vectors, thereby preventing any explosive estimation of $W_m$. Further, the space $\mathbb{L}_p(\hat{r}, \ell)$ corresponds to a simplex 
for the singular values of $\Phi$. As a result, the minimizer over $\mathbb{L}_p(\hat{r}, \ell)$ usually leads to a solution at the boundary (including vertices) of the space, which naturally results in the low-rankness of $\Phi$. In addition, $\mathbb{L}_p(\hat{r}, \ell)$ imposes an upper bound on the Frobenius norms of the matrices in the space. The intersection of the two constrained spaces has fixed $p$ ``degrees of freedom", thus resolving the identifiability issue between $\mathcal{W}$ and $\Phi$. Finally, note that simply imposing a nuclear norm constraint on the product $W_m\Phi$ for each $m$, does not resolve the identifiability issue, which is critical for model interpretability purposes.
\end{rem}

The optimization problem is then formulated as 
\begin{equation}\label{eq:full_obj}
    \min_{\mathcal{W}, \mathcal{S}, \Phi} F(\mathcal{W}, \mathcal{S}, \Phi; \mathcal{X}, \eta) = f(\mathcal{W}, \mathcal{S}, \Phi; \mathcal{X}) + \sum_{m=1}^M \mathcal{P}_S(S_m; \eta)\ \ \text{s.t.} \ \Phi \in \mathbb{B}_p \cap \mathbb{L}_p(\hat r, \ell).\end{equation}

We develop a multi-block ADMM algorithm\footnote{The structure of the constraint in \eqref{eq:full_obj} makes a block coordinate descent algorithm ineffective in updating the value of $\Phi$.} based on an auxiliary variable $\Phi_c$ that satisfies constraint \eqref{eq:con_phi}. The resulting augmented Lagrangian function is given by 
\begin{equation}\label{eq:admm}
    G(\mathcal{W}, \mathcal{S}, \Phi_c, \Phi, \Gamma; \mathcal{X}, \eta, \rho) = F(\mathcal{W}, \mathcal{S}, \Phi; \mathcal{X}, \eta) + \frac{\rho}{2} \|\Phi - \Phi_c\|_F^2 + \rho \langle \Gamma, \Phi - \Phi_c \rangle,
\end{equation}
where $\Phi \in \mathbb{R}^{p \times p}$, $\Phi_c \in \mathbb{B}_p \cap \mathbb{L}_p(\hat r, \ell)$, $\Gamma \in \mathbb{R}^{p \times p}$ is the dual variable matrix, and $\rho$ is a scalar coefficient that determines the step size of the parameter updates. Note that the domains $\mathbb{B}_p$ and $\mathbb{L}_p(\hat r, \ell)$ of $\Phi_c$, are both nonconvex, but connected compact semi-algebraic sets (as established in the proof of \autoref{thm:converge}).

The ADMM algorithm has three main primal descent blocks and one dual ascent block, summarized in \autoref{algo:iter}.
The solutions to all primal subproblems (\autoref{step:primal}) and the dual update (\autoref{step:dual}) are provided in \autoref{sec:solution}.

\begin{algorithm}[htbp]
    \caption{Estimation of $\mathcal{W}, \mathcal{S}, \Phi_c, \Phi$ and $\Gamma$}\label{algo:iter}
    \KwIn{Time series data $\{X_m, Y_m\}_{m=1}^M$, maximum rank $\hat r$, fixed nuclear norm $\ell$, Lasso parameter $\eta$, step size $\rho$, preset maximum iteration numbers $N$, convergence tolerance $\epsilon$.}
    \KwOut{Estimators $\mathcal{W}^{(n)}, \mathcal{S}^{(n)}, \Phi_c^{(n)}, \Phi^{(n)}, \Gamma^{(n)}$.}
    Initialize $\hat w_m^{(0)} = \mathbf{1}_p$, $\hat S_m^{(0)} = 0$ for $m = 1, \dots, M$, sample $\Phi^{(0)} = \Phi_c^{(0)}$ from $\mathbb{B}_p \cap \mathbb{L}_{p}(\hat r, \ell)$, and set $\Gamma^{(0)} = 0 \in \mathbb{R}^{p \times p}$.\\
    Evaluate the objective function as $G^{(0)}$ at $(\mathcal{W}^{(0)}, \mathcal{S}^{(0)}, \Phi_c^{(0)}, \Phi^{(0)}, \Gamma^{(0)})$.\\
    \For{$i = 1:N$}{
        Update $(\mathcal{W}^{(i)}, \mathcal{S}^{(i)})$, $\Phi_c^{(i)}$, and $\Phi^{(i)}$ sequentially with the optimal solutions of their corresponding subproblems.\label{step:primal}\\
        Update the dual variables $\Gamma^{(i)}$.\label{step:dual}\\
        Evaluate the objective function as $G^{(i)}$ at $(\mathcal{W}^{(i)}, \mathcal{S}^{(i)}, \Phi_c^{(i)}, \Phi^{(i)}, \Gamma^{(i)})$.\\
        Terminate and set $n = i$ if some convergence criteria are met, i.e., $|G^{(i-1)} - G^{(i)}| < \epsilon$ or $\frac{\|\Phi^{(i)} - \Phi^{(i-1)}\|_F}{\ell} < \epsilon$.\label{step:term}
    }
    Use ordinary least squares to refine and update the estimators $(\mathcal{W}^{(n)}, \mathcal{S}^{(n)})$ with fixed $\Phi^{(n)}$ and support of $\mathcal{S}^{(n)}$.\\
    Output the final estimators $\mathcal{W}^{(n)}, \mathcal{S}^{(n)}, \Phi_c^{(n)
}, \Phi^{(n)}, \Gamma^{(n)}$.\label{step:GN}
\end{algorithm}

\begin{rem}
    Note that the estimate of $\Phi_c$ is based on Dykstra's algorithm \citep{boyle1986}. The convergence of the corresponding subproblem's iterates can be verified using a similar proof strategy as outlined in \autoref{sec:optm_conv}; for more details, see \autoref{subsec:phic} in the Supplement.
\end{rem}

\begin{rem}\label{rem:bic}
    The selection of the step size $\rho$ and the nuclear norm paramter $\ell$ is discussed in \Cref{thm:converge,,thm:consist}, while the choice of the input rank $\hat{r}$ is addressed in \autoref{sec:cons}. These selections demonstrate strong empirical performance based on the synthetic data experiments presented in \autoref{subsec:rank}. 
    Finally, the sparsity tuning parameter $\eta$ is selected based on BIC \citep{wang2007,zou2007}, with additional implementation details provided in  \autoref{sec:simu}.
\end{rem}

\subsection{Convergence Guarantees for \autoref{algo:iter}}\label{sec:optm_conv}

As previously noted, establishing convergence guarantees for multi-block ADMM algorithms in the context of non-smooth and non-convex problems remains a challenging task.
We start by introducing assumptions and some additional notation.

\begin{ass}[Restricted Strong Convexity (RSC)]\label{ass:rsc}
    For an arbitrary matrix $\Delta \in \mathbb{R}^{p \times p}$, the following relationships holds
    $$
    \sum_{m=1}^M \frac{1}{2T} \| \Delta X_m \|_F^2  \ge \frac{\beta}{2} \sum_{m=1}^M \| \Delta \|_F^2,
    $$
    with $\beta = \min_m \{\beta_m$\}, and $\beta_m := \lambda_p(\frac{X_m X_m'}{T})$,
\end{ass}
\begin{rem}
    This is a commonly used assumption in the high-dimensional statistics literature \citep{wainwright2019high} and also in VAR models \citep{basu2015}. It is leveraged both in the convergence analysis of \autoref{algo:iter} and the consistency of the model parameters in \autoref{thm:consist}.
\end{rem}

The following assumptions are imposed on the objective function $F$.

\begin{ass}\label{ass:lip_hes}
Given a data realization $\mathcal{X}$, the function $F$:
\begin{enumerate}
    \item is $\beta_W$-convex with respect to $W_m$ for every $m$, and $\beta_\Phi$-convex with respect to $\Phi$.
    
    \item Its gradient is $\alpha_W$-Lipschitz continuous with respect to $W_m$ (for every $m$), $\alpha_S$-Lipschitz continuous with respect to $S_m$ (for every $m$), and $\alpha_\Phi$-Lipschitz continuous with respect to $\Phi$.
\end{enumerate}
\end{ass}

\begin{rem}
    Recall that the function $F$ in \eqref{eq:full_obj} consists of a sum of least squares terms and $\ell_1$ penalties. Therefore, the above convexity and continuity assumptions are satisfied provided the maximum and minimum eigenvalues of the corresponding data Gram matrices $\frac{X_m X_m'}{T}$ satisfy $0<\lambda_p(\frac{X_m X'_m}{T})\leq \lambda_1 (\frac{X_mX'_m}{T})<\infty$. Indeed, this condition is partially subsumed in \autoref{ass:rsc} regarding the lower bound of $\beta_m = \lambda_p(\frac{X_m X_m'}{T})$.
\end{rem}

The next result characterizes the behavior of the iterates of the parameter estimates generated by \autoref{algo:iter}.

\begin{thm}\label{thm:converge}
    Suppose that \Cref{ass:stat,ass:rsc,ass:exp_deg,ass:lip_hes} hold, and based on a data realization $\mathcal{X}$ from model \eqref{eq:pVAR-model} and for step size satisfying
    \begin{equation}\label{eq:rho_phi_lower}
    \rho \succsim \max \left\{ \frac{M \alpha_W^2}{\beta_W}, \frac{M \alpha_S^2}{\beta}, \alpha_\Phi \right\},
    \end{equation}
    the sequence of iterates of the model parameters 
    $(\mathcal{W}^{(i+1)}, \mathcal{S}^{(i+1)}, \Phi_c^{(i+1)}, \Phi^{(i+1)}, \Gamma^{(i+1)})$ from \autoref{algo:iter} converges globally to a stationary point $(\widehat{\mathcal{W}}, \widehat{\mathcal{S}}, \widehat{\Phi}_c, \widehat{\Phi}, \widehat{\Gamma})$ of the augmented Lagrangian function $G$.
\end{thm}

 The first step is to establish a “sufficient descent property,” provided in \Cref{prop:suff_desc}. The second step is to show that the augmented Lagrangian function is lower bounded along the sequence of iterates, as shown in \Cref{prop:conv_func}. When these two conditions hold, the set of accumulation points of \autoref{algo:iter} is guaranteed to be non-empty, compact, and connected \cite[see Remark 5 in][]{bolte2014}. The final requirement for proving global convergence to a critical point of $G$ is to verify that $G$ satisfies the Kurdyka-\L{}ojasiewicz (KL) property \citep{attouch2013}, which ensures that the sequence of iterates generated by \autoref{algo:iter} forms a Cauchy sequence. Together, these three components establish the claims in \autoref{thm:converge}. All technical proofs are deferred to \autoref{sec:proof_convergence}.

\begin{prop}[Sufficient Descent]\label{prop:suff_desc}
    Under the Assumptions of \autoref{thm:converge}, the objective function evaluated at the sequence $\{(\mathcal{W}^{(i)}, \mathcal{S}^{(i)}, \Phi_c^{(i)}, \Phi^{(i)}, \Gamma^{(i)})\}$ is monotonically decreasing. Moreoever, there exists a constant $\mathcal{C} > 0$, such that
    \begin{multline*}
    	G(\mathcal{W}^{(i)}, \mathcal{S}^{(i)}, \Phi_c^{(i)}, \Phi^{(i)}, \Gamma^{(i)}) - G(\mathcal{W}^{(i+1)}, \mathcal{S}^{(i+1)}, \Phi_c^{(i+1)}, \Phi^{(i+1)}, \Gamma^{(i+1)})\\
    	\ge \mathcal{C} (\sum_{m=1}^M \|W_m^{(i)} - W_m^{(i+1)}\|_F^2 + \sum_{m=1}^M \|S_m^{(i)} - S_m^{(i+1)}\|_F^2 + \|\Phi^{(i)} - \Phi^{(i+1)}\|_F^2 + \|\Phi_c^{(i)} - \Phi_c^{(i+1)}\|_F^2).
    \end{multline*}
\end{prop}

\begin{prop}[Lower bound of the augmented Lagrangian function] \label{prop:conv_func}
    Under the Assumptions of \autoref{thm:converge}, the evaluation of function $G$ at the estimated sequence from \autoref{algo:iter}, $G(\mathcal{W}^{(i+1)}, \mathcal{S}^{(i+1)}, \Phi_c^{(i+1)}, \Phi^{(i+1)}, \Gamma^{(i+1)})$, is lower bounded for all $i$ and converges as $i \to \infty$.
\end{prop}

The next result is based on \autoref{prop:conv_func}, and the fact that the sequence $G(\mathcal{W}^{(i+1)}, \mathcal{S}^{(i+1)}, \Phi_c^{(i+1)}, \Phi^{(i+1)}, \Gamma^{(i+1)})$ forms a Cauchy sequence.

\begin{prop}\label{prop:conv}
    Under the Assumptions of \autoref{thm:converge}, $(\mathcal{W}^{(i+1)}, \mathcal{S}^{(i+1)}, \Phi_c^{(i+1)}, \Phi^{(i+1)}, \Gamma^{(i+1)})$ is convergent to some limiting point $(\widehat{\mathcal{W}}, \widehat{\mathcal{S}}, \widehat{\Phi}_c, \widehat{\Phi}, \widehat{\Gamma})$ with $i \to \infty$. 
    
    Moreover, the limiting point $(\widehat{\mathcal{W}}, \widehat{\mathcal{S}}, \widehat{\Phi}_c, \widehat{\Phi}, \widehat{\Gamma})$ satisfies the local first-order optimality of \eqref{eq:admm} within the corresponding support domain.
\end{prop}

\begin{rem}
    The arguments used to prove \autoref{prop:conv} also imply that the sequence of ADMM iterates achieves  an $o(\frac{1}{i})$ convergence rate \cite[Lemma 1.1]{deng2017}. Consequently, for a given error tolerance $\epsilon > 0$ (measuring the difference between consecutive iterates), the number of iterations required for convergence is of order $O(1/\epsilon)$.
\end{rem}

\section{Consistency of the Model Parameter Estimates}\label{sec:cons}

In this section, we discuss the consistency results for the estimators obtained from \autoref{algo:iter}, focusing on the case where the model structure in \eqref{eq:pVAR-model} is correctly specified. Notably, the derivation shows that the effect of a potentially overspecified $\hat{r}$ can be well controlled under mild conditions.

Based on the convergence guarantees in \autoref{thm:converge}, the final output of \autoref{algo:iter} corresponds to a limiting stationary point $(\widehat{\mathcal{W}}, \widehat{\mathcal{S}}, \widehat{\Phi}_c, \widehat{\Phi}, \widehat{\Gamma})$. Given the nonconvex nature of the optimization problem, \autoref{thm:converge} does not ensure convergence to a global minimum of the objective function \eqref{eq:full_obj}. Moreover, not all stationary points of \eqref{eq:full_obj} possess equally desirable statistical properties. Therefore, in the following analysis, we focus on stationary points obtained by initializing \autoref{algo:iter} within a suitable neighborhood of the true model parameters $(\mathcal{W}^*, \mathcal{S}^*, \Phi^*)$ with $\Phi^* \in \mathbb{B}_p \cap \mathbb{L}_p(r, \ell)$.
In practice, however, our simulation studies in \autoref{sec:simu} indicate that \autoref{algo:iter} consistently achieves strong performance, regardless of the initialization. These empirical results suggest that, even under broader initialization settings, the favorable statistical properties we establish below may be implicitly supported by the structure of the optimization problem.

Before stating the theoretical results, we introduce some additional notation. Define
$W = \max_m \|W_m^*\|$, $\widehat{W} = \max (\max_m \| \widehat{W}_m \|, W)$, and $W_m^\dag = \frac{W_m^*}{W}$ for every $m$. The proposition below describes the estimation power of \autoref{algo:iter} with $(\mathcal{W}, \mathcal{S}, \Phi, \Phi_c)$ initialized in a neighborhood of $(\mathcal{W}^*, \mathcal{S}^*, \Phi^*, \Phi^*)$, given the data realization $\mathcal{X}$. Its proof is presented in \autoref{sec:proof_consistent}.

\begin{prop}\label{prop:upper}
    Suppose \Cref{ass:stat,ass:exp_deg,ass:rsc,ass:lip_hes} hold, and that \autoref{algo:iter} is initialized in a neighborhood of the true parameter $(\mathcal{W}^*, \mathcal{S}^*, \Phi^*)$. Further, assume the following:
\begin{enumerate}
        \item[(B1)] there exists a constant $\phi > 0$ such that
        \begin{itemize}
        	\item the true parameter $\Phi^* \in \mathbb{B}_p \cap \mathbb{L}_p(r, \ell)$ satisfies $\|\Phi^*\|_\infty \le \frac{\phi \ell}{\sqrt{r} p}$;
        	\item the estimator $\widehat{\Phi} \in \mathbb{B}_p \cap \mathbb{L}_p(\hat r, \ell)$ satisfies $\|\widehat{\Phi}\|_\infty \le \frac{\phi \ell}{\sqrt{r} p}$;
        \end{itemize}
        
        \item[(B2)] for every $m$, the sparse matrix $S_m^*$ has at most $s$ non-zero entries.
    \end{enumerate}
    Then, for tuning parameter $\eta \ge \frac{4 \beta \phi \widehat{W} \ell}{\sqrt{r} p} + \max_m \left| \frac{\varepsilon_m X_m'}{T} \right|_\infty$, and selecting $\zeta = \min\{ \frac{\beta_\Phi}{M \beta}, \frac{\beta_W}{\beta}, 1 \} \succsim \min \left\{ \frac{r p}{\ell^2}, \frac{\ell^2}{\hat{r} p}, 1 \right\}$, the solution $(\widehat{\mathcal{W}}, \widehat{\mathcal{S}}, \widehat{\Phi})$ satisfies
    \begin{multline}\label{eq:error}
        \|\widehat{\Phi} - \Phi^*\|_F^2 + \frac{1}{M} \sum_{m=1}^M (\|\widehat{S}_m - S_m^*\|_F^2 + \|\widehat{W}_m - W_m^*\|_F^2)\\
        \precsim \frac{1}{\zeta^2} \left( \frac{\hat{r}^2}{\ell^2} + \frac{\hat{r} s}{r p} + \frac{s}{\beta^2} \max_m \|\frac{\varepsilon_m X_m'}{T}\|_\infty^2 + \frac{\ell^2}{\beta^2 M r p} \sum_{m=1}^M \| \frac{\varepsilon_m X_m'}{T}\|_2^2 + \frac{r \hat{r} p}{\beta^2 \ell^2} \| \sum_{m=1}^M \frac{W_m^\dag \varepsilon_m X_m'}{MT}\|_2^2 \right).
    \end{multline}
\end{prop}

\begin{rem}\label{rem:bound}
    The assumptions on the $\infty$-norm bounds of $\Phi^*$ and $\widehat{\Phi}$ are mild, especially given that $p$ is assumed to be large and growing with the sample size $T$; see \autoref{sec:proof_consistent} in the Supplement for further explanation and discussion.
    Moreover, since the constants in \autoref{ass:lip_hes} have become tighter at the estimate $(\widehat{\mathcal{W}}, \widehat{\mathcal{S}}, \widehat{\Phi})$, the feasible choice for the step size coefficient $\rho$ can be refined from \eqref{eq:rho_phi_lower} to $\rho \succsim \frac{M \hat{r} p}{\beta \ell^2} \max_m (\|\frac{X_m X_m'}{T}\|^2 + \| \frac{\epsilon_m X_m'}{T} \|^2)$.
\end{rem}

Then, for the deviation bounds appearing in \eqref{eq:error}, we adapt analogous results from \cite{basu2015,basu2019} to the LSPVAR model under consideration, as formalized in \autoref{prop:deviation}.

\begin{prop}\label{prop:deviation}
    Given the model setup in \eqref{eq:pVAR-model}, suppose that \Cref{ass:stat,ass:lip_hes} are satisfied. Consider a realization of the data $\{X_m\}_{m=1}^M$ and $\{\varepsilon_m\}_{m=1}^M$, and define\footnote{The quantities $\tau_{\max}$, $\tau_{\min}$, $\Psi(h_{X_m})$, $\psi(h_{X_m})$ and $\mathcal{A}_m$ are related to the characteristic polynomials and spectral densities of the individual VAR models. Their rigorous definitions are given in \autoref{sec:VAR}.}
    $$
    \begin{aligned}
        \xi & = \max_m \lambda_1(\Sigma_m) \left[ 1 + \frac{1 + \tau_{\max}(\mathcal{A}_m)}{\tau_{\min}(\mathcal{A}_m)} \right],\\
        \xi_\dag & = \max_m \lambda_1(\Sigma_m) + \max_m \frac{\lambda_1(\Sigma_m)}{\tau_{\min}(\mathcal{A}_m)} + \max_m \frac{\lambda_1(\Sigma_m) \tau_{\max}(\mathcal{A}_m)}{\tau_{\min}(\mathcal{A}_m)}.
    \end{aligned}
    $$
    Assume $\log(p) \succsim \log(M)$. Then, there exist constants $c_i > 0$ ($i=1,2,3$) such that:
    \begin{enumerate}
        \item for $T \succsim p$ and $T \succsim \log(p)$, with probability at least $1 - c_2 \exp(-c_3 \log(p))$,
        $$
        \max_m \|\frac{\varepsilon_m X_m'}{T}\|_2 \le c_1 \xi \sqrt{\frac{p}{T}}, \qquad \max_m \|\frac{\varepsilon_m X_m'}{T}\|_\infty \le c_1 \xi \sqrt{\frac{\log(p)}{T}};
        $$

        \item for $T \succsim p$, with probability at least $1 - c_2 \exp(-c_3 \log(p))$,
        $$\|\sum_{m=1}^M \frac{W_m^\dag \varepsilon_m X_m'}{MT}\|_2 \le c_1 \xi_\dag \sqrt{\frac{p}{MT}};$$

        \item for $T \succsim p \Psi^2(h_{X_m}) / \psi^2(h_{X_m})$, with probability at least $1 - c_2 \exp(-c_3 \log(p))$,
        $$
        \min_m \lambda_p(\frac{X_m X_m'}{T}) \ge \min_m \frac{1}{4 \pi} \cdot \frac{\lambda_p(\Sigma_m)}{\tau_{\max}(\mathcal{A}_m)}.
        $$
    \end{enumerate}
\end{prop}

Combining \Cref{prop:upper} and \Cref{prop:deviation}, and noting that $\xi \le \xi_\dag$, we obtain the following theorem, which establishes the high-probability consistency of the proposed estimator.

\begin{thm}\label{thm:consist}
    Under the setting of \Cref{prop:upper,prop:deviation}, assume there exists a constant $\iota$ such that $\hat{r} \le \iota r$, and set$\ell = \sqrt{\hat{r} p}$. Then, there exist universal positive constants $C_i$ for $i = 1, 2, 3, 4$, such that the solution obtained from \autoref{algo:iter} satisfies, with probability at least $1 - C_3 \exp( - C_4 \log(p))$,
    \begin{multline}\label{eq:error_plug}
        \|\widehat{\Phi} - \Phi^*\|_F^2 + \frac{1}{M} \sum_{m=1}^M (\|\widehat{S}_m - S_m^*\|_F^2 + \|\widehat{W}_m - W_m^*\|_F^2)\\
        \le C_1 \iota \cdot \frac{r + s}{p} + C_2 \xi_\dag^2 \max_m \frac{\tau_{\max}^2(\mathcal{A}_m)}{\lambda_p^2(\Sigma_m)} \big( \frac{s\log(p) + \iota p}{T} + \frac{r p}{M T} \big).
    \end{multline}
\end{thm}

The roadmap to establish \autoref{thm:consist} is as follows. First, \autoref{prop:upper} establishes an upper bound on the estimation errors, which is dependent on the deviation bounds of the time series data. Next, the deviation bounds are controlled with high-probability based on \autoref{prop:deviation}. Then, the consistency rate is obtained in a straightforward manner, by plugging the obtained bounds back in \eqref{eq:error} and selecting appropriate values for $\ell$ and $\hat{r}$.

\begin{rem}\label{rem:error}
    Note that $\iota$ is treated more like a pseudo-constant and serves as a guiding reference for selecting $\hat{r}$. The derived error bound consists of two distinct parts. The first part arises from the inherent non-identifiability of the model in separating the low-rank and sparse components. It depends only on the rank $r$, the sparsity $s$ and the overspecification $\iota$. It will not vanish even with increasing sample size, but remains small under a reasonably chosen $\iota$ in correctly specified settings, where $r \ll p$ and $s \ll p$. The second term reflects  the randomness of the data and vanishes as the sample size (time series length $T$) increases. In particular, the term $\frac{s\log(p) + \iota p}{T}$ can be interpreted as the proxy parametric convergence rate of estimating $(W_m, S_m)$, whose effective degrees of freedom are approximately $DOF\approx s \log(p) + p$, based on time series of length $T$ for each entity. Similarly, the term $\frac{r p}{MT}$ corresponds to the rate of estimating the shared low-rank component $\Phi$, with $\mathnormal{DOF} \approx r p$, using the entire panel $\{X_m\}_{m=1}^M$ of size $M T$. The signal-to-noise ratio of the model further influences the second term through factors such as $\max_m \frac{\tau_{\max}^2(\mathcal{A}_m)}{\lambda_p^2(\Sigma_m)}$ and $\xi_\dag$.
\end{rem}

As previously noted,  there exists a source of indeterminacy between the diagonal matrices $\mathcal{W}$ and the low-rank structure $\Phi$, as scaling both factors in the product $W_m \Phi$ by the same constant leaves the transition matrices $A_m$ unchanged. While the constraint set $\mathbb{B}_p \cap \mathbb{L}_p(\hat{r}, \ell)$ serves to fix the overall magnitude of $\widehat{\Phi}$, there remains a possibility that part of $\widehat{\Phi}$'s nuclear norm may ``leak" outside the primary rank-$r$ subspace of the true low-rank component $\Phi^*$. To account for this, we present the following corollary, which provides classical consistency guarantees for the combination of $\widehat{\mathcal{W}}$ and $\widehat{\Phi}$. A proof sketch is included in \autoref{sec:proof_consistent}.

\begin{cor}\label{cor:error}
    Under the Assumptions of \autoref{thm:consist} and in addition assuming  $\eta \ge \frac{2 \beta \phi \widehat{W} \ell}{\sqrt{r} p} + \max_m \|\frac{\varepsilon_m X_m'}{T}\|_\infty$, there exist universal positive constants $C_i$ for $i = 1, 2, 3, 4$, such that the solution obtained from \autoref{algo:iter} satisfies, with probability at least $1 - C_3 \exp( - C_4 \log(p))$,
    \begin{equation}\label{eq:error_cor}
        \frac{1}{M} \sum_{m=1}^M (\|\widehat{S}_m - S_m^*\|_F^2 + \|\widehat{W}_m \widehat{\Phi} - W_m^* \Phi^*\|_F^2) \le C_1 \cdot \frac{\iota s}{p} + C_2 \xi^2 \cdot \max_m \frac{\tau_{\max}^2(\mathcal{A}_m)}{\lambda_p^2(\Sigma_m)} \cdot \frac{s\log(p) + \iota r p}{T}.
    \end{equation}
\end{cor}

\begin{rem}
    The error bound in \eqref{eq:error_cor} aligns with analogous results in the literature for a single VAR model; see, e.g., \cite{basu2019}. The main difference between the right-hand side bound in \autoref{cor:error} and that in \autoref{thm:consist} lies in the presence of the factor $\frac{1}{M}$ in the rates $\frac{\hat{r} p}{T}$ versus $\frac{\hat{r} p}{MT}$. This is due to the $\Phi$-relevant error bounds being influenced by idiosyncratic rescaling factors $\mathcal{W}$. Hence, the bound \eqref{eq:error_cor} is less sharp, requiring $T \gg r p$, compared to $T \gg \frac{r p}{M}$ in \eqref{eq:error_plug}.
\end{rem}

\section{Performance Evaluation}\label{sec:simu}

The proposed LSPVAR parameter estimation strategy based on \autoref{algo:iter} is assessed through numerical experiments with synthetic data. he simulation studies in \autoref{subsec:rank} primarily focus on practical strategies for selecting the input rank $\hat r$. Additionally, \autoref{subsec:cluster} explores variants of the scenario described in \autoref{ex:mixture}, highlighting the model’s capability to uncover heterogeneous patterns across a panel of VAR models.

To select the optimal penalty coefficient $\eta$, we perform a grid search over a geometric sequence covering an appropriate range, aiming to minimize the Bayesian Information Criterion (BIC). For each estimate $(\widehat{\mathcal{W}}, \widehat{\mathcal{S}}, \widehat{\Phi})$ corresponding to a given $\eta$, we calculate  $\mathnormal{RSS}_\eta = \sum_{m=1}^M \|Y_m - (\widehat{W}_m \widehat{\Phi} + \widehat{S}_m) X_m\|_F^2$, and define the model degrees of freedom $\hat d$ as
$$
\hat d = (2 p - \rank(\widehat{\Phi})) \cdot \rank(\widehat{\Phi}) + p (M - 1) + \sum_{m=1}^M \|\widehat{S}_m\|_0.
$$
The BIC is then computed as 
$$
\mathnormal{BIC}_\eta = M T p \cdot \log(\frac{\mathnormal{RSS}_\eta}{M T p}) + \hat d \cdot \log(M T).
$$
Besides information criteria, we also consider several classical metrics to evaluate the performance of our approach. Specifically, we assess sensitivity, specificity, and overall accuracy to study the sparse recovery of $\mathcal{S}$.
Additionally, we compute the relative errors of the coefficient matrices and the Frobenius norms of the normalized components to evaluate the overall quality of the algorithm’s estimates. These metrics are provided in \autoref{tab:rank} in \autoref{subsec:fig_tab}.

\subsection[Choice of Input Rank]{Choice of Input Rank $\hat r$}\label{subsec:rank}

Next, the performance of \autoref{algo:iter} is evaluated based on different choices for $\hat{r}$,  while also detailing tuning parameter selection guidelines for $\eta$ and $\rho$. The setting considered is $(M, p, r, s) = (20, 40, 5, 30)$ with $T = 2 r p = 400$. Candidate ranks are $
\hat r \in \{3, 5, 10, 15, 20, 25, 30, 35, 40\}$, and $\ell = \sqrt{\hat{r} p}$ based on \autoref{thm:consist}. The step size is selected as $\rho \in \{\frac{M}{20}, \frac{M}{5}, M\}$ based on \eqref{eq:rho_phi_lower} and \autoref{rem:bound}. The grid search interval for $\eta$ is set as $[4 \times 10^{-2}, 2.5 \times 10^{-1}]$, selected from a coarser pilot run.

\autoref{fig:BIC} illustrates how key performance metrics vary with the tuning parameter $\eta$, input rank $\hat{r}$, and step $\rho$. Focusing on $\eta$, with  other parameters fixed, we observe that the minimal BIC is consistently attained around $\eta \approx 10^{-2}$, which coincides with the minimum relative Frobenius error of the coefficient matrices. Additionally, \autoref{fig:BIC} demonstrates that the sparse component's recovery accuracy is also maximized near this $\eta$ value. These results validate the grid search procedure as an effective strategy for tuning $\eta$ in practice, even without prior knowledge of the sparsity level. The other two parameters $\hat{r}$ and $\rho$, based on \autoref{fig:BIC} exhibit minimal influence on performance, provided $\hat{r} \ge r$ and $\rho = O(M)$. 
This observation is consistent with our theoretical discussion of $\hat{r}$ and $\rho$ in \autoref{thm:consist} and \Cref{rem:bound,rem:error}.

It is worth emphasizing that the setting with $T = 2 r p$ under consideration, is challenging for either a single low-rank plus sparse model, or a general VAR model, due to lack of adequate sample size. Nevertheless, our empirical results demonstrate that \autoref{algo:iter} accurately recovers the LSPVAR model parameters.

\begin{figure}[htbp]
    \centering
    \includegraphics[width=\textwidth]{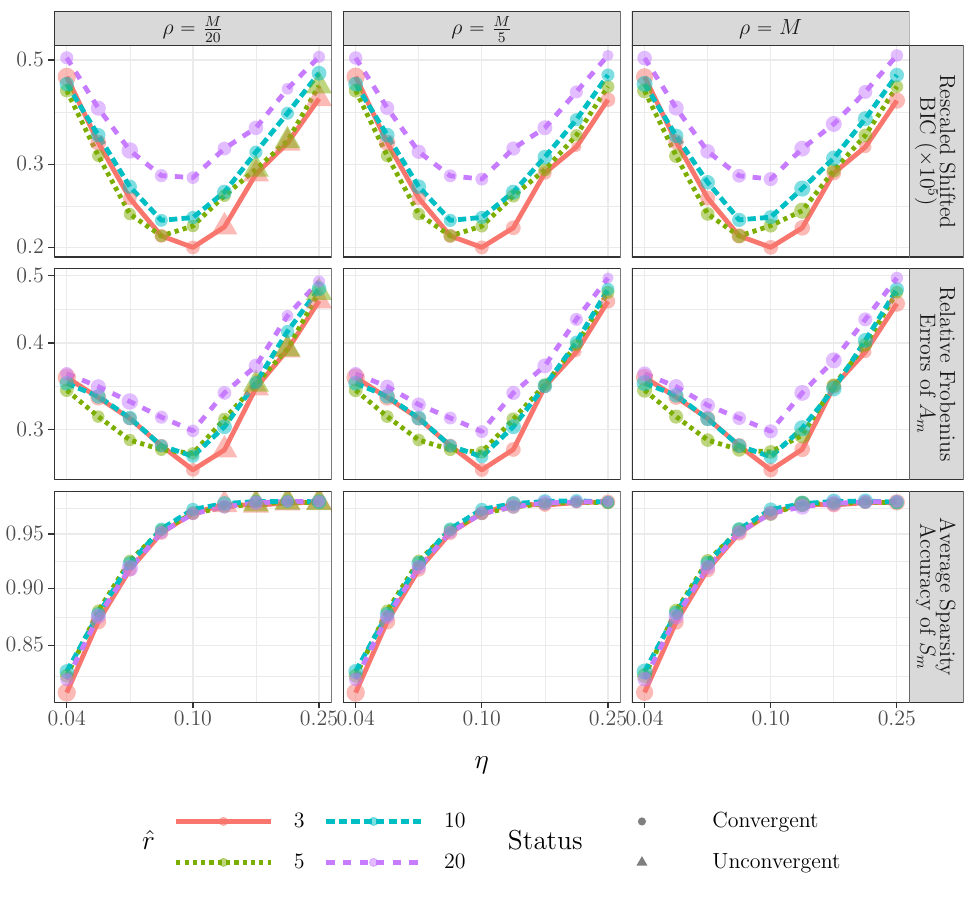}
    \caption{{\small The trajectories of BIC, relative errors of $A_m$ and sparse recovery accuracy of $S_m$ as functions of $\eta$ are depicted. For better visualization, the BIC values are shifted and rescaled. The size of each scatter point reflects the number of iterations required for convergence, while triangular shaped points indicate cases where \autoref{algo:iter} fails to meet the convergence tolerance $\epsilon = 5 \times 10^{-6}$ within the maximum allowed $N = 4\times10^5$ iterations.}}
    \label{fig:BIC}
\end{figure}

Based on the above findings and discussion, a choice of $\hat{r} \ge r$ and $\rho = O(M)$, is recommended in practice.

\subsection{Assessing Panel Heterogeneity}\label{subsec:cluster}

As stated in \autoref{sec:intro}, LSPVAR can effectively infer latent heterogeneity in the panel based on the obtained estimates from \autoref{algo:iter}, even in the presence of ``degenerate configurations", where some of the constituent models in the panel exhibit pure sparse or pure low-rank structure.

\noindent\textbf{Revisiting the model in \autoref{ex:mixture}}: recall the setup with $(M, p, r, s, T) = (20, 40, 5, 30, 400)$. The structure of the VAR models within the panel is as follows:
\begin{enumerate}
    \item Two clusters with nonsingular diagonal matrices $W_{1} = \dots = W_{5} \ne 0$, $W_{6} = \dots = W_{10} \ne 0$, and non-restrictive sparse matrices $\{S_m\}_{m=1}^{10}$. We refer to these as \emph{Cluster 1} and \emph{Cluster 2}, respectively.

    \item One cluster of purely sparse structure, i.e., the diagonal matrices $W_{11} = \dots = W_{14} = 0$. This is labeled as \emph{Singular W}.
    
    \item One cluster with identical diagonal matrices $W_{15} = \dots = W_{18} \ne 0$ and singular sparse matrices $S_{15} = \dots = S_{18} = 0$. This is labeled as \emph{Singular S}.

    \item Two isolated entities, \emph{Isolate 1} and \emph{Isolate 2},  whose diagonal matrices and sparse matrices are different than all the previous ones. 
\end{enumerate}

The estimators are obtained by minimizing the BIC through grid search over $\eta$, with $(\hat{r}, \rho) = (\frac{p}{2}, \frac{M}{20})$. For visualization, \autoref{fig:PC} presents a 3-dimensional scatter plot of the leading principal components of $\widehat{\mathbf{W}} = \left( \diag(\widehat{W}_1), \dots, \diag(\widehat{W}_M) \right)$. It is noticeable from \autoref{fig:PC} that the cluster patterns and/or isolated outliers are well separated and captured by the LSPVAR estimates.

Additionally, \autoref{tab:cluster} reports evaluation metrics based on twenty five replicates of the same underlying model. Notably, Notably, the estimation errors for $W_m$ in the \textit{Singular $W$} group and for $S_m$ in the \textit{Singular $S$} group are significantly smaller, demonstrating the model's ability to uncover latent singular structures in the panel. Overall, both the error metrics and sparsity recovery results confirm the effectiveness of the proposed model and algorithm in handling heterogeneous and idiosyncratic settings.

\begin{table}[htbp]
    \centering
    {\scriptsize
\begin{tabular}{|l|r|r|r|r|r|r|}

\hline
Cluster & \begin{tabular}[c]{@{}c@{}}Average\\Relative Frobenius\\Error of~$A_m$\end{tabular} & \begin{tabular}[c]{@{}c@{}}Average\\Absolute Frobenius\\Error of~$W_m$\end{tabular} & \begin{tabular}[c]{@{}c@{}}Average\\Absolute Frobenius\\Error of~$S_m$\end{tabular} & \begin{tabular}[c]{@{}c@{}}Sparsity\\Recovery\\Accuracy\end{tabular} & \begin{tabular}[c]{@{}c@{}}Sparsity\\Recovery\\Sensitivity\end{tabular} & \begin{tabular}[c]{@{}c@{}}Sparsity\\Recovery\\Specificity\end{tabular} \\ 
\hline
Cluster 1 & 0.099 & 0.303 & 0.274 & 0.992 & 0.877 & 0.994 \\ 
\hline
Cluster 2 & 0.099 & 0.409 & 0.390 & 0.990 & 0.822 & 0.995 \\ 
\hline
Singular~$W$ & 0.081 & \textbf{0.107} & 0.264 & 0.991 & 0.875 & 0.995 \\ 
\hline
Singular~$S$ & 0.134 & 0.953 & \textbf{0.070} & 0.999 & NA & 0.999 \\ 
\hline
Isolate 1 & 0.119 & 0.271 & 0.340 & 0.988 & 0.822 & 0.993 \\ 
\hline
Isolate 2 & 0.090 & 0.774 & 1.300 & 0.986 & 0.792 & 0.993 \\
\hline
\end{tabular}}
    \caption{{\small Summary statistics by cluster based on 25 simulation replicates. The absolute errors of the rescaling effects $W_m$'s are computed under the normalization $\|\Phi\|_* = \sqrt{\hat{r} p}$, consistent with the setup in \autoref{thm:consist}.}}
    \label{tab:cluster}
\end{table}

\noindent
\textbf{A larger size heterogeneous panel:} we consider a setting with
$(M, p, r, s, T) = (50, 80, 5, 100, 2000)$, wherein the structure of the VAR models in the panel is as follows:
   two 19-entity clusters, \emph{Cluster 1} and \emph{Cluster 2}, with identical weight matrices $W_m$  within each cluster; one 5-entity cluster with \emph{Singular W} ($W_m = 0$); one 5-entity cluster with \emph{Singular S} ($S_m = 0$), and two isolated entities, \emph{Isolate 1} and \emph{Isolate 2}.

The summary of the estimated weight matrices $W_m$ based on PCA is depicted in \autoref{fig:PC_highD}. Analogously to the smaller heterogeneous panel analyzed above, the estimated weight matrices accurately capture the latent structure of the panel. Additional performance metrics are reported in \autoref{tab:cluster_highD}, demonstrating that \autoref{algo:iter} yields highly accurate estimates of the LSPVAR model paramters.

\begin{figure}[htbp]
    \centering
    \includegraphics{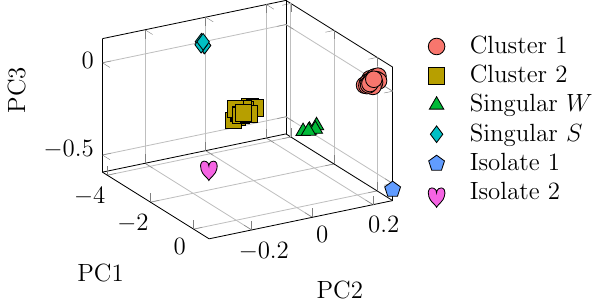}
    \caption{{\small Top three principal components of the estimated $\mathcal{W}$ for the larger heterogeneous panel comprising mixed sub-models, with parameters $(M, p, r, s, T) = (50, 80, 5, 100, 2000)$.}}
    \label{fig:PC_highD}
\end{figure}

\begin{table}[htbp]
    \centering
    {\scriptsize
\begin{tabular}{|l|r|r|r|r|r|r|}

\hline
Cluster & \begin{tabular}[c]{@{}c@{}}Average\\Relative Frobenius\\Error of~$A_m$\end{tabular} & \begin{tabular}[c]{@{}c@{}}Average\\Absolute Frobenius\\Error of~$W_m$\end{tabular} & \begin{tabular}[c]{@{}c@{}}Average\\Absolute Frobenius\\Error of~$S_m$\end{tabular} & \begin{tabular}[c]{@{}c@{}}Sparsity\\Recovery\\Accuracy\end{tabular} & \begin{tabular}[c]{@{}c@{}}Sparsity\\Recovery\\Sensitivity\end{tabular} & \begin{tabular}[c]{@{}c@{}}Sparsity\\Recovery\\Specificity\end{tabular} \\ 
\hline
Cluster 1 & 0.104 & 0.160 & 0.326 & 0.996 & 0.809 & 0.999 \\ 
\hline
Cluster 2 & 0.103 & 0.159 & 0.313 & 0.997 & 0.796 & 1.000 \\ 
\hline
Singular~$W$ & 0.187 & \textbf{0.130} & 0.372 & 0.994 & 0.694 & 1.000 \\ 
\hline
Singular~$S$ & 0.102 & 1.670 & \textbf{0.018} & 1.000 & NA & 1.000 \\ 
\hline
Isolate 1 & 0.097 & 0.357 & 0.982 & 0.995 & 0.754 & 0.998 \\ 
\hline
Isolate 2 & 0.053 & 0.201 & 0.317 & 0.996 & 0.852 & 0.999 \\
\hline
\end{tabular}}
    \caption{{\small The cluster-based summary statistics from 25 simulation replicates of the large size heterogeneous panel, with parameters $(M, p, r, s, T) = (50, 80, 5, 100, 2000)$.}}
    \label{tab:cluster_highD}
\end{table}

\section{Application to a Neuroscience Data Set}\label{sec:app}

To illustrate the usefulness in real life applications of LSPVAR, we apply it to a data set comprising EEG signals \citep{chen2008} obtained from 22 subjects. Specifically, data on 11 female and 11 male undergraduate students at Texas State University were collected (age range 18-26 years, mean age=21.1 years), while they performed the following two consecutive tasks: alternating one-minute intervals of eyes open (EO) and eyes closed (EC). The EEG signals, sampled at 256 Hz, were recorded from 71 scalp channels as illustrated in \autoref{fig:scalp}, with channel locations summarized in \autoref{tab:channels}.

For each subject, we separately extracted the EO and EC segments and filtered the alpha band signals (8–13 Hz), which are known to be relevant for visual processing tasks. We followed the preprocessing pipeline described in \citet{bai2022} and related references. The resulting dataset comprises $M = 44 = 22 \times 2$ panels, with dimension $p = 71$ and time length $T = 2000$.

The key objective is to identify potential clustering patterns corresponding to the EO and EC conditions, as well as to capture additional subject-level heterogeneity. Guided by the parameter selection discussion in \autoref{subsec:rank}, we set $(\hat{r}, \rho) = (18, 4.4)$, approximately $(\frac{p}{4}, \frac{M}{10})$, and estimate the parameters $W_m$, $\Phi$, and $S_m$ accordingly. Similar to the setup in \autoref{ex:mixture}, we summarize the results using PCA of the estimated $W_m$ matrices in \autoref{fig:EEG_W} and \autoref{fig:W_box}. The findings reveal considerable heterogeneity across subjects, particularly reflected in the $W_m$ scaling components. Notably, male subjects exhibit a somewhat greater degree of heterogeneity than females, an effect that is even more pronounced in the alpha band data.

\begin{figure}[htbp]
    \centering
    \includegraphics[width=0.4\linewidth]{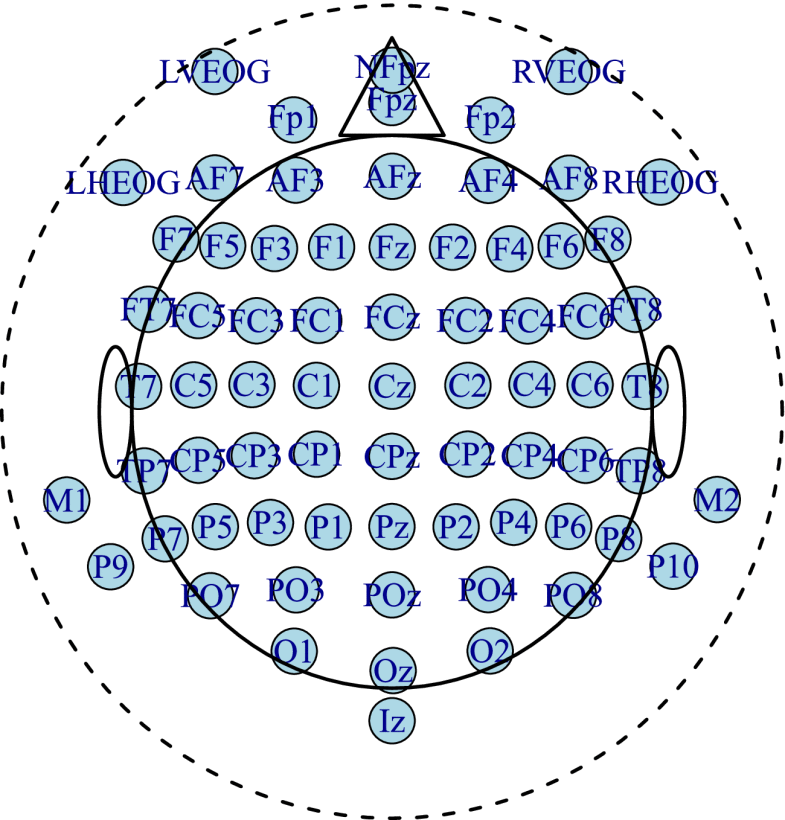}
    \caption{The abbreviations and scalp locations of the 71 EEG channels in the data set are summarized below. The illustration figure is reproduced from \cite{bai2022}.}
    \label{fig:scalp}
\end{figure}

\begin{table}[htbp]
    \centering
    \begin{tabular}{|m{.25\linewidth}|m{.7\linewidth}|}
    \hline
    Eye-Movement: & LHEOG, LVEOG, RVEOG, RHEOG;\\
    \hline
    Front: & NFpz, Fp1 - Fp2, AF7 - AF8, F7 - F8;\\
    \hline
    Central: & FC1, FCz, FC2, C1, Cz, C2, CP1, CPz, CP2;\\
    \hline
    Central-Left: & FT7, FC5, FC3, T7, C5, C3, M1, TP7, CP5, CP3;\\
    \hline
    Central-Right: & FT8, FC6, FC4, T8, C6, C4, M2, TP8, CP6, CP4;\\
    \hline
    Posterior: & P9 - P10, PO7 - PO8, O1 - O2, Iz.\\
    \hline
    \end{tabular}
    \caption{The 71 EEG channels are classified into the 5 classes for collective analysis and inferences.}
    \label{tab:channels}
\end{table}

To delineate further differences, selected channels from \autoref{tab:channels} are examined in greater detail. Specifically, boxplots of the rescaling factors $\{\widehat{W}_m[i]: m = 1, \dots, M\}$ for these channels are depicted in \autoref{fig:W_box}. It can be seen that the signs of the estimated rescaling factors are overall stable amongst subjects. Further, in the alpha band, greater variability is exhibited in the EC condition compared to the EO, but also smaller magnitudes. These observations align with findings reported in the literature \citep{chen2008}.

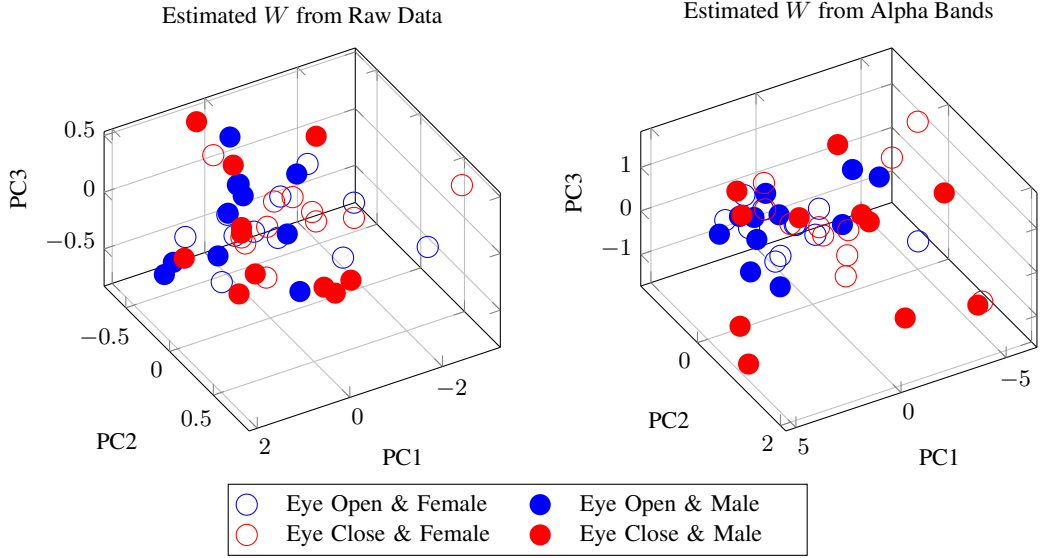
\begin{figure}[htbp]
	\centering
	\begin{tikzpicture}

\pgfplotsset{every axis legend/.append style={
    cells={anchor=west, text width=9em, align=left}, 
    at={(2,0.5)},
    anchor=north,
    legend columns=2,
    legend cell align=left, 
    column sep=1em, 
}}

\pgfplotsset{every axis/.append style={
    legend entries={Eye Open \& Female, Eye Open \& Male, Eye Close \& Female, Eye Close \& Male},
    legend to name=sharedlegend 
}}

\begin{axis}[
    name=plot1, 
    at={(0,0)}, 
    anchor=south west, 
    view={150}{60}, 
    xlabel = PC1, ylabel = PC2, zlabel = PC3,
    z post scale=1.6,
    grid=major,
    title = {Estimated $W$ from Raw Data},
    scatter/classes={%
      EOf={mark=o,blue,mark size=4pt},%
      EOm={mark=*,blue,mark size=4pt, draw opacity=0},%
      ECf={mark=o,red,mark size=4pt},%
      ECm={mark=*,red,mark size=4pt, draw opacity=0}%
    },
    width = .48\textwidth,
]
\addplot3[scatter, only marks, point meta=explicit symbolic]
table[
    x=PC1,
    y=PC2,
    z=PC3,
    meta=meta,
    col sep=comma
]{data/W_RD_PC.csv};
\end{axis}

\begin{axis}[
    at={(.5\textwidth,0)}, 
    anchor=south west,
    view={150}{60},
    xlabel = PC1, ylabel = PC2, zlabel = PC3,
    z post scale=1.6,
    grid=major,
    title = {Estimated $W$ from Alpha Bands},
    scatter/classes={%
      EOf={mark=o,blue,mark size=4pt},%
      EOm={mark=*,blue,mark size=4pt, draw opacity=0},%
      ECf={mark=o,red,mark size=4pt},%
      ECm={mark=*,red,mark size=4pt, draw opacity=0}%
    },
    width = .48\textwidth,
]
\addplot3[scatter, only marks, point meta=explicit symbolic]
table[
    x=PC1,
    y=PC2,
    z=PC3,
    meta=meta,
    col sep=comma
]{data/W_AB_PC.csv};
\end{axis}
\end{tikzpicture}

\begin{tikzpicture}[remember picture]
    \node[anchor=north] at (0, 0.5) {\pgfplotslegendfromname{sharedlegend}};
\end{tikzpicture}
	\caption{{\small The 3-dimensional scatter plot for the principal components of the rescaling effects estimated from the raw data and alpha bands, respectively.}}
	\label{fig:EEG_W}
\end{figure}

\begin{figure}[htbp]
    \centering
    \includegraphics[width=\textwidth]{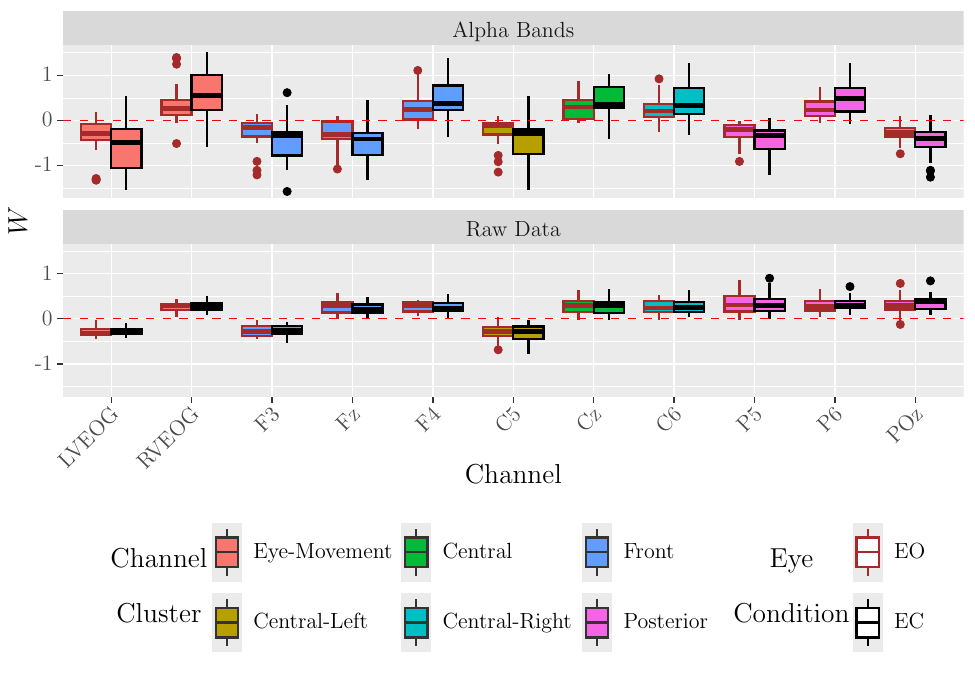}
    \caption{{\small The boxplot of the estimated rescaling effects for selective channels.}}
    \label{fig:W_box}
\end{figure}

Next, we examine the sparse components $\widehat{\mathcal{S}}_m$. For every $(i,j)$-th entry of the $p \times p$ sparse matrices, we calculate the frequency of non-zero values separately for the EO and EC conditions,  for different groups of channels based on their scalp locations. A heatmap illustrating these frequencies is shown in \autoref{tab:channels}. The  results reveal which groups of channels exhibit additional activity, as captured by Granger causal effects from the sparse components $\hat{\mathcal{S}}_m$ on top of the low-rank basis $\hat{\Phi}$. It can be seen from \autoref{fig:sparse_freq} that channels associated with eye movement-related alpha bands help filter out noise, especially for dynamics originating from the frontal channels under the EC condition. The channels in the Eye-Movement cluster, which are related to physical eyeball movements \citep{issa2019a}, are more active during the EO condition, especially in the alpha band data. channels involved in visual processing—primarily located in the posterior scalp region—show more frequent outward connections. These findings are consistent with existing literature \cite[see][for instance]{bai2022}. 

\begin{figure}[htbp]
    \centering
    \includegraphics[width=\textwidth]{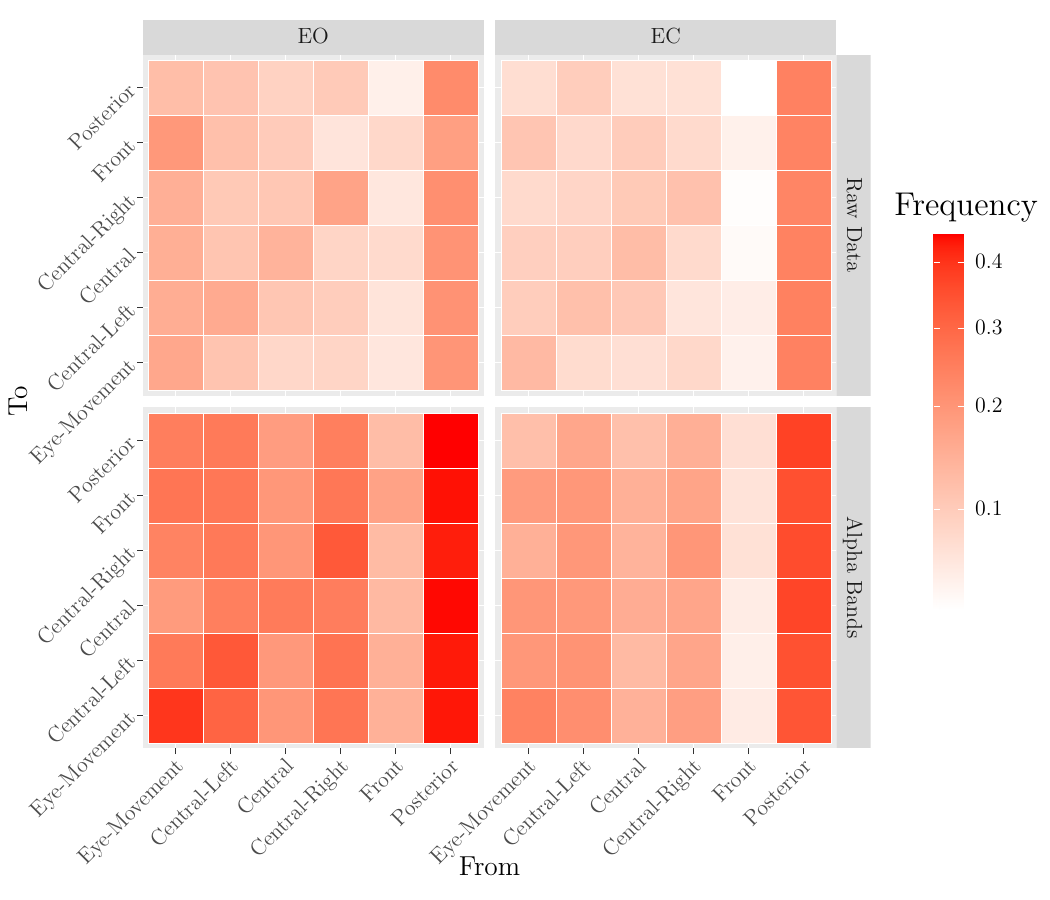}
    \caption{{\small The frequency of nonzero entries in the estimated sparse components $\widehat{\mathcal{S}}$ is summarized according to the channel clusters classified in \autoref{tab:channels}.}}
    \label{fig:sparse_freq}
\end{figure}

\section{Conclusion}\label{sec:conclusion}

There has been growing interest in modeling time series data across multiple entities. One approach organizes the data into a three-dimensional array and applies matrix-variate or tensor-based models, which naturally scale to large numbers of entities ($M$), but often sacrifice interpretability of the estimated parameters. In contrast, our proposed panel VAR framework preserves the classical VAR structure for each entity, capturing inter-entity relationships through structured constraints while allowing substantial heterogeneity across entities. Such flexibility is difficult to achieve with matrix variate or tensor models.

Building on this motivation, this work develop a panel VAR model--LSPVAR--designed to effectively capture both similarities and differences in Granger causal relationships of different entities in a meaningful and easily interpretable manner. The identification conditions for the model parameters lead to a nonsmooth, nonconvex optimization problem, which we address by developing a multi-block ADMM algorithm with established convergence guarantees. Additionally, we prove consistency of the estimators under mild assumptions commonly adopted in single high-dimensional VAR models. Simulation studies with synthetic data and an application to EEG signals demonstrate the effectiveness of the proposed approach.

\begin{appendix}

\section{Vector Auto-Regression (VAR) model}\label{sec:VAR}
Vector Auto-Regression (VAR) is a classical and widely used model for modeling multivariate time series data. Consider the time series observations ${X_t^m : t = 0, \dots, T}$ for the $m$-th subject. The VAR model is expressed as in \eqref{eq:pVAR-model}:
\begin{equation}
    X^m_t = A_m X^m_{t-1} + \epsilon^m_t, ~ \epsilon^m_t \sim N(0, \Sigma_m), \qquad A_m = W_m \Phi + S_m \tag{\ref{eq:pVAR-model}}
\end{equation}
where $A_m \in \mathbb{R}^
{p \times p}$ is the coefficient matrix, and $\Sigma_m \in \mathbb{R}^{p \times p}$ is the covariance matrix for the model innovations.

\begin{rem}
   For analytical simplicity, one may assume the noise covariance matrix is isotropic:
    $$
    \Sigma_m = \sigma_m^2 I_p,
    $$
   which aligns the log-likelihood function closely with the least squares objective used in estimation. Nonetheless, $\Sigma_m$ can take more general forms with complex structure to model dependencies in high-dimensional settings. Our least squares estimation approach remains robust and applicable even under such general covariance structures.
\end{rem}

We next introduce some standard definitions and key properties of VAR models, which will be useful in our subsequent analysis.

Recall that \autoref{ass:stat},
\begin{equation}\label{eq:spectral_bound}
	|\lambda_1(A_m)| < 1,
\end{equation}
ensures a VAR model to be stationary \citep{lutkepohl2005}. The inequality \eqref{eq:spectral_bound} can also be expressed with regard to the VAR model's characteristic polynomial $\mathcal{A}_m(z) = I_p - A_m z$ that $\det(\mathcal{A}_m(z)) \ne 0$ for $|z| \le 1$. As a consequence, the spectral density of the VAR model \eqref{eq:pVAR-model}, defined as
\begin{equation}\label{eq:spec_dens}
	h_m(\vartheta) = \frac{1}{2\pi} (\mathcal{A}_m^{-1}(e^{i\vartheta})) \Sigma_m (\mathcal{A}_m^{-1}(e^{i\vartheta}))^\dag, \quad \vartheta \in [-\pi, \pi],
\end{equation}
is also bounded above.

The following quantities play an essential role in the estimation analysis and will be used to control the consistency rates established in \autoref{sec:cons}:
\begin{align*}
    \Psi(h_{X^m}) & = \sup_{\vartheta \in [-\pi, \pi]} \lambda_1(h_{X^m}(\vartheta)),\\
    \psi(h_{X^m}) & = \sup_{\vartheta \in [-\pi, \pi]} \lambda_p(h_{X^m}(\vartheta)),\\
    \tau_{\max}(\mathcal{A}_m) & = \max_{|z| = 1} \lambda_1(\mathcal{A}^\dag_m(z) \mathcal{A}_m(z)),\\
    \tau_{\min}(\mathcal{A}_m) & = \min_{|z| = 1} \lambda_p(\mathcal{A}^\dag_m(z) \mathcal{A}_m(z)).
\end{align*}
These quantities are also related in VAR models in terms of the following inequalities
$$
\psi(h_{X^m}) \ge \frac{1}{2\pi} \frac{\lambda_p(\Sigma_m)}{\tau_{\max}(\mathcal{A}_m)}, \qquad \Psi(h_{X^m}) \le \frac{1}{2\pi} \frac{\lambda_1(\Sigma_m)}{\tau_{\min}(\mathcal{A}_m)}.
$$

\section{Panel VAR Literature Overview \& Comparisons}\label{sec:dis}

In this section, we provide a supplementary review of existing panel VAR (PVAR) models and compare them with our proposed approach.

A highly restrictive class of models assumes a common transition matrix shared across all entities, as studied extensively by \citet[Section 3.1]{canova2013}, \citet{breitung2015}, \citet{sigmund2021}, among others. These models leverage the entire data pool to estimate a single $p \times p$ matrix (for VAR(1)), typically using classical least squares. While simple and computationally efficient, the effectiveness of such models hinges on the assumption of homogenous dynamics across entities. When this assumption is violated, the model's validity deteriorates, often producing unreliable results.

At the other extreme, one can fully vectorize the panel data and model it with a large $Mp \times Mp$ transition matrix, allowing for maximal flexibility. However, for practical estimation, this flexibility is often curtailed by structural constraints, such as low-rank assumptions \citep[Section 3.2]{canova2013}, or sparsity and group sparsity, as in the Bayesian framework of \citet{korobilis2016}. Despite their flexibility, these approaches sacrifice interpretability at the entity level. For instance, a low-rank structure imposed on the large transition matrix does not necessarily induce meaningful or identifiable patterns in the individual $p \times p$ blocks corresponding to each entity. Moreover, unstructured sparsity can inadvertently produce degenerate cases, such as exactly zero transition matrices for certain entities, leading to unintended random walk behavior.

An intermediate line of work retains individual VAR models for each entity while introducing structured inter-entity dependencies. For example, \citet{skripnikov2019} model the entity-specific matrices $A_1, \dots, A_M$ as sparse perturbations of a shared baseline matrix, while \citet{xu2023} explore commonalities in eigen-structures through post-estimation hypothesis testing. While these models relax strict homogeneity, they remain sensitive to structural assumptions, making them difficult to apply robustly to complex real-world data where such assumptions may not hold.

Our proposed model \eqref{eq:pVAR-model} follows this intermediate philosophy but introduces a more flexible and interpretable structure. Specifically, we link the $M$ transition matrices through a shared low-rank basis, entity-specific rescaling effects, and sparse deviations. This formulation accommodates substantial heterogeneity while preserving clear interpretability at both the entity and panel level. As demonstrated in the main paper and through simulations, this setup strikes a balance between flexibility, interpretability, and robustness for analyzing high-dimensional panel VAR data.

\subsection{Matrix Variate \& Tensor Models}\label{subsec:mar}

An alternative line of research for analyzing panel time series treats the data as a three-dimensional array and applies matrix-variate or tensor decomposition methods. In this section, we compare our proposed model with these approaches, focusing particularly on the matrix auto-regression (MAR) model proposed by \citet{chen2021}, as well as selected tensor factor models; see, for example, \citet{chen2022a} and \citet{babii2023}.

The MAR model extends the classical VAR structure \eqref{eq:pVAR-model} to accommodate panel data in matrix format. Specifically, the panel time series are arranged as a matrix-valued process $Z_t = (X_t^1, \dots, X_t^M) \in \mathbb{R}^{p \times M}$. To illustrate, consider the following MAR(1) example, which is conceptually comparable to the setup proposed in our work:
\begin{equation}\label{eq:mar}
	Z_t = \sum_{i=1}^r B_i^f Z_{t-1} B_i^b + \epsilon_t,
\end{equation}
where $B_i^f \in \mathbb{R}^{p \times p}$, $B_i^b \in \mathbb{R}^{M \times M}$, and $\epsilon_t \in \mathbb{R}^{p \times M}$ are noise terms.

The MAR model can be interpreted as a combination of two VAR models: one operating along the columns and the other along the rows of the matrix-valued process.  In particular, when the coefficient matrices satisfy $B_i^b = I_M$ (identity matrices), the MAR model in \eqref{eq:mar} simplifies to a panel of VAR models where every entity shares a common coefficient matrix $B^f = \sum_{i=1}^r B_i^f$, i.e., $Z_t[\cdot, m] = B^f Z_{t-1}[\cdot, m] + \epsilon_t[\cdot, m]$, which corresponds to the classical VAR formulation in \eqref{eq:pVAR-model} with $A_m \equiv B^f$) for all $m$. Alternatively if $B_i^f = I_p$, the MAR model \eqref{eq:mar} describes a panel of VAR models of dimension $M$, where the $p$ variables are treated as entities. In this case, the shared coefficient matrix is $B^b = \sum_{i=1}^r B_i^b$ leading to $Z_t[i,\cdot] = Z_{t-1}[i,\cdot] B^b + \epsilon_t[i,\cdot]$.

Estimation of the coefficient matrices $B_i^f$ and $B_i^b$ typically follows an alternating regression procedure, where left and right pseudo-inverses are applied iteratively.

Analogous to the MAR setting, tensor factor models have also been explored in the context of panel time series analysis \citep{chen2022a,babii2023}. These models generalize \eqref{eq:mar} by introducing latent factors $F_t \in \mathbb{R}^{q \times N}$ to capture low-dimensional structures, resulting in the formulation:
\begin{equation}\label{eq:tensor_factor}
    Z_t = B^f F_t B^b + \epsilon_t,
\end{equation}
where $B^f \in \mathbb{R}^{p\times q}$ and $B^b \in \mathbb{R}^{N \times M}$ are loading matrices, with $q < p$ and $N < M$ to capture the underlying lower-dimensional structure. Mathematically, this formulation is equivalent to concatenating the matrix-variate data matrices $Z_t$ into a multi-way array $\mathcal{Z} \in \mathbb{R}^{p \times M\times T}$ and then performing inference from a tensor point of view. In this framework,classical tensor decomposition techniques, such as Tucker decomposition or canonical polyadic (CP) decomposition, can be employed to characterize cross-sectional dependencies. For instance, one may write \eqref{eq:tensor_factor} as the tensor decomposition with the outer product algebra
$$
\mathcal{Z} = \mathcal{F} \times_1 B^f \times_2 (B^b)' \times_3 I_T + \epsilon,
$$
where $\times_i$ denotes the usual tensor product of the the mode $i$. Relevant estimation procedures for such tensor models have been discussed in \cite{han2024a,han2024}.

Notably, the tensor factor model \eqref{eq:tensor_factor} does not, by itself, incorporate explicit temporal dynamics. To address this limitation, it is common to impose an MAR structure \eqref{eq:mar} on the latent factors $\mathcal{F}$; see \citet{surana2016}. This leads to a tensor model with auto-regressive latent factors:
\begin{equation}\label{eq:factor_mar}
    Z_t = B^f F_t B^b + \epsilon_t, \qquad F_t = \sum_{i=1}^r A_i^f F_{t-1} A_i^b + \varepsilon_t,
\end{equation}
where $A_i^f$ and $A_i^b$ are autoregressive coefficient matrices of appropriate dimensions. This formulation effectively embeds temporal dependence into the model but does so under a strong assumption: the latent temporal dynamics governing all entities are driven by a shared set of core processes.

Compared to our LSPVAR model \eqref{eq:pVAR-model}, a key distinction is that the MAR and tensor factor models explicitly mix panel data through post-coefficient (loading) matrices $B_i^b$. In contrast, the dependence across entities in our panel VAR framework is implicit, induced by algebraic constraints imposed on the set of transition matrices ${A_m}$, as illustrated in \autoref{fig:model_setup}.

In practice, both the MAR and tensor factor models can be overly restrictive for heterogeneous panels. If the entities exhibit significant heterogeneity, the bilinear structures in \eqref{eq:mar} and \eqref{eq:tensor_factor} may fail to capture the true underlying dynamics, especially for the tensor MAR factor model \eqref{eq:factor_mar}, which enforces that all entities share identical auto-regressive latent processes. In contrast, the proposed LSPVAR model flexibly accommodates heterogeneous structures while preserving interpretability at both the entity and panel levels.

On the other hand, if the panel is truly homogeneous, models like \eqref{eq:mar} and \eqref{eq:tensor_factor} may achieve more efficient recovery by leveraging their stronger mixture structure through pre- and post-loading matrices. However, as the number of entities $M$ grows large relative to $p$ or $T$, estimating the $M \times M$ loading matrices $B_i^b$ in these models can become computationally expensive, particularly under classical least squares approaches. Moreover, when sparsity or low-rank constraints are introduced on $B_i^b$, these models begin to resemble the structure of the LSPVAR model

In summary, the posited model strikes a favorable balance by providing enhanced flexibility and interpretability in heterogeneous settings, making it particularly appealing for complex real-world applications.

\section{Estimation Procedure} \label{sec:solution}

Next, we provide detailed descriptions of the subproblem updates within \autoref{algo:iter}, which iteratively optimizes the proposed objective function.
\begin{equation}\tag{\ref{eq:admm}}
G(\mathcal{W}, \mathcal{S}, \Phi_c, \Phi, \Gamma; \mathcal{X}, \eta, \rho) = F(\mathcal{W}, \mathcal{S}, \Phi; \mathcal{X}, \eta) + \frac{\rho}{2} \|\Phi - \Phi_c\|_F^2 + \rho \langle \Gamma, \Phi - \Phi_c \rangle.
\end{equation}
With slight abuse of notation, we continue to use $G$ to denote the the objective functions for the subproblems, where the precise form and inputs of $G$ may vary depending on the context. The following subsections detail the update routines for each estimator at the $(i+1)$-th iteration.

\subsection[Subproblem for (W, S)]{Subproblem of $(\mathcal{W}, \mathcal{S})$}\label{subsec:WS}
The first block update involves optimizing \eqref{eq:admm} with respect to $(\mathcal{W}, \mathcal{S})$. This subproblem can be viewed as a classical LASSO-type regression with blockwise penalties. Thanks to the blockwise convexity of the objective, the optimization is well-posed and can be efficiently solved using standard algorithms developed for high-dimensional penalized regression.

Moreover, this subproblem naturally decomposes into independent optimizations across entities. Specifically, for the $m$-th entity, the objective function simplifies to:
\begin{equation}\label{eq:sep_m}
    G^{(i+1)}(W_m, S_m) \varpropto \frac{1}{2T} \| Y_m - W_m \Phi^{(i)} X_m - S_m X_m\|_F^2 + \eta \|S_m\|_1,
\end{equation}
and the optimization is separable row-wise with LASSO-type solutions.

For the purpose of proving the sufficient descent property (\autoref{prop:suff_desc}) of an update by optimizing \eqref{eq:sep_m}, we may further split and optimize the two arguments $W_m$ and $S_m$ sequentially, as \eqref{eq:sep_m} is strongly convex with respect to $W_m$ and $S_m$ separately when \autoref{ass:lip_hes} is satisfied.

\subsection[Subproblem for Auxiliary Component]{Subproblem for $\Phi_c$}\label{subsec:phic}
Next, we consider the subproblem of optimizing $\Phi_c$. This task essentially involves finding a feasible point in the intersection of two constraint sets, $\mathbb{L}_p(\hat{r}, \ell)$ and $\mathbb{B}_p$. To address this, we draw inspiration from Dykstra’s algorithm \citep{boyle1986}, which computes the nearest projection onto the intersection by iteratively projecting onto each set.

We implement this via an internal ADMM routine, formulated as a proximal problem with an additional step size parameter $\kappa$, which will be specified later in \autoref{sec:proof_convergence} (see the proof of \autoref{prop:suff_desc}). Specifically,
$$
G^{(i+1)}(\Phi_c) = \frac{\rho}{2} \|\Phi^{(i)} + \Gamma^{(i)} - \Phi_c\|_F^2 + \frac{\kappa \rho}{2} \|\Phi_c - \Phi_c^{(i)}\|_F^2,
$$
and it can be reformulated under the ADMM framework by introducing auxiliary and dual variables as
\begin{equation}\label{eq:G_c}
    G_c(\Phi_B, \Phi_L, \Gamma_{BL}) = \frac{1}{2} \|\Phi_0^{(i)} - \Phi_L\|_F^2 + \frac{1}{2} \|\Phi_B - \Phi_L\|_F^2 + \langle \Gamma_{BL}, \Phi_B - \Phi_L \rangle,
\end{equation}
where $\Phi_0^{(i)} = \frac{\Phi^{(i)} + \Gamma^{(i)} + \kappa \Phi_c^{(i)}}{1 + \kappa}$. Analogously we solve it by blockwise updates. For instance, at the $k$-th iteration,
\begin{enumerate}
    \item Solve $\Phi_L^{(k+1)} = \argmin_{\Phi_L \in \mathbb{L}_p(\hat{r}, \ell)} \|\Phi_L - \frac{1}{2} \left( \Phi_0^{(i)} + \Phi_B^{(k)} + \Gamma_{BL}^{(k)} \right)\|_F^2$. The solution is indeed the projection of $\frac{1}{2} \left( \Phi_0^{(i)} + \Phi_B^{(k)} + \Gamma_{BL}^{(k)} \right)$ onto the set $\mathbb{L}_p(\hat{r}, \ell)$. Explicit form can be found below in \autoref{subsubsec:phi_L}.

    \item Solve $\Phi_B^{(k+1)} = \argmin_{\Phi_B \in \mathbb{B}_p} \|\Phi_B - \left( \Phi_L^{(k+1)} - \Gamma_{BL}^{(k)} \right)\|_F^2$, which is the projection of $\Phi_L^{(k+1)} - \Gamma_{BL}^{(k)}$ onto the set $\mathbb{B}_p$. Details are covered in \autoref{subsubsec:phi_B}.

    \item Update $\Gamma_{BL}^{(k+1)} = \Gamma_{BL}^{(k)} + \Phi_B^{(k+1)} - \Phi_L^{(k+1)}$.
\end{enumerate}

For notational simplicity, we denote by $\Pi_B$ and $\Pi_L$ the projection operators onto the sets $\mathbb{B}_p$ and $\mathbb{L}_p(r, \ell)$, respectively.

\subsubsection[Subproblem for low-rank component]{Subproblem for $\Phi_L$}\label{subsubsec:phi_L}

Note that any $\Phi_L \in \mathbb{L}_p(\hat{r}, \ell)$ can be expressed via its singular value decomposition (SVD) as
\begin{equation}\label{eq:Phi}
    \Phi_L = U D V'
\end{equation}
where $U, V$ are all orthogonal matrices in $\mathbb{R}^{p \times p}$, $D = \diag(d) \in \mathbb{R}^{p \times p}$, $d = (d_1, \dots, d_p) \in \mathbb{R}_+^p$ and $\|d\|_1 = \ell$. Note that here we restrict the parameterization \eqref{eq:Phi} to be at most rank-$\hat r$, then at most $\hat r$ elements of $d$ are potentially non-zero, and hence only the corresponding $\hat r$ columns of $U$ and $V$ matter in our optimization.

With the expression \eqref{eq:Phi} and condition on $\Phi_{L,0}^{(k)} \coloneqq \frac{1}{2} \left( \Phi_0^{(i)} + \Phi_B^{(k)} + \Gamma_{BL}^{(k)} \right)$, the objective function of $\Phi_L = U D V'$ can be viewed as
\begin{equation}\label{eq:Phi_tau}
\begin{aligned}
    & G_c^{(k+1)}(\Phi_L) \varpropto \tr(\Phi_L \Phi_L' - 2 \Phi_{L,0}^{(k)} \Phi_L')\\
    & \quad \Leftrightarrow G_c^{(k+1)}(U, V, D) \varpropto \tr(D^2 - D U' \Phi_{L,0}^{(k)} V).
\end{aligned}
\end{equation}

The optimization can be sequentially decomposed into two main steps, as formalized in the following lemmas, which separately address the updates of the singular vector matrices and the singular values, respectively.

\begin{lem}\label{lem:UV_sol}
    Conditioned on a given matrix $\Phi_{L,0}^{(k)}$, the objective function \eqref{eq:Phi_tau} attains its minimum at $(U, V) = (U_0, V_0)$, if and only if $U_0$ and $V_0$ are the left and right singular vector matrices of $\Phi_{L,0}^{(k)}$, respectively.
\end{lem}

\begin{proof}
This result corresponds to a well-known inequality frequently encountered in matrix inner product analysis. Nonetheless, we provide here a more intricate proof that offers a different perspective.

The following parameterization for orthogonal matrices facilitates deriving a closed-form solution for this subproblem. Specifically, we consider perturbations of the pair $(U, V)$ around a candidate optimizer $(U_0, V_0)$ using the matrix exponential of skew-symmetric matrices, a popular tool for exploring neighborhoods of orthogonal matrices. The parameterization is given by:
$$
U(k_u; K_U, U_0) = U_0 e^{k_u K_U}, \quad V(k_v; K_V, V_0) = V_0 e^{k_v K_V},
$$
where $(k_u, k_v) \in \mathbb{R}^2$, and $K_U, K_V \in \mathbb{R}^{p \times p}$ are skew-symmetric matrices. Then, the objective function as a variant of \eqref{eq:Phi_tau} is introduced with $\Phi_L = U(k_u; K_U, U_0) D V(k_v; K_V, V_0)'$ plugged in,
$$
G_{U_0, V_0}^{(k+1)}(k_u, k_v, D; K_U, K_V) = \tr(D^2 - 2 D e^{-k_u K_U} U_0' \Phi_{L,0}^{(k)} V_0 e^{k_v K_V}).
$$
This reformulation allows us to analyze the minimizer by examining the behavior of $G_{U_0, V_0}^{(k+1)}$ with respect to small perturbations around $(U_0, V_0)$.

The following claim, which follows naturally from critical point analysis, provides a necessary and sufficient condition for when the optimizer $\Phi_L$ of the subproblem admits the singular vector matrices $U_0$ and $V_0$. Specifically, the objective function
\begin{equation}\label{eq:min}
\begin{aligned}
    G_c^{(k+1)}(U_0, V_0, D) & = G_{U_0, V_0}^{(k+1)}(0, 0, D; K_U, K_V)\\
    & = \tr( D^2 - D V_0' (\Phi_{L,0}^{(k)})' U_0 - D U_0' \Phi_{L,0}^{(k)} V_0 )
\end{aligned}
\end{equation}
attains a (local) minimum of \eqref{eq:Phi_tau} with respect to all orthogonal pairs $(U, V)$, given a fixed $\Phi_{L,0}^{(k)}$, if and only if $(U_0, V_0)$ correspond to the left and right singular vector matrices of $\Phi_{L,0}^{(k)}$.

\begin{claim}\label{claim:orth_opt}
    Conditioned on a given $\Phi_{L,0}^{(k)}$, a pair of matrices $(U_0, V_0)$ corresponds to the left and right singular vectors of $\Phi_L$ as defined in \eqref{eq:Phi}, and $\Phi_L$ optimizes the objective function \eqref{eq:Phi_tau}, if and only if, for any arbitrary skew-symmetric matrix pair $(K_U, K_V)$, the following holds:
    \begin{equation}\label{eq:exp_crit}
        \begin{cases}
        \frac{\partial}{\partial k_u} G_{U_0, V_0}^{(i+1)}(0, 0, D; K_U, K_V) = 0,\\
        \frac{\partial}{\partial k_v} G_{U_0, V_0}^{(i+1)}(0, 0, D; K_U, K_V) = 0.
        \end{cases}
    \end{equation}
\end{claim}

Indeed, solving the system \eqref{eq:exp_crit}, with $\Theta = U_0' \Phi_{L,0}^{(k)} V_0$, we have
$$
\begin{cases}
    D \Theta - \Theta' D = 0,\\
    D \Theta' - \Theta D = 0.
\end{cases}
$$
The above system implies that the only valid solution requires the matrix $\Theta$ to be diagonal. Hence, the matrices $U_0$ and $V_0$ are the left and right singular vectors of $\Phi_{L,0}^{(k)}$, respectively.
\end{proof}

Based on \autoref{lem:UV_sol}, set $\Theta = \diag(\theta) = U_0' \Phi_{L,0}^{(k)} V_0$ with $\theta = (\theta_1, \dots, \theta_p)$ being the singular values of $\Phi_{L,0}^{(k)}$, recall that $D = \diag(d)$, then \eqref{eq:Phi_tau} can be derived as
\begin{equation}\label{eq:quad}
    G_c^{(k+1)}(U_0, V_0, D) \varpropto \sum_{i=1}^r (d_i^2 - 2\theta_i d_i), \quad \text{subject to } \|d\|_1 = \ell, ~ \|d\|_0 \le \hat r, ~ d_i \ge 0.
\end{equation}
The optimization can be efficiently handled using classical quadratic programming techniques over a simplex \citep{frank1956}. In fact, the following lemma provides a closed-form solution to \eqref{eq:quad}.

\begin{lem}\label{lem:lambda_sol}
    Define the cumulative sum-type function as
    $$
    \mathnormal{CS}_\theta(j) = \frac{\sum_{k=1}^j \theta_k - \ell}{j}, \quad j = 1, \dots, \hat{r}.
    $$
    Then, the optimizer to \eqref{eq:quad} at the $(k+1)$-th iteration is given by
    $$
    \hat d = \mathbbm{P}_\gamma(\theta) - \mathnormal{CS}_\theta(\gamma) \cdot (\mathbf{1}_\gamma', 0)', \quad \text{where } \gamma = \max \{ j: \mathnormal{CS}_\theta(j) \le \theta_j, j = 1, \dots, \hat r \}.
    $$
    Here $\mathbbm{P}_\gamma(\theta)$ denotes the projection of the vector $\theta$ onto the subspace spanned by its first $\gamma$ elements .
\end{lem}

\begin{rem}\label{rem:choice_r}
    Note that when we input the extreme case $\hat r = p$, the above solution indicates that the quadratic programming may automatically induce sparsity on $\hat{d}$, since it is likely that $\gamma < p$. Hence, instead of tuning the hyper-parameter $\hat r$, the true rank $r$ may be estimated as $\gamma$ by our algorithm adaptively from the data when an over-specified $\hat r$ is used. Further discussion on this adaptive behavior is provided in our consistency analysis; see \autoref{rem:error} and corresponding simulation results in \autoref{subsec:rank}.
\end{rem}

In summary, the results from the preceding lemmas lead to the following proposition, which characterizes the solution to the subproblem for $\Phi_L$, that is, the projection $\Pi_L(\Phi_{L,0}^{(k)})$. 

\begin{prop}\label{prop:sub_Phi}
    Given $\Phi_{L,0}^{(k)}$, the subproblem \eqref{eq:Phi_tau} admits a unique minimizer $\Phi_L^{(k+1)}$ of the form
    \begin{equation}
        \Phi_L^{(k+1)} = U_0 \hat D V_0',
    \end{equation}
    where $U_0$ and $V_0$ are the matrices containing the left and the right singular vectors of $\Phi_{L,0}^{(k)}$, respectively, and $\hat D = \diag(\hat{d})$ is given by the solution in \autoref{lem:lambda_sol}.
\end{prop}

\begin{proof}

It follows directly from \Cref{lem:UV_sol} and \Cref{lem:lambda_sol} that \autoref{prop:sub_Phi} holds. Therefore, it remains to prove \Cref{lem:lambda_sol}.

To this end, assume that $U_0$ and $V_0$ are full $p \times p$ matrices containing the singular vector of $\Phi_{L,0}^{(k)}$, and the low-rank structure of $\Phi_L$ is induced by sparsifying the diagonal matrix of singular values $D$. The following claim characterizes how this low-rank structure influences the optimization in \eqref{eq:quad}. 

\begin{claim}\label{claim:minrank}
    The rank-$\hat{r}$ minimizer $D$ of \eqref{eq:quad} (with $\hat{r} \le p$) has nonzero entries only on the diagonal of its leading principal $\hat{r} \times \hat{r}$ submatrix, denoted by $D_0$.
\end{claim}

\autoref{claim:minrank} can be verified by contradiction: if there exists a rank-$r$ minimizer with a nonzero diagonal element outside $D_0$, one can construct an alternative $D$ by moving this nonzero element into $D_0$, without increasing the objective in \eqref{eq:quad}, leveraging the non-increasing property of the diagonal entries of $\Theta$.

Given \autoref{claim:minrank}, \Cref{lem:lambda_sol} follows by sequentially allocating the total nuclear norm budget $\ell$ to the diagonal entries of $\hat{d}$ in descending order, until the threshold specified in \Cref{lem:lambda_sol} is reached.
\end{proof}

\subsubsection[Subproblem of balanced component]{Subproblem of $\Phi_B$}\label{subsubsec:phi_B}

The corresponding objective function has the form
$$
G_c^{(k+1)}(\Phi_B) \varpropto \tr(\Phi_B \Phi_B' - 2 \Phi_{B,0}^{(k)} \Phi_B'),
$$
where $\Phi_{B,0}^{(k)} = \Phi_L^{(k+1)} - \Gamma_{BL}^{(k)}$.

Next, if we assume that the rows of $\Phi_B \in \mathbb{B}_p$ have norms equal to $K$, then
$$
\begin{aligned}
    G_c^{(k+1)}(\Phi_B) & \varpropto p K^2 - 2 \sum_{i=1}^p e_i' \Phi_{B,0}^{(k)} \Phi_B' e_i\\
    & \ge p K^2 - 2 K \sum_{i=1}^p \|e_i' \Phi_{B,0}^{(k)}\|,
\end{aligned}
$$
with the equality obtained when there exist positive constants $k_i > 0$ such that $e_i' \Phi_B =  k_i e_i' \Phi_{B,0}^{(k)}$ for $i = 1, \dots, p$. The optimal $K$ can then be derived as $K = \frac{1}{p} \sum_{i=1}^p \|e_i' \Phi_{B,0}^{(k)}\|$, the average norm of the rows in the matrix $\Phi_{B,0}^{(k)}$. Note that the derivation requires that $e_i' \Phi_{B,0}^{(k)} \ne 0$ for all $i$. The solution leads to the following proposition introducing the property of the projection operator $\Pi_B$ onto $\mathbb{B}_p$.

\begin{prop}\label{prop:sub_phib}
    For arbitrary $x \in \mathbb{R}^{p \times p}$ that $e_i x \ne 0$ for all $i$, there exists a positive definite diagonal matrix $A = \diag(a_1, \dots, a_p)$ with $\tr(A) = p$, such that $y_0 = \Pi_B(x)$ and $x = A y_0$.
\end{prop}

\subsubsection{Convergence Analysis}

For the convergence of this subproblem, we provide an analogous roadmap for proving \autoref{thm:converge} below.
\begin{enumerate}
    \item The objective function $G_c$ exhibits sufficient descent in every iteration (\autoref{prop:suff_desc_dykstra}).
    \item The evaluations of the objective function $G_c$ are bounded below (\autoref{prop:bounded_dykstra}).
    \item The norm of the subgradients are convergent to zero.
    \item The objective function $G_c$ satisfies the KL-property.
\end{enumerate}

\begin{prop}\label{prop:suff_desc_dykstra}
    Let the subproblem for $\Phi_L$ in \autoref{subsubsec:phi_L} adopt the proximal setup, and with the notations from \autoref{prop:sub_phib}, assume that $a_i$ (obtained from the projection of $\Phi_{B,0}^{(k)}$) are uniformly lower bounded. Then, there exists a constant $\mathcal{C} > 0$ such that
    $$
    G_c(\Phi_B^{(k)}, \Phi_L^{(k)}, \Gamma_{BL}^{(k)})  -  G_c(\Phi_B^{(k+1)}, \Phi_L^{(k+1)}, \Gamma_{BL}^{(k+1)}) \ge \mathcal{C} \left( \|\Phi_B^{(k)} - \Phi_B^{(k+1)}\|_F^2 + \|\Phi_L^{(k)} - \Phi_L^{(k+1)}\|_F^2 \right).
    $$
\end{prop}

Indeed, \autoref{prop:suff_desc_dykstra} is a straightforward corollary of the proximal setup of the $\Phi_L$'s subproblem and the three-point property of $\mathbb{B}_p$ (as stated in \autoref{prop:3pp} in \autoref{sec:constraint}), particularly given that the $a_i$'s of the matrices $\Phi_{B,0}^{(k)}$ are bounded below according to the assumption in \autoref{prop:3pp}.

\begin{prop}\label{prop:bounded_dykstra}
    The evaluations of the objective function $G_c$ \eqref{eq:G_c} at the sequence of estimators $(\Phi_B^{(k)}, \Phi_L^{(k)}, \Gamma_{BL}^{(k)})$ are lower bounded.
\end{prop}

\begin{proof}
    We note that the objective function satisfies
    $$
    \begin{aligned}
        G_c^{(k)} & \coloneqq \frac{1}{2} \|\Phi_0^{(i)} - \Phi_L^{(k)}\|_F^2 + \frac{1}{2} \|\Phi_B^{(k)} - \Phi_L^{(k)}\|_F^2 + \langle \Gamma_{BL}^{(k)}, \Phi_B^{(k)} - \Phi_L^{(k)} \rangle\\
        & \ge \frac{1}{2} \|\Phi_0^{(i)} - \Phi_L^{(k)}\|_F^2 + \frac{1}{2} \|\Phi_B^{(k)} - \Phi_L^{(k)}\|_F^2 - \frac{1}{2} \|\Gamma_{BL}^{(k)}\|_F^2 - \frac{1}{2} \|\Phi_B^{(k)} - \Phi_L^{(k)}\|_F^2\\
        & \ge \frac{1}{2} \|\Phi_0^{(i)} - \Phi_L^{(k)}\|_F^2 - \frac{1}{2} \|\Gamma_{BL}^{(k)}\|_F^2.
    \end{aligned}
    $$
    
    Then it suffices to show that $\Gamma_{BL}^{(k)}$ is upper bounded. We prove it by induction. Assume there is a constant $\Upsilon > 0$ that $\|\Phi_L^{(k+1)} - \Gamma_{BL}^{(k)}\|_F \le \Upsilon$, then $\|\Gamma_{BL}^{(k+1)}\|_F \le \sqrt{\frac{p-1}{p}} \Upsilon$, and hence noting that $\|\Phi_L^{(k+2)}\|_F \le \ell$, we have
    $$
    \|\Phi_L^{(k+2)} - \Gamma_{BL}^{(k+1)}\|_F \le \ell + \sqrt{\frac{p-1}{p}} \Upsilon \le \Upsilon,
    $$
    given any $\Upsilon \ge 2 p \ell$. Note that when $k = 0$, $\|\Phi_L^{(k+1)} - \Gamma_{BL}^{(k)}\|_F = \|\Phi_L^{1}\|_F \le \ell$, we can simply choose $\Upsilon = 2p\ell$, and hence the induction is valid.
\end{proof}

The KL-property of the two sets will be established later in \autoref{sec:proof_convergence} in the proof of \autoref{thm:converge}. Hence, as a result, the subproblem is convergent irrespective of the initialization. Indeed, with the sufficient descent property and our choice of initialization $\Phi_L^{(0)} = \Phi_B^{(0)} = \Phi_c^{(i)}$, use the convergence result to set $\Phi_L^{(\infty)} = \Phi_B^{(\infty)} = \Phi_c^{(i+1)}$. We then obtain 
$$
\frac{\rho}{2} \|\Phi^{(i)} + \Gamma^{(i)} - \Phi_c^{(i)}\|_F^2 \ge \frac{\rho}{2} \|\Phi^{(i)} + \Gamma^{(i)} - \Phi_c^{(i+1)}\|_F^2 + \frac{\kappa \rho}{2}\|\Phi_c^{(i)} - \Phi_c^{(i+1)}\|_F^2,
$$
which will be used in the proof of \autoref{prop:suff_desc}.

\subsection[Subproblem of Phi]{Subproblem of $\Phi$}\label{subsec:phi}
The next primal subproblem focuses on updating the matrix $\Phi$. The corresponding objective function with respect to $\Phi$ is 
$$
\begin{aligned}
    G^{(i+1)}(\Phi) & \varpropto \tr( \rho \Phi \Phi' + \sum_{m=1}^M \frac{1}{T} W_m^{(i+1)} \Phi X_m X_m' \Phi' W_m^{(i+1)})\\
    & \qquad - 2 \tr( \Phi' (\rho (\Phi_c^{(i+1)} - \Gamma^{(i)}) + \sum_{m=1}^M \frac{1}{T} W_m^{(i+1)} (Y_m - S_m^{(i+1)} X_m) X_m' )).
\end{aligned}
$$
The first order optimality condition yields
$$
\rho \Phi + \sum_{m=1}^M (W_m^{(i+1)})^2 \Phi \frac{X_m X_m'}{T} = \rho (\Phi_c^{(i+1)} - \Gamma) + \sum_{m=1}^M \frac{1}{T} W_m^{(i+1)} (Y_m - S_m^{(i+1)} X_m) X_m'.
$$
This equation is row-wise separable and admits explicit solutions for each row, given the other variables $(\mathcal{W}, \mathcal{S}, \Phi_c, \Gamma)$ fixed.

\subsection[Dual Ascent Update of Multiplier]{Dual Ascent Update of $\Gamma$}

Finally, the update of the dual variable $\Gamma$ is given by
\begin{equation}\label{eq:dual_update}
    \Gamma^{(i+1)} \mapsto \Gamma^{(i)} + (\Phi^{(i+1)} - \Phi_c^{(i+1)}).
\end{equation}

\section{Constraint Space Properties}\label{sec:constraint}

In this section, we discuss properties of the constraint spaces, with a particular focus on the intersection $\mathbb{B}_p \cap \mathbb{L}_p(r, \ell)$. Specifically, \autoref{subsec:convex-like} highlights properties that make 
 $\mathbb{B}_p$ behave similarly to a convex space, while \autoref{subsec:tangent_normal} examines the tangent spaces and normal cones of both  $\mathbb{B}_p$ and $\mathbb{L}_p(r, \ell)$ characterizing their intersection.

\subsection[Convex-like Property]{Convex-like Property of $\mathbb{B}_p$}\label{subsec:convex-like}
We begin by presenting the three-point property of the set $\mathbb{B}_p$ \cite[see][Assumption 1]{zhu2019}. The following proposition states this property, and its proof follows straightforwardly. This three-point property exhibits characteristics similar to those of (strongly) convex sets in certain respects.

\begin{prop}[Three-point Property]\label{prop:3pp}
   Under the conditions of \autoref{prop:sub_phib}, if there exists a constant $0 < k_B \le 1$ such that $a_i \ge k_B > 0$, then the set $\mathbb{B}_p$ satisfies the three-point property at $x$. Specifically, for arbitrary $y \in \mathbb{B}_p$,
    $$
    \|x - y\|_F^2 - \|x - y_0\|_F^2 \ge k_B \|y - y_0\|_F^2.
    $$
\end{prop}

\begin{proof}
    For an arbitrary $y \in \mathbb{B}_p$,
    $$
    \begin{aligned}
        \|x - y\|_F^2 - \|x - y_0\|_F^2 & = \|y\|_F^2 - \|y_0\|_F^2 - 2 \langle x, y - y_0 \rangle\\
        & = \|y\|_F^2 + \langle y_0, (2A - I_p) y_0 \rangle - 2 \langle A y_0, y \rangle\\
        & = \|y\|_F^2 + \|y_0\|_F^2 - 2 \langle A y_0, y \rangle\\
        & = \|y\|_F^2 + \|y_0\|_F^2 - \frac{2}{p} \sum_{i=1}^p a_i \|y\|_F \|y_0\|_F \cos(\omega_i),
    \end{aligned}
    $$
    where $\omega_i$ is the angle of the $i$-th rows of the two matrices $y_0$ and $y$. In order to prove the property with some constant $k$ (to be determined), it suffices to show that with $0 < k \le 1$,
    $$
    (1 - k) \left( \|y\|_F^2 + \|y_0\|_F^2 \right) - \frac{2}{p} \|y\|_F \|y_0\|_F \sum_{i=1}^p (a_i - k) \cos(\omega_i) \ge 0
    $$
    for all $y$ and $y_0$ satisfying our assumptions. Some algebra then implies that
    $$
    k \le \min \left(\frac{p - \sum_{i=1}^p a_i \cos(\omega_i)}{p - \sum_{i=1}^p \cos(\omega_i)}, \frac{p + \sum_{i=1}^p a_i \cos(\omega_i)}{p + \sum_{i=1}^p \cos(\omega_i)}\right).
    $$
    Now given that $\sum_{i=1}^p a_i = p$, $a_i \ge k_B$ and $\cos(\omega_i) \in [-1, 1]$, we can then derive that the right hand side of the above inequality is lower bounded by $k_B$, and hence the three point property holds with the constant $k = k_B$.
    
\end{proof}

\begin{rem}
    Note that \autoref{prop:3pp} plays a crucial role in proving \autoref{prop:suff_desc_dykstra} by setting $z = \Phi_B^{(k)}$ and $y = \Phi_B^{(k+1)}$. Further, \autoref{prop:3pp} leads to the following inequality,
    $$
    \langle x - y_0, y - y_0 \rangle \le \frac{1-k_B}{2} \|y - y_0\|_F^2.
    $$
   Although this condition is still weaker than full convexity, it effectively controls the dual ascent steps in the algorithm described in \autoref{subsec:phic}.
\end{rem}

\subsection{Tangent Spaces \& Normal Cones of \texorpdfstring{$\mathbb{B}_p$ and $\mathbb{L}_p(r, \ell)$}{Constraint Spaces}} \label{subsec:tangent_normal}

Assume there exists a matrix $\Phi \in \mathbb{B}_p \cap \mathbb{L}_p(r, \ell)$, where
$$
\Phi = U D V' = (\Phi_1', \dots, \Phi_p')'
$$
is of rank-$\tilde{r}$ ($\tilde{r} \le r$) with $U, V \in \mathbb{R}^{p\times\tilde{r}}$ being singular vector matrices and $D = \diag(d_1, \dots, d_{\tilde{r}})$ storing the singular values that $\sum_{i=1}^{\tilde{r}} d_i = \ell$. Below we analyze the tangent spaces and the normal cones of the two spaces at the intersection $\Phi$.

Regarding the space $\mathbb{B}_p$, we can derive from its definition that the tangent space and normal cone can be expressed as
$$
\begin{aligned}
    \mathcal{T}_{\mathbb{B}_p}(\Phi) & = \{\Delta\Phi = (\Delta\Phi_1', \dots, \Delta\Phi_p')': \Phi_i \Delta\Phi_i' = \Phi_j \Delta\Phi_j', ~ \forall i \ne j \}.\\
    \mathcal{N}_{\mathbb{B}_p}(\Phi) & = \{ H = \diag(b_1, \dots, b_p) \Phi: \sum_{i=1}^p b_i = 0\}.
\end{aligned}
$$

For the low-rank component, we draw on the works of \cite{schneider2015,hosseini2019,li2023} and note that the stratified space
$$
\mathbb{L}_p^{\tilde{r}}(\ell) = \{ \Phi: \rank(\Phi) = \tilde{r}, ~ \|\Phi\|_* = \ell, ~ \lambda_i(\Phi \Phi') \ne \lambda_j(\Phi \Phi'), ~ i \ne j\}
$$
is a smooth manifold as it restricts all the singular values to be away from zero and of order 1. In addition, define
$$
\mathbb{L}_p^{\tilde{r}}(\ell, U, V) = \mathbb{L}_p^{\tilde{r}}(\ell) \cap \{\Phi: U' \Phi V \text{ diagonal} \}.
$$
Then, the tangent space and normal cone can be derived with the following form,
$$
\begin{aligned}
    \mathcal{T}_{\mathbb{L}_p^{\tilde{r}}(\ell)}(\Phi) & = \{\Delta\Phi: U_\perp \Delta\Phi V_\perp' = 0, ~ \tr(U' \Delta \Phi V) = 0\} \coloneqq \mathcal{T} (U, V),\\
    \mathcal{N}_{\mathbb{L}_p^{\tilde{r}}(\ell)}(\Phi) & = \{H = U_\perp E V_\perp': E \in \mathbb{R}^{(p-\tilde{r}) \times (p - \tilde{r})}\} \oplus U V' \coloneqq \mathcal{N}(U, V) \oplus UV',
\end{aligned}
$$
where the orthogonal complement $U_\perp$ and $V_\perp$ can be arbitrary as long as $U' U_\perp = V' V_\perp = 0$. Finally, note that $\mathbb{L}_p^{\tilde{r}}(\ell) \subset \mathbb{L}_p(r,\ell)$, the tangent space and normal cone of $\mathbb{L}_p(r, \ell)$ at $\Phi$ can be derived as
$$
\begin{aligned}
    \mathcal{T}_{\mathbb{L}_p(r,\ell)}(\Phi) & = \{\Delta\Phi + H - k UV': \Delta\Phi \in \mathcal{T}(U, V), ~ H \in \mathcal{N}(U, V), ~ \rank(H) \le r - \tilde{r}, ~ \|H\|_* = \tilde{r}k\}.\\
    \mathcal{N}_{\mathbb{L}_p(r,\ell)}(\Phi) & = \{H \in \mathcal{N}(U, V): \rank(H) \le p - r\} \oplus UV'.
\end{aligned}
$$
Note that the closest rank-$\tilde{r}$ approximation to any rank-$r$ ($r > \tilde{r}$) matrix preserves the singular vector structure. Therefore, any sequence satisfies $\Phi_n \to \Phi$ ($n \to \infty$) with $\Phi_n \in \mathbb{L}_p^{r_n}(\ell)$ ($r_n > \tilde{r}$),  one can design $\delta\ell_n > 0$ and $\delta\ell_n \to 0$, such that
$$
\Phi_n - (\underline{\Phi}_{n} + H_{n}) \to 0, ~ \text{where } \begin{cases}
    \underline{\Phi}_{n} \in \mathbb{L}_p^{\tilde{r}} (l - \delta\ell_n, U, V),\\
    H_n \in \mathcal{N}(U, V), ~ \rank(H_n) = r_n - \tilde{r}, ~ \|H_n\|_* = \delta\ell_n.
\end{cases}
$$
Further, note that for $\underline{\Phi}_{n} \in \mathbb{L}_p^{\tilde{r}} (l - \delta\ell_n, U, V)$,
$$
\delta\Phi_n \coloneqq \Phi - \underline{\Phi}_n = U \delta D_n V' = k_n U V' + U (\delta D_n - k_n I_{\tilde{r}}) V' \subset U V' \oplus \mathcal{T}_{\tilde{r}}(U, V) 
$$
with $k_n = \frac{\delta\ell_n}{\tilde{r}}$ and $\tilde{r} k_n = \|H_n\|_*$, then
$$
\Delta \Phi_n \coloneqq H_n - \delta\Phi_n \in \mathcal{T}_{\mathbb{L}_p(r,\ell)}(\Phi), ~ (\Phi_n - \Phi) - \Delta\Phi_n \to 0.
$$
The other direction is rather straightforward.

In summary, if the left singular vectors of $\Phi$ in $U$ are not standard basis vectors, we can then verify that
$$
\mathcal{N}_{\mathbb{L}_p(r,\ell)}(\Phi) \cap \mathcal{N}_{\mathbb{B}_p}(\Phi) = \{0\},
$$
which implies that $\Phi$ is a linearly regular intersection of the two sets. This result suggests that any alternating projection or averaged projection algorithms exhibit local convergence.

\section{Rank-Sparsity Incoherence}\label{sec:inco}

To ensure identifiability between the basis $\Phi$ and the sparse components $\mathcal{S}$ in \eqref{eq:pVAR-model}, certain incoherence conditions are required. Following \citet{hsu2011}, we consider the following ones.

Assume the setup \eqref{eq:pVAR-model}, where each coefficient matrix decomposes as  $A_m = W_m \Phi + S_m$.  Corresponding to this decomposition, define the following matrix spaces:
\begin{align*}
    \Omega_m & = \{ S \in \mathbb{R}^{p \times p}: \supp(S) \subset \supp(S_m) = \supp(W_m^{-1} S_m) \}, \quad m = 1, \dots, M\\
    \varPhi & = \{ \Phi \in \mathbb{R}^{p \times p}: \Phi =  \Phi_1 + \Phi_2, \Phi_1 \in \range(\Phi), \Phi_2' \in \range(\Phi') \}.
\end{align*}
The projections of an arbitrary matrix $A \in \mathbb{R}^{p \times p}$ onto the two spaces $(\Omega_m, \varPhi)$ are given by
\begin{align*}
    \Pi_{\Omega_m}(A) & = (A[i,j] \cdot \mathbbm{1}\{ (i,j) \in \supp(S_m) \})_{i,j = 1}^p,\\
    \Pi_{\varPhi}(A) & = U U' A + A V V' - U U' A V V',
\end{align*}
respectively, where $U, V \in \mathbb{R}^{p \times r}$ are the left and right orthonormal singular vector matrices of $\Phi$ respectively, and $r$ is the rank of $\Phi$.

\begin{defn}\label{def:ab}
    For the coefficient matrices $A_m = W_m \Phi + S_m$ for $m = 1, \dots, M$, define the following two quantities:
    \begin{align}
        \mu_m(\varsigma) & = \max\{ \varsigma \|\sign(S_m)\|_{1 \to 1}, \varsigma^{-1} \| \sign(S_m) \|_{\infty \to \infty} \}, \\
        \nu(\varsigma) & =\varsigma^{-1} \| U U' \|_\infty + \varsigma \|V V'\|_\infty + \|U\|_{2\to\infty} \|V\|_{2\to\infty},
    \end{align}
    where $\varsigma > 0$ is a balancing parameter to adjust for disparity between the number of rows and columns.
\end{defn}

Note that the coefficient matrices $A_m$ are all square matrices of dimension $p \times p$, hence $\varsigma$ can typically be set to 1. The first quantity $\mu_m$ measures the number of nonzero entries in the sparse matrix $S_m$, while the second quantity $\nu$ quantifies the sparsity of the low-rank basis $\Phi$. In addition, based on \autoref{def:ab}, we impose the following assumption for estimation guarantee, which will be detailed in the sequel.

\begin{ass}\label{ass:id}
    For arbitrary $m$, $\inf_\varsigma \nu(\varsigma) \mu_m(\varsigma) < 1$.
\end{ass}

\begin{prop}\label{prop:id}
    Given the model \eqref{eq:pVAR-model} and under \autoref{ass:id}, with $\mathcal{W}$ fixed, each coefficient matrix $A_m$ admits a unique decomposition into the pair $(\Phi, S_m)$ for $m = 1, \dots, M$.
\end{prop}

\begin{rem}\label{rem:inco_verify}
    We do not impose a specific random sparsity model on the matrices in $\mathcal{S}$, nor do we assume particular properties on the singular vectors of the low-rank matrix $\Phi$. Therefore, it is challenging to derive explicit estimation constraints directly from \autoref{ass:id}. To better understand how \autoref{algo:iter}’s performance relates to the intrinsic low-rank plus sparse structure, we empirically examine violations of \autoref{ass:id} in cases where the estimation results are suboptimal.
\end{rem}

The key idea behind the proof of \autoref{prop:id} is to show that, under \autoref{ass:id}, the intersection set satisfies $\Omega_m \cap \varPhi = \{0\}$ for every $m$. In particular, assume there is $M \in \Omega_m \cap \varPhi$, then
$$
\Pi_{\Omega_m} \big( \Pi_{\varPhi} (M) \big) = M,
$$
and using the triangular inequality, we obtain
$$
\|M\| \le \mu_m(\varsigma) \nu(\varsigma) \|M\|.
$$
By \autoref{ass:id}, the product $\mu_m(\varsigma) \nu(\varsigma) <1$, which forces $\|M\| = 0$.

\begin{rem}
Despite the discussion in \autoref{rem:inco_verify}, we emphasize that \autoref{prop:id} provides only a sufficient condition for unique decomposition.
Therefore, the manual verification approach suggested in \autoref{rem:inco_verify} should be viewed as a heuristic or supplementary diagnostic tool, rather than a definitive criterion, especially when the algorithm exhibits poor performance.
\end{rem}

\section{Multi-Lag LSPVAR Extension}\label{sec:multi_lag}

A natural extension of the LSPVAR model \eqref{eq:pVAR-model} to higher-order lags can be handled using analogous methods. For instance, consider the lag-$q$ model
$$
X_t^m = \sum_{i=1}^q A_{m,i} X^m_{t-i} + \epsilon_t^m, \qquad A_{m,i} = W_{m,i} \Phi_i + S_{m,i}, ~ \forall ~ i = 1, \dots, q,
$$
where the transition matrices $A_{m,i}$, satisfy the same type of constraints as in \eqref{eq:pVAR-model} for each lag $i$.

In the simplest case where $\Phi_1 = \cdots = \Phi_q$, the objective function and estimation procedure remain unchanged. If the low-rank basis matrices $\Phi_i$ differ across lags, the optimization framework extends naturally by augmenting the objective with terms for each lag:
$$
G = F + \rho \sum_{i=1}^q \left( \frac{1}{2} \|\Phi_i - \bar\Phi_i\|_F^2 + \langle \Gamma_i, \Phi_i - \bar\Phi_i \rangle \right),
$$
where $F$ is the classical least squares objective combined with LASSO penalties on the sparse components $\mathcal{S}$.

For estimation, the algorithmic steps described in \autoref{sec:solution} remain applicable, as the augmented variables $\bar{\Phi}_i$ and dual variables $\Gamma_i$ are separable across lags in the augmented objective. Similarly, the convergence guarantees established in \autoref{sec:optm_conv} extend directly to this higher-lag setting.

\section{Proof of Convergence Result for Algorithm \ref{algo:iter}}\label{sec:proof_convergence}

\begin{proof}[Proof of \autoref{prop:suff_desc}]

    For the primal update of $\mathcal{W}$ and $\mathcal{S}$, \autoref{ass:lip_hes} implies that the relevant objective function for the $m$-th subproblem,
    $$
    G(W_m, S_m) = f_m(W_m, S_m) + \eta \|S_m\|_1 = \frac{1}{2T} \|W_m \Phi X_m + S_m X_m - Y_m\|_F^2 + \eta \|S_m\|_1,
    $$
    is strongly convex with respect to its two arguments, respectively. Since \Cref{ass:rsc,ass:lip_hes} require that the objective function $G$ function $\beta_W$-convex with respect to $W_m$ and $\beta$-convex with respect to $S_m$, it follows that the corresponding updates satisfy
    $$
    \begin{aligned}
        & G(W_m^{(i)}, S_m^{(i)}) - G(W_m^{(i+1)}, S_m^{(i)}) & & \ge & &  \frac{\beta_W}{2} \|W_m^{(i)} - W_m^{(i+1)}\|_F^2,\\
        & G(W_m^{(i+1)}, S_m^{(i)}) - G(W_m^{(i+1)}, S_m^{(i+1)}) & & \ge & &  \frac{\beta}{2} \|S_m^{(i)} - S_m^{(i+1)}\|_F^2.
    \end{aligned}
    $$
    Therefore,
    \begin{multline}\label{eq:desc_WS}
        G(\mathcal{W}^{(i)}, \mathcal{S}^{(i)}, \Phi_r^{(i)}, \Phi^{(i)}, \Gamma_\Phi^{(i)}) - G(\mathcal{W}^{(i+1)}, \mathcal{S}^{(i+1)}, \Phi_r^{(i)}, \Phi^{(i)}, \Gamma_\Phi^{(i)})\\
        \ge \sum_{m=1}^M \frac{\beta_W}{2} \|W_m^{(i)} - W_m^{(i+1)}\|_2^2 + \sum_{m=1}^M \frac{\beta}{2} \|S_m^{(i)} - S_m^{(i+1)}\|_F^2.
    \end{multline}
    
    By optimality and convergence of the subproblem's algorithm in \autoref{subsec:phic}, and noting that the objective function is monotonically decreasing and the algorithm is initialized at $\Phi_c^{(i)}$, we have that
    \begin{equation}\label{eq:desc_phic}
        G(\mathcal{W}^{(i+1)}, \mathcal{S}^{(i+1)}, \Phi_c^{(i)}, \Phi^{(i)}, \Gamma^{(i)}) - G(\mathcal{W}^{(i+1)}, \mathcal{S}^{(i+1)}, \Phi_c^{(i+1)}, \Phi^{(i)}, \Gamma^{(i)}) \ge \frac{\kappa\rho}{2} \|\Phi_c^{(i)} - \Phi_c^{(i+1)}\|_F^2.
    \end{equation}

    For $\Phi$, since \autoref{ass:lip_hes} states that the objective function $F$ is $\beta_\Phi$-convex with respect to $\Phi$, we consequently have
    \begin{multline}\label{eq:phi_descent}
        G(\mathcal{W}^{(i+1)}, \mathcal{S}^{(i+1)}, \Phi_c^{(i+1)}, \Phi^{(i)}, \Gamma^{(i)}) - G(\mathcal{W}^{(i+1)}, \mathcal{S}^{(i+1)}, \Phi_c^{(i+1)}, \Phi^{(i+1)}, \Gamma^{(i)})\\
        \ge \frac{\beta_\Phi + \rho}{2} \|\Phi^{(i)} - \Phi^{(i+1)} \|_F^2,
    \end{multline}
    given $\Phi^{(i+1)}$ is the subpoblem minimizer that satisfies the first order optimality condition. In addition, the first order condition also suggests $\rho \Gamma^{(i+1)} = - F_\Phi(\mathcal{W}^{(i+1)}, \mathcal{S}^{(i+1)}, \Phi^{(i+1)})$, where $F_\Phi$ is the gradient of $F$ with respect to $\Phi$ defined as
    $$
    F_\Phi(\mathcal{W}, \mathcal{S}, \Phi) = \nabla_\Phi F(\mathcal{W}, \mathcal{S}, \Phi; X, \eta) = \sum_{m=1}^M W_m \frac{Y_m X_m'}{T} - \sum_{m=1}^M W_m (W_m \Phi + S_m) \frac{X_m X_m'}{T}.
    $$
    
    Denote $\nabla_\Phi F^{(i)} = F_\Phi(\mathcal{W}^{(i)}, \mathcal{S}^{(i)}, \Phi^{(i)})$. Since \autoref{ass:lip_hes} ensures that $F_\Phi$ is $\alpha_W$-Lipschitz continuous with respect to $W_m$, $\alpha_S$-Lipschitz continuous with respect to $S_m$ and $\alpha_\Phi$-Lipschitz continuous with respect to $\Phi$, we get that
    \begin{equation}\label{eq:dual_ascent}
    \begin{aligned}
        & \| \nabla_\Phi F^{(i)} - \nabla_\Phi F^{(i+1)} \|_F^2 \\
        & \quad \le \left( \alpha_\Phi \|\Phi^{(i)} - \Phi^{(i+1)}\|_F + \sum_{m=1}^M \alpha_W \|W_m^{(i)} - W_m^{(i+1)}\|_F + \sum_{m=1}^M \alpha_S \|S_m^{(i)} - S_m^{(i+1)}\|_F \right)^2\\
        & \quad \le 3 \alpha_\Phi^2 \|\Phi^{(i)} - \Phi^{(i+1)}\|_F^2 + \sum_{m=1}^M 3 M \alpha_W^2 \|W_m^{(i)} - W_m^{(i+1)}\|_F^2 + \sum_{m=1}^M 3 M \alpha_S^2 \|S_m^{(i)} - S_m^{(i+1)}\|_F^2
    \end{aligned}
    \end{equation}
    Hence, for the dual ascent step,
    \begin{equation}\label{eq:asc_gamma}
    \begin{aligned}
            & G(\mathcal{W}^{(i+1)}, \mathcal{S}^{(i+1)}, \Phi_c^{(i+1)}, \Phi^{(i+1)}, \Gamma^{(i)}) - G(\mathcal{W}^{(i+1)}, \mathcal{S}^{(i+1)}, \Phi_c^{(i+1)}, \Phi^{(i+1)}, \Gamma^{(i+1)})\\
            & \quad = - \frac{1}{\rho} \|\nabla_\Phi F^{(i)} - \nabla_\Phi F^{(i+1)}\|_F^2\\
            & \quad \ge - \frac{3 \alpha_\Phi^2}{\rho} \|\Phi^{(i)} - \Phi^{(i+1)}\|_F^2 - \sum_{m=1}^M \frac{3 M \alpha_W^2}{\rho} \|W_m^{(i)} - W_m^{(i+1)}\|_F^2 - \sum_{m=1}^M \frac{3 M \alpha_S^2}{\rho} \|S_m^{(i)} - S_m^{(i+1)}\|_F^2
    \end{aligned}
    \end{equation}
    where the last line uses \eqref{eq:dual_ascent}. The final result then combines \Cref{eq:desc_WS,eq:desc_phic,eq:phi_descent,eq:asc_gamma} and yields
    $$
    \begin{aligned}
        & G(\mathcal{W}^{(i)}, \mathcal{S}^{(i)}, \Phi_c^{(i)}, \Phi^{(i)}, \Gamma^{(i)}) - G(\mathcal{W}^{(i+1)}, \mathcal{S}^{(i+1)}, \Phi_c^{(i+1)}, \Phi^{(i+1)}, \Gamma^{(i+1)})\\
        & \qquad \ge \sum_{m=1}^M \|W_m^{(i)} - W_m^{(i+1)}\|_F^2 (\frac{\beta_W}{2} - \frac{3 M \alpha_W^2}{\rho}) + \sum_{m=1}^M \|S_m^{(i)} - S_m^{(i+1)}\|_F^2 (\frac{\beta}{2} - \frac{3 M \alpha_S^2}{\rho} )\\
        & \qquad \qquad + \|\Phi^{(i)} - \Phi^{(i+1)}\|_F^2 (\frac{\beta_\Phi + \rho}{2} - \frac{3 \alpha_\Phi^2}{\rho}) + \frac{\kappa \rho}{2} \|\Phi_c^{(i)} - \Phi_c^{(i+1)}\|_F^2.
    \end{aligned}
    $$
    Given a fixed data realization, a sufficiently large $\rho$ leads to the sufficient descent property.

    Note that in terms of orders, $\frac{\beta}{2} - \frac{3 M \alpha_S^2}{\rho}$ is essentially of order $O(1)$. Hence, for balanced coefficients, one can select $\kappa$ proportionally as  $\kappa \varpropto \frac{M}{\rho}$. Similarly, imposing  $\frac{\beta_W}{2} - \frac{3 M \alpha_W^2}{\rho} = O(1)$ and $\frac{\beta_\Phi + \rho}{2} - \frac{3 \alpha_\Phi^2}{\rho} = O(M)$ provides practical guidance for selecting $\ell$ accordingly.

\end{proof}

\begin{proof}[Proof of \autoref{prop:conv_func}]
    For bounding $G(\mathcal{W}^{(i+1)}, \mathcal{S}^{(i+1)}, \Phi_c^{(i+1)}, \Phi^{(i+1)}, \Gamma^{(i+1)})$ below, note that
    $$
    \begin{aligned}
        & G(\mathcal{W}^{(i+1)}, \mathcal{S}^{(i+1)}, \Phi_c^{(i+1)}, \Phi^{(i+1)}, \Gamma^{(i+1)})\\
        & \qquad = F(\mathcal{W}^{(i+1)}, \mathcal{S}^{(i+1)}, \Phi^{(i+1)}) + \frac{\rho}{2} \|\Phi^{(i+1)} - \Phi_c^{(i+1)}\|_F^2 + \rho \langle \Gamma^{(i+1)}, \Phi^{(i+1)} - \Phi_c^{(i+1)} \rangle\\
        & \qquad = F(\mathcal{W}^{(i+1)}, \mathcal{S}^{(i+1)}, \Phi^{(i+1)}) + \langle \nabla_\Phi F^{(i+1)}, \Phi_c^{(i+1)} - \Phi^{(i+1)} \rangle + \frac{\rho}{2} \|\Phi^{(i+1)} - \Phi_c^{(i+1)}\|_F^2\\
        & \qquad \ge F(\mathcal{W}^{(i+1)}, \mathcal{S}^{(i+1)}, \Phi_c^{(i+1)}) + \frac{\rho - \alpha_\Phi}{2} \|\Phi^{(i+1)} - \Phi_r^{(i+1)}\|_F^2\\
        & \qquad \ge 0,
    \end{aligned}
    $$
    where the first inequality uses  \autoref{ass:lip_hes}, namely that the partial derivative $F_\Phi$ is $\alpha_\Phi$-Lipschitz continuous with respect to $\Phi$, and the last inequality is due to the selection criteria of $\rho$ in \eqref{eq:rho_phi_lower}. Therefore, the sequence $G(\mathcal{W}^{(i+1)}, \mathcal{S}^{(i+1)}, \Phi_c^{(i+1)}, \Phi^{(i+1)}, \Gamma^{(i+1)})$ is lower bounded.

  Combined with the sufficient descent property established in \autoref{prop:suff_desc}, the sequence is monotonically decreasing and lower bounded, and is therefore guaranteed to converge.
    
\end{proof}

\begin{proof}[Proof of \autoref{prop:conv}]
    According to \Cref{prop:suff_desc,prop:conv_func}, the sequence of estimates, $(\mathcal{W}^{(i+1)}, \mathcal{S}^{(i+1)}, \Phi_c^{(i+1)}, \Phi^{(i+1)}, \Gamma^{(i+1)})$, forms a Cauchy sequence, and hence as $i \to \infty$, there exists a limiting point $(\widehat{\mathcal{W}}, \widehat{\mathcal{S}}, \widehat{\Phi}_c, \widehat{\Phi}, \widehat{\Gamma})$ such that
    $$
    \lim_{s\to\infty} (\mathcal{W}^{(i+1)}, \mathcal{S}^{(i+1)}, \Phi_c^{(i+1)}, \Phi^{(i+1)}, \Gamma^{(i+1)}) = (\widehat{\mathcal{W}}, \widehat{\mathcal{S}}, \widehat{\Phi}_c, \widehat{\Phi}, \widehat{\Gamma}).
    $$

    For the second part, concerning first-order optimality, we will show that the sub-gradient sets of the function $G(\mathcal{W}, \mathcal{S}, \Phi_c, \Phi, \Gamma)$ at $(\widehat{\mathcal{W}}, \widehat{\mathcal{S}}, \widehat{\Phi}, \widehat{\Gamma}_\Phi)$ include the element $0$, so that $(\widehat{\mathcal{W}}, \widehat{\mathcal{S}}, \widehat{\Phi}, \widehat{\Gamma}_\Phi)$ is a critical point. Indeed, given that the function $G$ and its sub-gradient sets are continuous, it suffices to establish the following lemma.

    \begin{lem}\label{lem:subg_bound}
        There exist sequences of sub-gradients of $G$ with respect to the variables $\{\mathcal{W}, \mathcal{S}, \Phi_c, \Phi, \Gamma\}$, evaluated at $(\mathcal{W}^{(i+1)}, \mathcal{S}^{(i+1)}, \Phi_c^{(i+1)}, \Phi^{(i+1)}, \Gamma^{(i+1)})$, that converge to $0$, i.e., for $i\to\infty$,
        $$
        \begin{aligned}
            \frac{\partial G}{\partial w_m}(\mathcal{W}^{(i+1)}, \mathcal{S}^{(i+1)}, \Phi_c^{(i+1)}, \Phi^{(i+1)}, \Gamma^{(i+1)}) ~ \ni ~ & \mathbf{d}_{w_m}^{(i+1)} \to 0, ~ m = 1, \dots, M,\\
            \frac{\partial G}{\partial S_m}(\mathcal{W}^{(i+1)}, \mathcal{S}^{(i+1)}, \Phi_c^{(i+1)}, \Phi^{(i+1)}, \Gamma^{(i+1)}) ~ \ni ~ & \mathbf{d}_{S_m}^{(i+1)} \to 0, ~ m = 1, \dots, M,\\
            \frac{\partial G}{\partial \Phi_c}(\mathcal{W}^{(i+1)}, \mathcal{S}^{(i+1)}, \Phi_c^{(i+1)}, \Phi^{(i+1)}, \Gamma^{(i+1)}) ~ \ni ~ & \mathbf{d}_{\Phi_c}^{(i+1)} \to 0,\\
            \frac{\partial G}{\partial \Phi}(\mathcal{W}^{(i+1)}, \mathcal{S}^{(i+1)}, \Phi_c^{(i+1)}, \Phi^{(i+1)}, \Gamma^{(i+1)}) ~ \ni ~ & \mathbf{d}_\Phi^{(i+1)} \to 0,\\
            \frac{\partial G}{\partial \Gamma_\Phi}(\mathcal{W}^{(i+1)}, \mathcal{S}^{(i+1)}, \Phi_c^{(i+1)}, \Phi^{(i+1)}, \Gamma^{(i+1)}) ~ \ni ~ & \mathbf{d}_\Gamma^{(i+1)} \to 0.
        \end{aligned}
        $$
    \end{lem}
\end{proof}

\begin{proof}[Proof of \autoref{lem:subg_bound}]
    We prove these limits by upper bounding selective sequences of sub-gradients that converge to $0$, and utilize the fact that the sets of subgradients are continuous.

    For $\mathbf{d}_{w_m}^{(i+1)}$, we know that $\tilde{\mathbf{d}}_{w_m}^{(i+1)} = 0 \in \frac{\partial G}{\partial w_m}(\mathcal{W}^{(i+1)}, \mathcal{S}^{(i)}, \Phi_c^{(i)}, \Phi^{(i)}, \Gamma^{(i)})$, hence
    $$
    \begin{aligned}
        \mathbf{d}_{w_m}^{(i+1)} & = \tilde{\mathbf{d}}_{w_m}^{(i+1)} + \sum_{i=1}^p E_i (\Phi^{(i+1)} \frac{X_m X_m'}{T} (\Phi^{(i+1)})' - \Phi^{(i)} \frac{X_m X_m'}{T} (\Phi^{(i)})') E_i w_m^{(i+1)}\\
        & \quad + \sum_{i=1}^p E_i (\Phi^{(i+1)} - \Phi^{(i)}) \frac{X_m (S_m^{(i+1)} X_m - Y_m)'}{T} e_i + \sum_{i=1}^p E_i \Phi^{(i)} \frac{X_m X_m'}{T} (S_m^{(i+1)} - S_m^{(i)})' e_i.
    \end{aligned}
    $$

    Regarding $S_m$, optimality implies that $\tilde{\mathbf{d}}_{S_m}^{(i+1)} = 0 \in \frac{\partial G}{\partial S_m}(\mathcal{W}^{(i+1)}, \mathcal{S}^{(i+1)}, \Phi_c^{(i)}, \Phi^{(i)}, \Gamma_\Phi^{(i)})$, therefore one of the sub-gradients $\mathbf{d}_{S_m}^{(i+1)}$ is
    $$
    \mathbf{d}_{S_m}^{(i+1)} = \tilde{\mathbf{d}}_{S_m}^{(i+1)} + \frac{X_m X_m'}{T} (\Phi^{(i+1)} - \Phi^{(i)})' W_m^{(i+1)}.
    $$

    The fact that the subproblem of $\Phi_c$ converges implies that there exists $H^{(i+1)}$ in the normal cone of $\mathbb{B}_p \cap \mathbb{L}_p(\hat{r}, \ell)$ at $\Phi_c^{(i+1)}$, such that
    $$
    \tilde{\mathbf{d}}_{\Phi_c}^{(i+1)} = H^{(i+1)} + \rho (\Phi_c^{(i+1)} - \Phi^{(i)} - \Gamma^{(i)}) + \kappa \rho (\Phi_c^{(i+1)} - \Phi_c^{(i)}) = 0,
    $$
    and $\tilde{\mathbf{d}}_{\Phi_c}^{(i+1)} \in \frac{\partial G}{\partial \Phi_c}(\mathcal{W}^{(i+1)}, \mathcal{S}^{(i+1)}, \Phi_c^{(i+1)}, \Phi^{(i)}, \Gamma^{(i)})$. Therefore, we have
    $$
    \begin{aligned}
    \mathbf{d}_{\Phi_c}^{(i+1)} & = \tilde{\mathbf{d}}_{\Phi_c}^{(i+1)} - \rho (\Phi^{(i+1)} - \Phi^{(i)} + \Gamma^{(i+1)} - \Gamma^{(i)}) - \kappa \rho (\Phi_c^{(i+1)} - \Phi_c^{(i)})\\
    & = - \rho (\Phi^{(i+1)} - \Phi^{(i)} + \Gamma^{(i+1)} - \Gamma^{(i)}) - \kappa \rho (\Phi_c^{(i+1)} - \Phi_c^{(i)}).
    \end{aligned}
    $$

    The optimality of $\Phi$ implies that $\tilde{\mathbf{d}}_\Phi^{(i+1)} = 0 \in \frac{\partial G}{\partial \Phi}(\mathcal{W}^{(i+1)}, \mathcal{S}^{(i+1)}, \Phi_c^{(i+1)}, \Phi^{(i+1)}, \Gamma^{(i)})$, and
    $$
    \mathbf{d}_\Phi^{(i+1)} = \tilde{\mathbf{d}}_\Phi^{(i+1)} + \rho (\Gamma^{(i+1)} - \Gamma^{(i)}) = \rho_\Phi (\Gamma^{(i+1)} - \Gamma^{(i)}).
    $$

    Finally, for the dual variable, we get
    $$
    \mathbf{d}_\Gamma^{(i+1)} = \rho (\Phi^{(i+1)} - \Phi_c^{(i+1)}) = \rho (\Gamma^{(i+1)} - \Gamma^{(i)}).
    $$

    Then, the limiting behaviors of the gradients' norms are all controlled by norm quantities that converge to 0.
\end{proof}

\begin{proof}[Proof of \autoref{thm:converge}]
    Note that the convergence and first-order optimality of \autoref{algo:iter} follow directly from the sequence of results established in \Cref{prop:suff_desc,prop:conv_func,prop:conv}. It remains to verify that the objective function \eqref{eq:admm} satisfies the Kurdyka–Łojasiewicz (KL) property, ensuring that the convergence of \autoref{algo:iter} holds regardless of initialization. 
    
    If we temporarily disregard the domain constraints $\Phi_c \in \mathbb{B}_p \cap \mathbb{L}_p(\hat r, \ell)$, the objective function \eqref{eq:admm} is evidently semi-algebraic. Therefore, it suffices to establish that both $\mathbb{B}_p$ and $\mathbb{L}_p(\hat r, \ell)$ are semi-algebraic sets. Under this condition, the optimization problem satisfies the KL property, guaranteeing global convergence.

    One can easily verify the semi-algebraic property of the space $\mathbb{B}_p$, as its definition is given by polynomial constraints.

    The space $\mathbb{L}_p(\hat r, \ell)$ can be considered as the intersection of two spaces: a matrix space with rank at most $\hat r$, and a matrix space with fixed nuclear norm $\ell$. Since the former is known to be semi-algebraic, we  establish next the semi-algebraic property of the latter.

    The nuclear norm can be expressed as
    $$
    \|\Phi\|_* = \tr( \sqrt{\Phi \Phi'}).
    $$
    Here the square root of a symmetric (semi-)positive definite matrix $A = Q D Q'$ with diagonal non-negative $D$ is defined as $\sqrt{A} = Q D^{1/2} Q'$ where $D^{1/2}$ takes element-wise square roots of the diagonal entries. This square root operation is a semi-algebraic operator. Thus, since the nuclear norm is a composition of multiple semi-algebraic functions, the set $\mathbb{L}_p(\hat r, \ell)$ is semi-algebraic.
\end{proof}

\section{Proofs of Consistency Results}\label{sec:proof_consistent}

\begin{proof}[Proof of \autoref{prop:upper}]

We first briefly digress to present the following lemma, which suggests that the assumptions on the $\infty$-norm bounds of $\Phi^*$ and $\widehat{\Phi}$ in \autoref{prop:upper} are not restrictive.

\begin{lem}\label{lem:lr_infty}
    Given an arbitrary rank-$r$ matrix $\Phi$ ($r \le p$), whose matrices contain the singular vectors and are sampled uniformly from the set of rank-$r$ partial isometries in $\mathbb{R}^{p\times r}$, there exist positive constants $\varphi$ and $c$, such that
    $$
    \|\Phi\|_\infty \le \varphi \|\Phi\|_2 \frac{\sqrt{r_p}}{p}
    $$
    holds with probability at least $1 - c p^{-3} \log(p)$, where $r_p = \log^2(p) \cdot \max(r, \log(p))$.
\end{lem}

\autoref{lem:lr_infty} can be derived by adapting results from the literature \cite[for example, see][Lemma 2.2]{candes2009}. In particular, we justify the bound on $\widehat{\Phi}$ by arguing that $\widehat{\Phi}$ primarily lies in a subspace of rank at least $r$.

Then, referring to the model setup discussed in \autoref{ass:exp_deg}, and given that $\Phi^* \in \mathbb{B}_p \cap \mathbb{L}_p(r, \ell)$, we can derive from $\| W_m^* \Phi^* \| = O(1)$ that $W \precsim \frac{\sqrt{r p}}{\ell}$.

The following claim is helpful to understand the constants given in \autoref{ass:lip_hes} and the function $F$ itself evaluated at the output $(\widehat{\mathcal{W}}, \widehat{\mathcal{S}}, \widehat{\Phi})$.

\begin{claim}\label{claim:rates}
  Given a data realization $\mathcal{X}$, \Cref{ass:stat,ass:exp_deg}, together with the least squares structure of the subproblems, imply that the output $(\widehat{\mathcal{W}}, \widehat{\mathcal{S}}, \widehat{\Phi})$ from \autoref{algo:iter} satisfies the following: 
  \begin{itemize}
        \item $\frac{\ell^2}{\hat{r}} \precsim \|\widehat{\Phi}\|_F^2 \precsim \frac{\ell^2}{r}$;
        \item $\| \widehat{W}_m \widehat{\Phi} \|^2 = O(1)$, $\| \widehat{S}_m \|^2 = O(1)$;
        \item $\min_j \frac{1}{M} \sum_{m=1}^M \|e_j' \widehat{W}_m \widehat{\Phi}\|^2 = O(1)$.
    \end{itemize}
    Consequently, when evaluated at the true parameters $(\mathcal{W}^*, \mathcal{S}^*, \Phi^*)$ and the algorithm output $(\widehat{\mathcal{W}}, \widehat{\mathcal{S}}, \widehat{\Phi})$, the constants in \autoref{ass:lip_hes} can be tightened to satisfy:
    \begin{itemize}
        \item $\frac{\beta_W}{\beta} = \min\left\{ \min_j \|e_j' \widehat{\Phi}\|^2, \min_j \|e_j' \Phi^*\|^2 \right\} \succsim \frac{\ell^2 }{\hat{r} p}$;
        \item $\frac{\beta_\Phi}{\beta} = \min \left\{ \min_j \sum_{m=1}^M \|e_j' \widehat{W}_m\|^2, \min_j \sum_{m=1}^M \|e_j' W_m^*\|^2 \right\} \succsim \frac{M r p}{\ell^2}$;
        \item $\widehat{W} \precsim \frac{\sqrt{\hat{r} p}}{\ell}$.
    \end{itemize}
\end{claim}

We provide some further justification for \autoref{claim:rates}. The first three bullet points follow directly from the model structure and the least squares formulation of the associated regression subproblems. The latter three can be verified by explicitly expanding the corresponding second-order partial derivatives of the objective function $F$.

Since \autoref{algo:iter} is initialized within a neighborhood of the true parameters $(\mathcal{W}^*, \mathcal{S}^*, \Phi^*)$, the sufficient descent argument from \autoref{prop:suff_desc}, combined with the local smoothness conditions in \autoref{ass:lip_hes}, ensures that the final output$(\widehat{\mathcal{W}}, \widehat{\mathcal{S}}, \widehat{\Phi})$ satisfies
$$
F(\widehat{\mathcal{W}}, \widehat{\mathcal{S}}, \widehat{\Phi}) \le F(\mathcal{W}^*, \mathcal{S}^*, \Phi^*).
$$
Therefore, we obtain
$$
\sum_{m=1}^M \frac{1}{2T} \|Y_m - (\widehat{W}_m \widehat{\Phi} + \widehat{S}_m ) X_m \|_F^2 + \eta \|\widehat{S}_m\|_1 \le \sum_{m=1}^M \frac{1}{2T} \|Y_m - ( W_m^* \Phi^* + S_m^* ) X_m \|_F^2 + \eta \|S_m^*\|_1,
$$
The model implies that $Y_m - (W_m^* \Phi^* + S_m^*) X_m = \varepsilon_m$. Introducing the notation $\Delta_\Phi = \widehat{\Phi} - \Phi^*$, $\Delta_{W_m} = \widehat{W}_m - W_m^*$, and $\Delta_{S_m} = \widehat{S}_m - S_m^*$, we consider pairs of decomposable subspaces $(\mathcal{Q}_m, \mathcal{Q}^\perp_m)$ satisfying$\|S_m\|_1 = \|S_m^{\mathcal{Q}_m}\|_1 + \|S_m^{\mathcal{Q}^\perp_m}\|_1$ and $\mathcal{Q}_m \bigcup \mathcal{Q}^\perp_m = \mathbb{R}^{p \times p}$ for all $m = 1, \dots, M$. Then
$$
\begin{aligned}
    & \sum_{m=1}^M \frac{1}{2T} \|(\Delta_{W_m} \widehat{\Phi} + W_m^* \Delta_\Phi + \Delta_{S_m}) X_m\|_F^2\\
    & \le \sum_{m=1}^M \langle \Delta_{W_m} \widehat{\Phi} + W_m^* \Delta_\Phi + \Delta_{S_m}, - \frac{\varepsilon_m X_m'}{T} \rangle + \sum_{m=1}^M \eta (\|S^*_m\|_1 - \|\widehat{S}_m\|_1)\\
    & \le \sum_{m=1}^M \langle \Delta_{W_m} \widehat{\Phi} + W_m^* \Delta_\Phi + \Delta_{S_m}, - \frac{\varepsilon_m X_m'}{T} \rangle + \sum_{m=1}^M \big(\eta \|\Delta_{S_m}^{\mathcal{Q}_m}\|_1 - \eta \|\Delta_{S_m}^{\mathcal{Q}^\perp_m}\|_1 + 2 \eta \|(S_m^*)^{\mathcal{Q}^\perp_m}\|_1 \big)
\end{aligned}
$$
Note that
$$
\begin{aligned}
    & \sum_{m=1}^M \langle \Delta_{W_m} \widehat{\Phi} + W_m^* \Delta_\Phi + \Delta_{S_m}, - \frac{\varepsilon_m X_m'}{T} \rangle\\
    & \le \sum_{m=1}^M \|\Delta_{W_m} \widehat{\Phi}\|_* \|\frac{\varepsilon_m X_m'}{T}\|_2 + \|\Delta_\Phi\|_* \|\sum_{m=1}^M \frac{W_m^* \varepsilon_m X_m'}{T}\|_2 + \sum_{m=1}^M \|\Delta_{S_m}\|_1 \|\frac{\varepsilon_m X_m'}{T}\|_{\infty}\\
    & \le W \|\Delta_\Phi\|_* \| \sum_{m=1}^M \frac{W_m^\dag \varepsilon_m X_m'}{T}\|_2 + \sum_{m=1}^M \frac{\phi \ell}{\sqrt{r} p} \|\Delta_{W_m}\|_* \|\frac{\varepsilon_m X_m'}{T}\|_2 + \sum_{m=1}^M (\|\Delta_{S_m}^\mathcal{Q}\|_1 + \|\Delta_{S_m}^{\mathcal{Q}^\perp}\|_1) \|\frac{\varepsilon_m X_m'}{T}\|_{\infty}
\end{aligned}
$$
Then, combining the above two inequalities yields
$$
\begin{aligned}
    \sum_{m=1}^M \frac{1}{2T} & \|(\Delta_{W_m} \widehat{\Phi} + W_m^* \Delta_\Phi + \Delta_{S_m}) X_m\|_F^2\\
    & \le W \|\Delta_\Phi\|_* \|\sum_{m=1}^M \frac{W_m^\dag \varepsilon_m X_m'}{T}\|_2 + \frac{\phi \ell}{\sqrt{r p}} \sum_{m=1}^M \|\Delta_{W_m}\|_F \|\frac{\varepsilon_m X_m'}{T}\|_2\\
    & \quad + \sum_{m=1}^M (\|\Delta_{S_m}^{\mathcal{Q}_m}\|_1 + \|\Delta_{S_m}^{\mathcal{Q}^\perp_m}\|_1) \|\frac{\varepsilon_m X_m'}{T}\|_\infty + \sum_{m=1}^M \big( \eta \|\Delta_{S_m}^{\mathcal{Q}_m}\|_1 - \eta \|\Delta_{S_m}^{\mathcal{Q}^\perp_m}\|_1 + 2 \eta \|(S_m^*)^{\mathcal{Q}^\perp_m}\|_1 \big)\\
    & \le W \|\Delta_\Phi\|_* \|\sum_{m=1}^M \frac{W_m^\dag \varepsilon_m X_m'}{T}\|_2 + \frac{\phi \ell}{\sqrt{r p}} \sum_{m=1}^M \|\Delta_{W_m}\|_F \| \frac{\varepsilon_m X_m'}{T}\|_2\\
    & \quad + \sum_{m=1}^M \left[(\|\frac{\varepsilon_m X_m'}{T}\|_\infty + \eta) \|\Delta_{S_m}^{\mathcal{Q}_m}\|_1 + (\|\frac{\varepsilon_m X_m'}{T}\|_\infty - \eta) \|\Delta_{S_m}^{\mathcal{Q}^\perp_m}\|_1 \right]
\end{aligned}
$$
The last inequality follows from choosing $\mathcal{Q}_m$ as the support of the true sparse $S_m^*$. Under this choice, and by applying the restricted strong convexity condition stated in \autoref{ass:rsc}, we obtain
$$
\noindent
\begin{aligned}
    & \sum_{m=1}^M \frac{1}{2 T} \|(\Delta_{W_m} \widehat{\Phi} + W_m^* \Delta_\Phi + \Delta_{S_m}) X_m\|_F^2\\
    & \quad \ge \frac{\beta}{2}\sum_{m=1}^M \|\Delta_{W_m} \widehat{\Phi} + W_m^* \Delta_\Phi + \Delta_{S_m}\|_F^2\\
    & \quad \ge \frac{\beta_\Phi}{2} \|\Delta_\Phi\|_F^2 + \sum_{m=1}^M (\frac{\beta_W}{2} \|\Delta_{W_m}\|_F^2 + \frac{\beta}{2}\|\Delta_{S_m}\|_F^2 - \beta \|W_m^* \Delta_{W_m}\|_2 \|\widehat{\Phi}\|_\infty \|\Delta_\Phi\|_* \\
    & \quad \qquad \qquad - \beta W \|\Delta_\Phi\|_\infty \|\Delta_{S_m}\|_1 - 2 \beta \widehat{W} \|\widehat{\Phi}\|_\infty \|\Delta_{S_m}\|_1)\\
    & \quad \ge \frac{\beta_\Phi}{2} \|\Delta_\Phi\|_F^2 - \frac{2 \beta \phi W \widehat{W} M \ell}{\sqrt{r} p} \|\Delta_{\Phi}\|_* + \sum_{m=1}^M (\frac{\beta_W}{2} \|\Delta_{W_m}\|_F^2 + \frac{\beta}{2} \|\Delta_{S_m}\|_F^2 - \frac{4 \beta \phi \widehat{W} \ell}{\sqrt{r} p} \|\Delta_{S_m}\|_1)
\end{aligned}
$$

Hence, we can find $\zeta = \min\{ \frac{\beta_\Phi}{M \beta}, \frac{\beta_W}{\beta}, 1 \} \succsim \min \left\{ \frac{r p}{\ell^2}, \frac{\ell^2}{\hat{r} p}, 1 \right\}$, and $\eta \ge \frac{4 \beta \phi \widehat{W} \ell}{\sqrt{r} p} + \|\frac{\varepsilon_m X_m'}{T}\|_\infty$ for all $m$,
$$
\begin{aligned}
    & \frac{\beta \zeta}{2} \big( \|\Delta_\Phi\|_F^2 + \frac{1}{M} \sum_{m=1}^M (\|\Delta_{W_m}\|_F^2 + \|\Delta_{S_m}\|_F^2) \big)\\
    & \quad \le \frac{\beta_\Phi}{2 M} \|\Delta_\Phi\|_F^2 + \frac{1}{2 M} \sum_{m=1}^M (\beta_W \|\Delta_{W_m}\|_F^2 + \beta \|\Delta_{S_m}\|_F^2) \\
    & \quad \le \|\Delta_\Phi\|_* \big( W \|\sum_{m=1}^M \frac{W_m^\dag \varepsilon_m X_m'}{M T}\|_2 + \frac{2 \beta \phi W \widehat{W} \ell}{\sqrt{r} p} \big) + \frac{1}{M} \sum_{m=1}^M \frac{\phi \ell}{\sqrt{r p}} \|\Delta_{W_m}\|_F \| \frac{\varepsilon_m X_m'}{T}\|_2\\
    & \qquad \qquad + \frac{1}{M} \sum_{m=1}^M \Big(( \frac{4 \beta \phi \widehat{W} \ell}{\sqrt{r} p} + \|\frac{\varepsilon_m X_m'}{T}\|_\infty + \eta) \|\Delta_{S_m}^{\mathcal{Q}_m}\|_1\\
    & \qquad \qquad \qquad \qquad \quad + ( \frac{4 \beta \phi \widehat{W} \ell}{\sqrt{r} p} + \|\frac{\varepsilon_m X_m'}{T}\|_\infty - \eta) \|\Delta_{S_m}^{\mathcal{Q}^\perp_m}\|_1\Big)\\
    & \quad \le \frac{\phi \ell}{M \sqrt{r p}} \sum_{m=1}^M \|\Delta_{W_m}\|_F \| \frac{\varepsilon_m X_m'}{T}\|_2 + \frac{2 \eta \sqrt{s}}{M} \sum_{m=1}^M \|\Delta_{S_m}\|_F\\
    & \qquad \qquad + \|\Delta_\Phi\|_F \big( W \sqrt{2 \hat{r}} \|\sum_{m=1}^M \frac{W_m^\dag \varepsilon_m X_m'}{MT}\|_2 + \frac{2 \sqrt{2} \beta \phi W \widehat{W} \ell \sqrt{\hat{r}}}{\sqrt{r} p} \big)\\
    & \quad \le \sqrt{ \frac{1}{M} \sum_{m=1}^M \Big( \frac{\phi^2 \ell^2}{r p} \| \frac{\varepsilon_m X_m'}{T}\|_2^2 + 4 \eta^2 s + 4 W^2 \hat{r} \|\sum_{m=1}^M \frac{W_m^\dag \varepsilon_m X_m'}{MT}\|_2^2 + \frac{16 \beta^2 \phi^2 W^2 \widehat{W}^2 \ell^2 \hat{r}}{r p^2}\Big)}\\
    & \qquad \qquad \times \sqrt{ \|\Delta_\Phi\|_F^2 + \frac{1}{M} \sum_{m=1}^M (\|\Delta_{S_m}\|_F^2 + \|\Delta_{W_m}\|_F^2) }.
\end{aligned}
$$
Therefore,
\begin{multline*}
	\frac{\beta^2 \zeta^2}{4} \big( \|\Delta_\Phi\|_F^2 + \frac{1}{M} \sum_{m=1}^M (\|\Delta_{S_m}\|_F^2 + \|\Delta_{W_m}\|_F^2) \big)\\
	\le \frac{1}{M} \sum_{m=1}^M \Big( \frac{\phi^2 \ell^2}{r p} \| \frac{\varepsilon_m X_m'}{T}\|_2^2 + 4 \eta^2 s + \frac{16 \beta^2 \phi^2 W^2 \widehat{W}^2 \ell^2 \hat{r}}{r p^2} \Big) + 4 W^2 \hat{r} \|\sum_{m=1}^M \frac{W_m^\dag \varepsilon_m X_m'}{M T}\|_2^2.
\end{multline*}
Hence, when we select $\eta = \frac{4 \beta \phi \widehat{W} \ell}{\sqrt{r} p}  + \max_m \|\frac{\varepsilon_m X_m'}{T}\|_\infty$, we obtain
$$
\begin{aligned}
    & \|\Delta_\Phi\|_F^2 + \frac{1}{M} \sum_{m=1}^M (\|\Delta_{S_m}\|_F^2 + \|\Delta_{W_m}\|_F^2)\\
    & \quad \le \frac{4}{\zeta^2} \bigg( 16 \phi^2 \widehat{W}^2 \ell^2  \frac{W^2 \hat{r} + 8 s}{r p^2} + \frac{8 s}{\beta^2} \max_m \|\frac{\varepsilon_m X_m'}{T}\|_\infty^2\\
    & \qquad \qquad + \frac{\phi^2 \ell^2}{\beta^2 M r p} \sum_{m=1}^M \| \frac{\varepsilon_m X_m'}{T}\|_2^2 + \frac{4 W^2 \hat{r}}{\beta^2} \| \sum_{m=1}^M \frac{W_m^\dag \varepsilon_m X_m'}{MT}\|_2^2 \bigg)\\
    & \quad \precsim \frac{1}{\zeta^2} \left( \frac{\hat{r}^2}{\ell^2} + \frac{\hat{r} s}{r p} + \frac{s}{\beta^2} \max_m \|\frac{\varepsilon_m X_m'}{T}\|_\infty^2 + \frac{\ell^2}{\beta^2 M r p} \sum_{m=1}^M \| \frac{\varepsilon_m X_m'}{T}\|_2^2 + \frac{r \hat{r} p}{\beta^2 \ell^2} \| \sum_{m=1}^M \frac{W_m^\dag \varepsilon_m X_m'}{MT}\|_2^2 \right).
\end{aligned}
$$
Note that the last inequality uses the bounds of $W$ and $\widehat{W}$ from \autoref{claim:rates}.
\end{proof}

\begin{proof}[Proof of \autoref{prop:deviation}]
   The proofs for the first and last statements closely follow the arguments in \citet{basu2015,basu2019} and are therefore omitted. We focus here on proving the second statement, for which we begin with the following lemma.
   
    \begin{lem}\label{lem:bound}
        For $M$ arbitrary stationary centered Gaussian time series $\{H_m \in \mathbb{R}^{p \times T}\}_{m=1}^M$ that are mutually independent, and an arbitrary unit vector $v \in \mathbb{R}^p$ that $\|v\| = 1$, there exists constant $c > 0$, such that for any $k > 0$,
        $$
        P(|v' (\sum_{m=1}^M \frac{H_m H_m'}{MT} - \sum_{m=1}^M \frac{\Gamma_{H_m}(0)}{M}) v| > 2 \pi k \max_m \Psi(h_{H_m})) \le 2\exp(- c M T \min\{k^2, k\}).
        $$
    \end{lem}

    The proof of the lemma also follows based on arguments in \cite{basu2015}. Specifically, note that $(v' H_1, \cdots, v' H_M)' \sim N(0, Q_H)$ with
    $$
    Q_H = \blkdiag(\Upsilon_1, \dots, \Upsilon_M), \qquad \Upsilon_m[r,s] = v' \Gamma_{H_m}(r-s) v,
    $$
    and $\|Q_H\| = \max_m \|\Upsilon_m\|$ due to its block-diagonal structure.

    Next, by denoting $\bar\varepsilon_m = W_m^\dag \varepsilon_m$, we get
    $$
    \begin{aligned}
        \sum_{m=1}^M \frac{2 v' \bar \varepsilon_m X_m' v}{MT} & = [\sum_{m=1}^M \frac{(X_m' v + \bar\varepsilon_m' v)' (X_m' v + \bar\varepsilon_m' v)}{MT} - \sum_{m=1}^M \frac{v' (\Gamma_{X_m}(0) + W_m^\dag \Sigma_m W_m^\dag) v}{M}]\\
        & \qquad - [\sum_{m=1}^M \frac{v' X_m X_m' v}{MT} - \sum_{m=1}^M \frac{v' \Gamma_{X_m}(0) v}{M}] - [\sum_{m=1}^M \frac{v' \bar\varepsilon_m \bar\varepsilon_m' v}{MT} - \sum_{m=1}^M \frac{v' W_m^\dag \Sigma_m W_m^\dag v}{M}]
    \end{aligned}
    $$

    Applying \autoref{lem:bound} to the collections $\{X_m\}_{m=1}^M$, $\{\bar\varepsilon_m\}_{m=1}^M$ and $\{X_m + \bar\varepsilon_m\}_{m=1}^M$, and using the fact that
    $$
    \Psi(h_{X_m}) \le \frac{\lambda_1(\Sigma_m)}{\tau_{\min}(\mathcal{A}_m)}, \quad \Psi(h_{\bar\varepsilon_m}) \le \lambda_1(\Sigma_m), \quad \Psi(h_{X_m + \bar\varepsilon_m}) \le \frac{\lambda_1(\Sigma_m) \tau_{\max}(\mathcal{A}_m)}{\tau_{\min}(\mathcal{A}_m)},
    $$
    the conclusion follows by arguments analogous to those in \citet{basu2019}.
\end{proof}

\begin{proof}[Proof of \autoref{cor:error}]
With similar arguments to those used in the proof of \autoref{prop:upper}, define $\Delta_m = \widehat{W}_m \widehat{\Phi} - W_m^* \Phi^*$. Then, we have
\begin{multline*}
	\frac{1}{2T} \sum_{m=1}^M \| (\Delta_m + \Delta_{S_m}) X_m\|_F^2 \\
	\le \sum_{m=1}^M \|\Delta_m\|_* \|\frac{\varepsilon_m X_m'}{T}\|_2 + \sum_{m=1}^M \left[(\|\frac{\varepsilon_m X_m'}{T}\|_\infty + \eta) \|\Delta_{S_m}^{\mathcal{Q}_m}\|_1 + (\|\frac{\varepsilon_m X_m'}{T}\|_\infty - \eta) \|\Delta_{S_m}^{\mathcal{Q}^\perp_m}\|_1 \right].
\end{multline*}
and with the assumption implying that $\|\Delta_m\|_\infty \le \frac{2 \phi \widehat{W}\ell}{\sqrt{r} p}$,
$$
\frac{1}{2T} \sum_{m=1}^M \| (\Delta_m + \Delta_{S_m}) X_m\|_F^2 \ge \frac{\beta}{2} \sum_{m=1}^M (\|\Delta_m\|_F^2 + \|\Delta_{S_m}\|_F^2 - \frac{4 \phi \widehat{W} \ell}{\sqrt{r} p} \|\Delta_{S_m}\|_1).
$$
Hence, selecting $\eta$ as $\eta \ge \frac{2 \beta \phi \widehat{W} \ell}{\sqrt{r} p} + \max_m \|\frac{\varepsilon_m X_m'}{T}\|_\infty$, we get
$$
\frac{1}{M} \sum_{m=1}^M (\|\Delta_m\|_F^2 + \|\Delta_{S_m}\|_F^2) \le \frac{2}{\beta M} \sum_{m=1}^M (\|\Delta_m\|_* \|\frac{\varepsilon_m X_m'}{T}\|_2 + 2 \eta \|\Delta_{S_m}^{\mathcal{Q}_m}\|_1),
$$
and therefore
$$
\begin{aligned}
    \frac{1}{M} \sum_{m=1}^M (\|\Delta_m\|_F^2 + \|\Delta_{S_m}\|_F^2) & \le \frac{4}{\beta^2 M} \sum_{m=1}^M (2 \hat{r} \|\frac{\varepsilon_m X_m'}{T}\|_2^2 + 4 \eta^2 s)\\
    & \le \frac{4}{\beta^2 M} \sum_{m=1}^M (\frac{4 c_1^2 \xi^2 \hat{r} p}{T} + \frac{8 c_1^2 \xi^2 s \log(p)}{T}) + \frac{128 \phi^2 \widehat{W}^2 \ell^2 s}{r p^2}\\
    & \precsim \xi^2 \cdot \max_m \frac{\tau_{\max}^2(\mathcal{A}_m)}{\lambda_p^2(\Sigma_m)} \cdot \frac{s\log(p) + \hat{r} p}{T} + \frac{\iota s}{p}.
\end{aligned}
$$
\end{proof}

\section{Implementation Details}\label{sec:simu_sup}

\subsection{Data Generating Process (DGP)}\label{subsec:dgp}

Next, we introduce the data generating process (DGP) corresponding to the LSPVAR model in \eqref{eq:pVAR-model}. The parameters $(M, p, r, s)$ are specified upfront, and the time series length $T$ is chosen accordingly, guided by the consistency results discussed in \autoref{sec:cons}.

For the simulation in \autoref{ex:mixture}, we generate only 6 distinct diagonal matrices in $\mathcal{W}$, each representing one of the 6 clusters. In contrast, the 20 sparse matrices $\mathcal{S}$ are sampled independently following the procedure described in \autoref{algo:DGP}.

\begin{algorithm}[htbp]
    \KwIn{Number of models $M$, time series dimension $p$, time series length $T$, rank $r$ of matrix $\Phi$, expected percentage of non-zero elements $\frac{s}{p^2}$ in sparse matrices $\mathcal{S}$, time series length $T$, number of replicates $N$ per setup.}
    \KwOut{Model parameters $(\Phi, \mathcal{W}, \mathcal{S}, \{A_m, \Sigma_m\}_{m=1}^M)$, and the replicates of Panel time series data $\{\mathcal{X}_n(M, T, p, r, s)\}_{n=1}^N$.}
    Randomly generate two singular vector matrices $U, V \in \R^{p \times r}$ ($r \le p$), sample the entries of the diagonal matrix $D$ from an $r$-dimensional Dirichlet distribution, and generate $\Phi = U D V'$.\\
    \For{$m = 1, \dots, M$}{
    Sample $s_m$ from a Poisson distribution with mean $s$, generate a sparse matrix $S_m$ with support $\|S_m\|_0 = s_m$, and elements from centered normal distribution with standard deviation $\sqrt{p} \|\Phi\|_\infty$.\\
    Compute the eigenvalues of $\Phi + S_m$, and sample the diagonal entries of $W_m$ from the uniform distribution on $[\frac{1}{2} |\lambda_1(\Phi + S_m)|^{-1}, |\lambda_1(\Phi + S_m)|^{-1}]$.\\
    Update the sparse matrix $S_m \mapsto W_m S_m$, and define $A_m = W_m \Phi + S_m$.\\
    Randomly sample $\Sigma_m = \sigma_m^2 I_p$ as the innovation covariance matrix, where $\sigma_m^2$ are sampled from an inverse Gamma distribution.
    }
    \For{$n = 1, \dots, N$}{
    Sample the time series of length $T$ as $\mathcal{X}_n(M, p, r, s, T) = \{\{X_t^m\}_{t=1}^T\}_{m=1}^M$ according to our PVAR model setup specified by \eqref{eq:pVAR-model} and parameters $\{A_m, \Sigma_m\}_{m=1}^M$.
    }
    \caption{DGP for $\{\mathcal{X}_n(M, p, r, s, T)\}_{n=1}^N$} \label{algo:DGP}
\end{algorithm}

\subsection{Initialization}

We also provide an initialization algorithm intended to produce estimates close to the true parameters. \autoref{algo:init} is motivated by the observation that, when temporarily ignoring the sparse components, corresponding rows from individual fits tend to align consistently in the same direction, differing primarily by a scaling factor.

\begin{algorithm}[htbp]
    \KwIn{Time series data $\{(X_m, Y_m) \in \mathbb{R}^{p \times T} \times \mathbb{R}^{p \times T}\}_{m=1}^M$, maximum rank $\hat r$, fixed nuclear norm $\ell$.}
    \KwOut{Initialization $(\Phi^{(0)}, \Phi_c^{(0)})$.}
    For each entity with data pair $(X_m, Y_m)$, individually fit the VAR coefficient matrix $\Phi_m \in \mathbb{R}^{p \times p}$ .\\
    \For{$i = 1, \dots, p$}{
    Take the $i$-th rows of $\Phi_m$ ($m = 1, \dots, M$) and stack them as $\varPhi_i \in \mathbb{R}^{M \times p}$.\\
    Find the top right singular vector of $\varPhi_i$ and set it as the $i$-th row of $\Phi^{(0)}$.\\
    }
    Update $\Phi^{(0)} \mapsto \frac{\ell}{\|\Phi^{(0)}\|_*} \Phi^{(0)}$.\\
    Calculate $\Phi_c^{(0)}$ by projecting $\Phi^{(0)}$ onto $\mathbb{B}_p \cap \mathbb{L}_p(\hat r, \ell)$.\\
    \caption{Initialization for $\Phi^{(0)} = \Phi_c^{(0)}$} \label{algo:init}
\end{algorithm}

\subsection{Supplementary Figures and Tables}\label{subsec:fig_tab}

In this section, we supplement the simulation and neuroscience application results with additional visualizations, including figures and tables.

We begin with \autoref{tab:rank}, which summarizes the key statistics from our simulation study in \autoref{subsec:rank}, focusing on the effects of the input rank $\hat{r}$ and the step size $\rho$.

\begin{table}[htbp]
    \centering
    {\scriptsize
\begin{tblr}{
  column{3} = {r},
  column{4} = {r},
  column{5} = {r},
  column{6} = {r},
  column{7} = {r},
  column{8} = {r},
  column{9} = {r},
  cell{2}{1} = {r=3}{c},
  cell{5}{1} = {r=3}{c},
  cell{8}{1} = {r=3}{c},
  cell{11}{1} = {r=3}{c},
  cell{14}{1} = {r=3}{c},
  cell{17}{1} = {r=3}{c},
  vlines,
  hline{1-2,5,8,11,14,17,20} = {-}{},
  hline{3-4,6-7,9-10,12-13,15-16,18-19} = {2-9}{},
}
 & \diagbox{$\rho$}{$\hat{r}$} & 3 & 5 & 10 & 15 & 20 & 30 & 40\\
{Average Rescaled\\Shifted BIC\\$(\times 10^5)$} & $\frac{M}{20}$ & 0.200 & 0.212 & 0.228 & 0.260 & 0.281 & 0.315 & 0.322\\
 & $\frac{M}{5}$ & 0.200 & 0.212 & 0.228 & 0.260 & 0.279 & 0.315 & 0.322\\
 & $M$ & 0.200 & 0.211 & 0.229 & 0.260 & 0.279 & 0.315 & 0.322\\
{Average\\Relative\\Error of $A_m$} & $\frac{M}{20}$ & 0.263 & 0.281 & 0.285 & 0.290 & 0.299 & 0.322 & 0.322\\
 & $\frac{M}{5}$ & 0.263 & 0.281 & 0.285 & 0.289 & 0.298 & 0.322 & 0.322\\
 & $M$ & 0.262 & 0.281 & 0.285 & 0.289 & 0.298 & 0.321 & 0.322\\
{Average\\Relative\\Error of $S_m$} & $\frac{M}{20}$ & 0.245 & 0.285 & 0.350 & 0.406 & 0.451 & 0.453 & 0.455\\
 & $\frac{M}{5}$ & 0.245 & 0.285 & 0.350 & 0.406 & 0.449 & 0.453 & 0.455\\
 & $M$ & 0.245 & 0.285 & 0.350 & 0.405 & 0.449 & 0.452 & 0.455\\
{Average\\Sparsity\\Accuracy} & $\frac{M}{20}$ & 0.970 & 0.954 & 0.955 & 0.971 & 0.969 & 0.951 & 0.951\\
 & $\frac{M}{5}$ & 0.970 & 0.954 & 0.955 & 0.971 & 0.970 & 0.951 & 0.951\\
 & $M$ & 0.970 & 0.954 & 0.955 & 0.971 & 0.970 & 0.951 & 0.951\\
{Average\\Sparsity\\Sensitivity} & $\frac{M}{20}$ & 0.745 & 0.767 & 0.745 & 0.687 & 0.670 & 0.699 & 0.702\\
 & $\frac{M}{5}$ & 0.745 & 0.767 & 0.745 & 0.684 & 0.671 & 0.699 & 0.702\\
 & $M$ & 0.745 & 0.767 & 0.745 & 0.685 & 0.671 & 0.700 & 0.702\\
{Average\\Sparsity\\Specificity} & $\frac{M}{20}$ & 0.976 & 0.959 & 0.960 & 0.979 & 0.977 & 0.958 & 0.958\\
 & $\frac{M}{5}$ & 0.976 & 0.959 & 0.960 & 0.979 & 0.977 & 0.958 & 0.958\\
 & $M$ & 0.976 & 0.959 & 0.960 & 0.979 & 0.977 & 0.958 & 0.958
\end{tblr}}

    \caption{Supplement to \autoref{subsec:rank}. Summary metrics from simulations examining the choice of input rank $\hat{r}$ and ADMM step size $\rho$ }
    \label{tab:rank}
\end{table}

\autoref{fig:hist} depicts the histograms of the estimated $\|\widehat{W}_m \|^2$ under the setting specified in \autoref{ex:mixture}.
They illustrate that the entities within the singular low-rank cluster have significantly smaller magnitude norms $\|\widehat{W}_m \|_F^2$. This observation closely aligns with the model’s ground truth, which states that their Frobenius norms should be zero.

\begin{figure}[htbp]
    \centering
    \includegraphics[width=\linewidth]{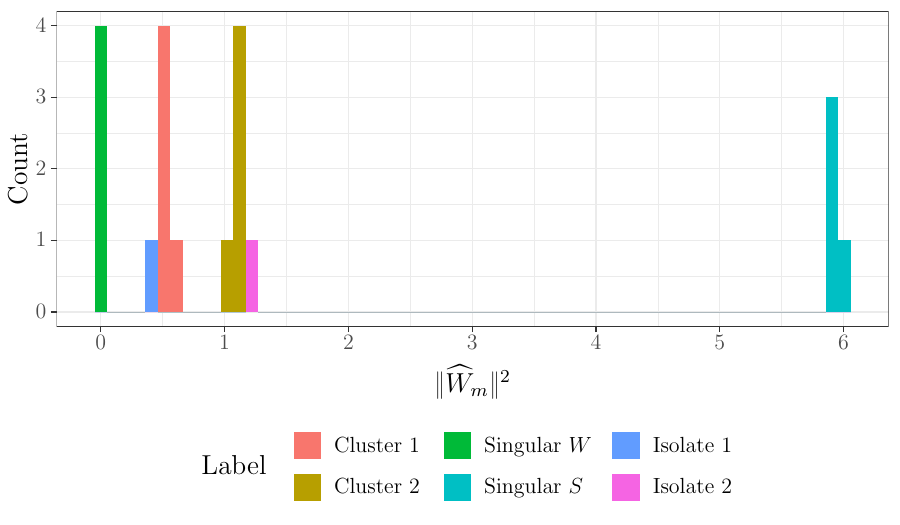}
    \caption{Supplement to \autoref{subsec:cluster}. The histogram above displays the frequencies of the estimated Frobenius norms $|\widehat{W}_m|_F^2$ from a single replicate. It demonstrates that the clustering structure can be consistently recovered in accordance with our simulation setup.}
    \label{fig:hist}
\end{figure}

Finally, for the EEG application presented in \autoref{sec:app}, we display boxplots of all rescaling effects $W$, estimated from both the raw data and the alpha band-filtered data. Similar conclusions can be drawn from \Cref{fig:W_rd,fig:W_ab}: the estimates based on the alpha band data exhibit greater variability across entities (students). In addition, aside from some potential outliers, the effects under the EO condition condition appear more stable, as indicated by shorter box lengths, and less pronounced in magnitude—especially evident in the alpha band estimates.

\begin{figure}[htbp]
    \centering
    \includegraphics[width=\linewidth]{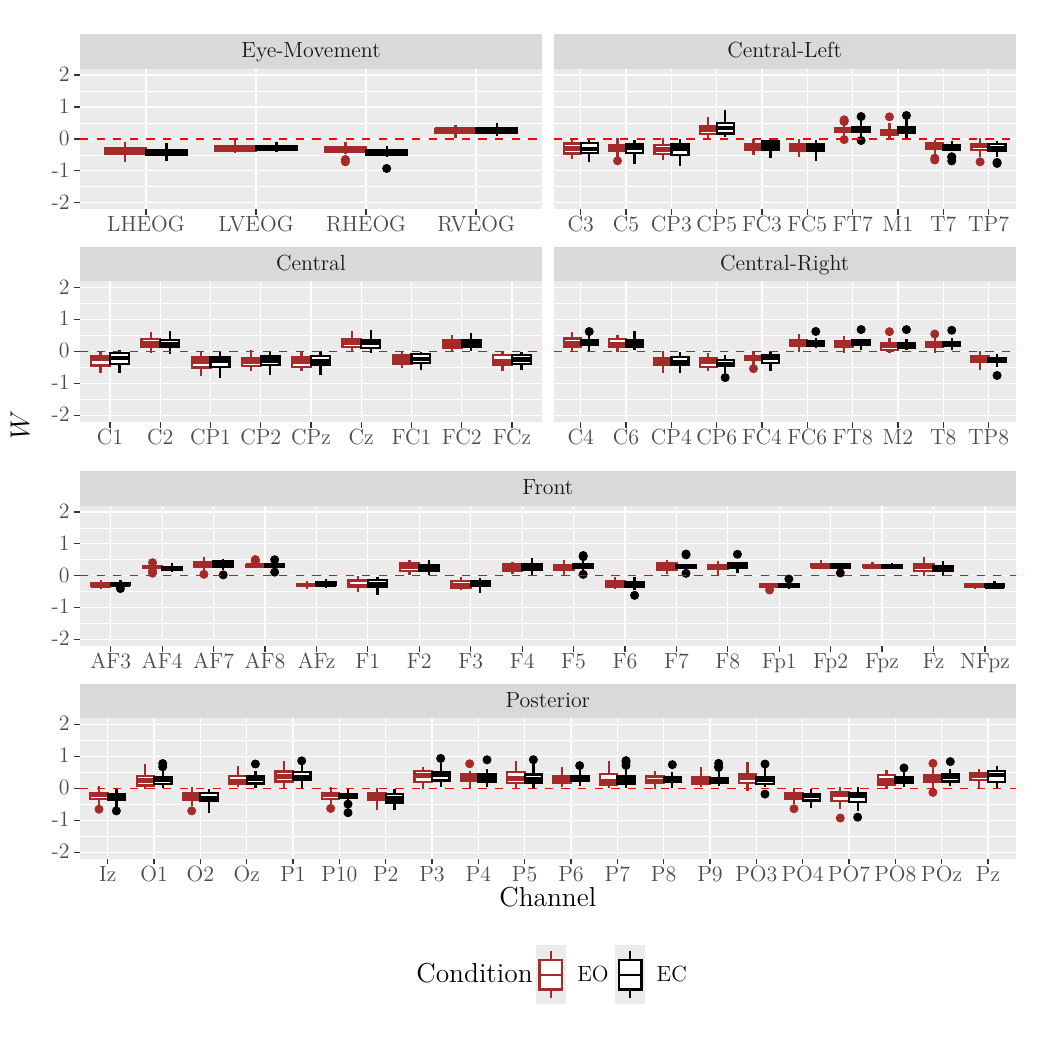}
    \caption{Boxplots of the rescaling effects for all channels estimated from the raw data.}
    \label{fig:W_rd}
\end{figure}

\begin{figure}[htbp]
    \centering
    \includegraphics[width=\linewidth]{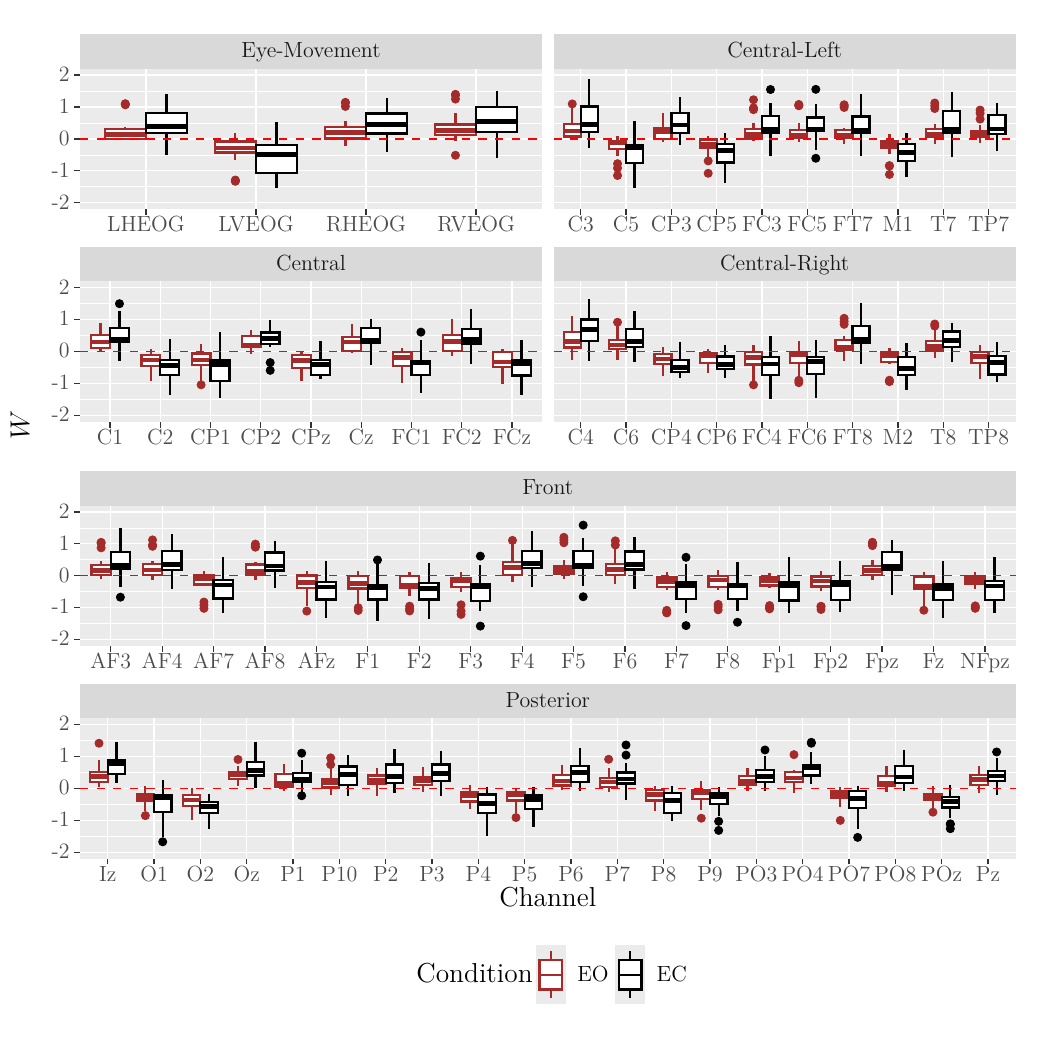}
    \caption{Boxplots of the rescaling effects for all channels estimated from the alpha bands.}
    \label{fig:W_ab}
\end{figure}

We also include a brief visualization of the estimated low-rank component $\Phi$ in \autoref{fig:Phi_hm}. We observe that the outward connections from the posterior channels are relatively stronger on the human scalp, and these connections become even more pronounced after filtering out the underlying noise signals.

 \begin{figure}[htbp]
    \centering
    \includegraphics[width=\textwidth]{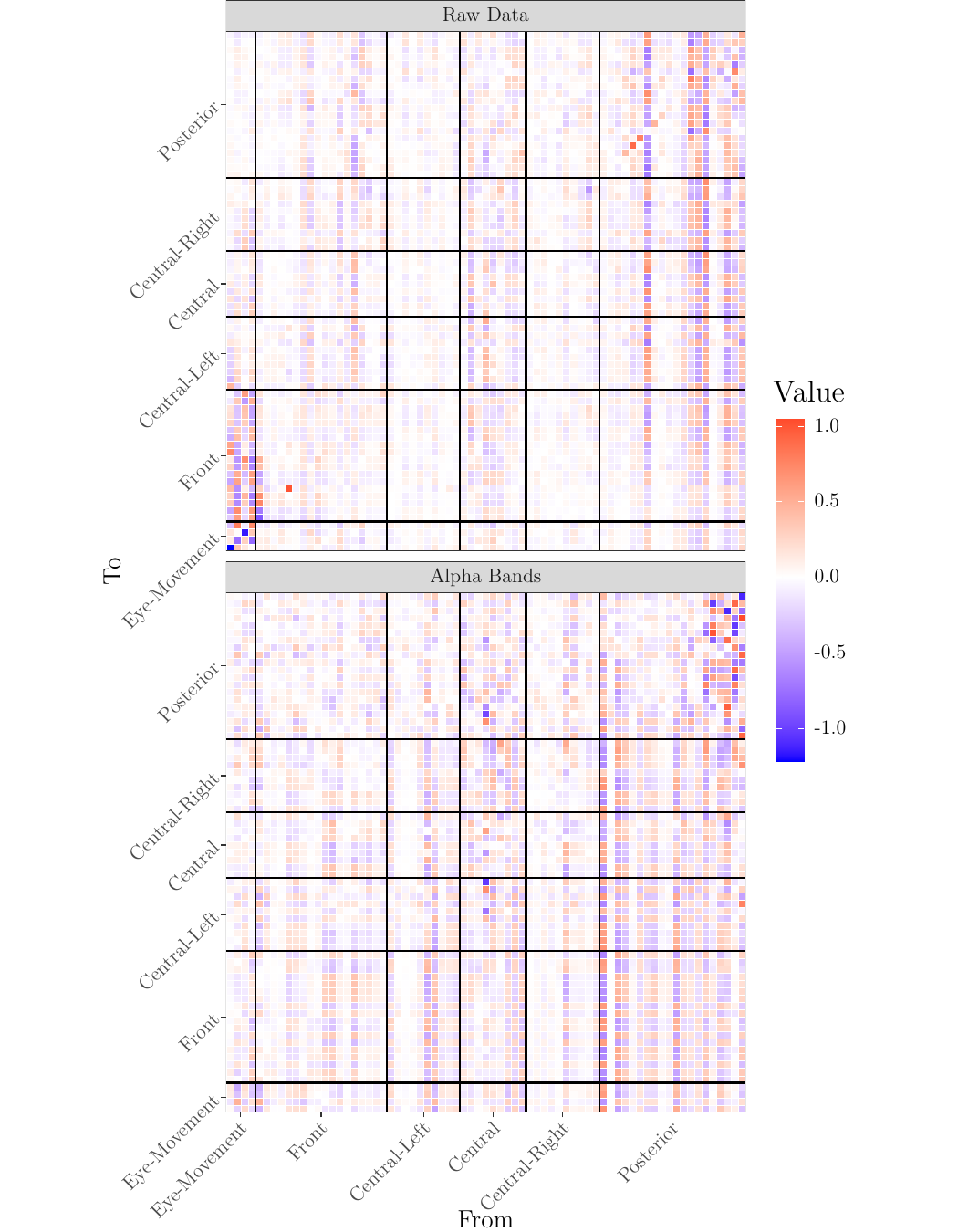}
    \vspace{-.5em}
    \caption{The heatmap for the estimated $\Phi$, with blocks illustrating the dynamics at the cluster level.}
    \label{fig:Phi_hm}
 \end{figure}

\end{appendix}

\bibliographystyle{imsart-nameyear} 
\bibliography{main}


\end{document}